\documentclass[11pt,draftclsnofoot,onecolumn,journal,letterpaper]{IEEEtran}

\usepackage{cite}
\usepackage{amsmath,amssymb,amsfonts}
\usepackage{algorithmic}
\usepackage{graphicx}
\usepackage{textcomp}
\usepackage{xcolor}
\def\BibTeX{{\rm B\kern-.05em{\sc i\kern-.025em b}\kern-.08em
    T\kern-.1667em\lower.7ex\hbox{E}\kern-.125emX}}

\usepackage{graphicx}
\usepackage{epstopdf}
\usepackage{amsmath,amssymb,amsthm,mathrsfs,amsfonts,dsfont}
\usepackage{epsfig}
\usepackage[ruled,linesnumbered]{algorithm2e}
\usepackage{color}
\usepackage{subfigure}
\usepackage{cite}
\usepackage{diagbox}
\usepackage{multirow}
\usepackage{float}
\usepackage{stfloats}
\usepackage{enumitem}
\usepackage{cases}
\usepackage{setspace}
\usepackage{adjustbox,lipsum}
\usepackage{verbatim}
\graphicspath{{fig/}}

\newcommand{\mypara}[1]{{\smallskip \noindent \bf #1}\hspace{0.1in}}
\newcommand{\ssf}[1]{\textrm{$\sf{#1}$}{}}

\newtheorem{theorem}{Theorem}

\newtheorem{lemma}{Lemma}
\newtheorem{assumption}{Assumption}

\newtheorem{remark}{Remark}

\DeclareMathOperator*{\argmin}{arg\,min}

\newcommand{\vect}[1]{\mathbf{#1}}
\newcommand{\avgvect}[1]{\mathbf{\overline{#1}}}
\newcommand{\expt}{\mathbb{E}}
\newcommand{\norm}[1]{\left \| #1 \right \|}
\newcommand{\squab}[1]{\left [ #1 \right ]}
\newcommand{\dotp}[2]{\left \langle #1, #2 \right \rangle}
\usepackage{xcolor}

\ifodd 0
\newcommand{\congr}[1]{{\color{magenta}#1}}
\else
\newcommand{\congr}[1]{#1}
\fi

\ifodd 0
\newcommand{\congc}[1]{{\color{red}(Cong: #1)}}
\else
\newcommand{\congc}[1]{}
\fi

\ifodd 0
\newcommand{\zixiangb}[1]{{\color{blue}#1}}
\else
\newcommand{\zixiangb}[1]{#1}
\fi

\ifodd 0
\newcommand{\zixiang}[1]{{\color{blue}(Zixiang: #1)}}
\else
\newcommand{\zixiang}[1]{}
\fi

\ifodd 1
\newcommand{\zixiangTCCN}[1]{{\color{black}#1}}
\else
\newcommand{\zixiangTCCN}[1]{}
\fi

\begin{document}

\title{Federated Learning over Noisy Channels: Convergence Analysis and Design Examples}

\author{Xizixiang~Wei \qquad Cong~Shen%
\thanks{A preliminary version of this work was presented at the 2021 IEEE International Conference on Communications \cite{Wei2021icc}.}
\thanks{The authors are with the Charles L. Brown Department of Electrical and Computer Engineering, University of Virginia, Charlottesville, VA 22904, USA. E-mail: \texttt{\{xw8cw,cong\}@virginia.edu}.}
\thanks{The work was partially supported by the US National Science Foundation (NSF) under Grant ECCS-2033671.}
}

\maketitle

\begin{abstract}
Does Federated Learning (FL) work when \emph{both} uplink and downlink communications have errors? How much communication noise can FL handle and what is its impact to the learning performance? This work is devoted to answering these practically important questions by explicitly incorporating both uplink and downlink noisy channels in the FL pipeline. We present several novel convergence analyses of FL over simultaneous uplink and downlink noisy communication channels, which encompass full and partial clients participation, direct model and model differential transmissions, and non-independent and identically distributed (IID) local datasets. These analyses characterize the sufficient conditions for FL over noisy channels to have the same convergence behavior as the ideal case of no communication error. More specifically, in order to maintain the $\mathcal{O}({1}/{T})$ convergence rate of \textsc{FedAvg} with perfect communications, the uplink and downlink signal-to-noise ratio (SNR) for direct model transmissions should be controlled such that they scale as $\mathcal{O}(t^2)$ where $t$ is the index of communication rounds, but can stay $\mathcal{O}(1)$ (i.e., constant) for model differential transmissions. The key insight of these theoretical results is a ``flying under the radar'' principle -- stochastic gradient descent (SGD) is an inherent noisy process and uplink/downlink communication noises can be tolerated as long as they do not dominate the time-varying SGD noise. We exemplify these theoretical findings with two widely adopted communication techniques -- transmit power control and receive diversity combining -- and further validate their performance advantages over the standard methods via numerical experiments using several real-world FL tasks.
\end{abstract}


\section{Introduction}

Federated learning (FL) \cite{mcmahan2017fl,konecny2016fl} is an emerging distributed machine learning paradigm that has many attractive properties which can address new challenges in machine learning (ML). In particular, FL is motivated by the growing trend that massive amount of the real-world data are \textit{exogenously} generated at the edge devices \zixiangTCCN{and is considered as one of the potential key applications in 6th generation (6G) cellular communication systems \cite{yang2021federated}}. 

Communication efficiency has been at the front and center of FL ever since its inception \cite{konecny2016fl,mcmahan2017fl}, and it is widely regarded as one of its primary bottlenecks \cite{bonawitz2019towards,kairouz2019advances,li2020federated}. \zixiangTCCN{Communication schemes for FL can be divided into two categories: digital communication and analog communication. Digital communication for FL is usually considered to incur a heavy burden for wireless networks, as it allocates different communication resources to the ML model parameters of each client. Analog communication reduces the communication overhead by allowing different clients to transmit FL models using shared resources.} Early research has largely focused on either reducing the number of communication rounds \cite{mcmahan2017fl,li2018federated}, or decreasing the size of the payload for transmission \cite{Zheng2020jsac,reisizadeh2019fedpaq,du2020high}. \zixiangb{However,} in most FL literature that deals with communication efficiency, it is often assumed that a perfect communication ``tunnel'' has been established, and the task of improving communication efficiency largely resides on the ML design that trades off computation and communication. More recent research starts to close this gap by focusing on the system design, particularly for wireless FL; see Section~\ref{sec:related} for an overview. Nevertheless, the focus has been on bandwidth allocation, device selection, or either uplink or downlink (but not both) cellular system designs. 

While the early studies provide a glimpse of the potential of optimizing the communication design for FL, the important and more practical issue of noisy communications for {both uplink} (clients send local models to the parameter server) {and downlink} (server sends global model to clients) has not been well investigated. Analytically speaking, a joint consideration of both noisy uplink and downlink complicates the convergence analysis because of \emph{noise propagation} in both directions of every communication round. Furthermore, all of these noisy uplink and downlink communications collectively determine the final learning performance, which requires a holistic design and analysis.

The goal of this paper is two-fold: we want to first understand the impact of \emph{communication-induced noise}, in both upload (uplink) and download (downlink) phases of FL, on the ML model convergence and accuracy performance, and then design communication algorithms to control the signal-to-noise ratio (SNR) to improve FL performance under a total resource budget. 
We focus on \emph{analog} communications for model updates \cite{amiri2020federated,zhu2019broadband,yang2020federated} and investigate SNR control in both uplink and downlink, which is especially crucial when the underlying ML method is stochastic gradient descent (SGD) as considered in this work, because SGD is much more sensitive to noise than the (full) gradient descent \cite{kairouz2019advances,stich2018local}. 
{Our treatment is novel because all prior works either study uplink-only \cite{amiri2020federated,zhu2019broadband,xia2020fast,sery2020over,guo2021iot} or downlink-only \cite{amiri2020convergence} noisy communications, but not both. }  
We present novel convergence analyses of the standard \zixiangTCCN{Federated Averaging (\textsc{FedAvg})} scheme under non-IID datasets, full or partial clients participation, direct model or model differential transmissions, and simultaneous noisy downlink and uplink analog communications.  
These analyses are based on very general receive noise assumptions, and hence are broadly applicable to a variety of communication systems. 
The key insight of these theoretical results is a ``flying under the radar'' principle: SGD is inherently a noisy process, and as long as uplink/downlink channel noises do not dominate the SGD noise during model training (which is controlled by the time-varying learning rate), the scaling of convergence is not affected. 
This general principle is exemplified with two widely adopted communication techniques -- \textit{transmit power control} and \textit{receive diversity combining} -- by controlling the resulting post-processing SNR to satisfy the theoretical analyses {under a fixed total budget constraint}.   Comprehensive numerical evaluations on three widely adopted ML tasks with increasing difficulties ({MNIST}, {CIFAR-10} and {Shakespeare}) are carried out using these techniques. We carry out a series of experiments to demonstrate that the fine-tuned transmit power control and receive diversity combining that are guided by the theoretical analyses can significantly outperform the equal-SNR-over-time baseline, and in fact can approach the ideal noise-free communication performance in many of the experiment settings. 

\congr{
To summarize, the main contributions of this work include the following.
\begin{itemize}[leftmargin=*]\itemsep=0pt
\item We present novel convergence analyses for FL with {simultaneous uplink and downlink noisy analog communications}, with full vs. partial clients participation, direct model vs. model differential, and non-IID local datasets. To the best of the authors' knowledge, this is the first time FL convergence analysis is carried out when {both upload and download phases are over noisy communication channels}, which introduce significant challenges because of the {noise propagation} in both directions. 
\item We establish \textit{SNR scaling laws}. In particular, we prove that in order to maintain the well established $\mathcal{O}({1}/{T})$\footnote{\zixiangTCCN{Notation $f=\mathcal{O}(g)$ denotes $f$ is of order at most of $g$.}} convergence rate of \textsc{FedAvg} with noise-free communications, $\mathcal{O}(t^2)$ SNR scaling is needed for direct model and $\mathcal{O}(1)$ (i.e., constant) for model differential. This \textit{$t^2$-vs-$1$ scaling law comparison} under the same communication environment is novel. 
\item We enhance the widely adopted transmit power control and receive diversity combining algorithms to better serve FL over noisy channels, and validate their performance advantages over the state of the art methods under the same total resource budget via extensive numerical experiments.
\end{itemize}

}

The remainder of this paper is organized as follows. Related works are surveyed in Section~\ref{sec:related}. The system model that captures the noisy channels  in both uplink and downlink of FL is described in Section \ref{sec:model}. Theoretical analyses are presented in Section \ref{sec:convergence} for three different FL configurations. These results inspire novel communication designs of transmit power control and receive diversity combining that are presented in Section~\ref{sec:design}. Experimental results are given in Section \ref{sec:experiment}, followed by the conclusions in Section \ref{sec:conclusion}. All technical proofs are given in the Appendices.

\section{Related Works}
\label{sec:related}


\mypara{Improve FL communication efficiency.} The original \textsc{FedAvg} reduces the communication overhead by only periodically averaging the local models. Theoretical understanding of the communication-computation tradeoff has been actively pursued and, depending on the underlying assumptions (\zixiangTCCN{e.g.,} IID or non-IID local datasets, convex or non-convex loss functions, GD or SGD), rigorous analyses of the convergence behavior have been carried out \cite{stich2018local,wang2018cooperative,li2019convergence}.  For the approach of reducing the size of messages, general discussions on sparsification, subsampling, and quantization are given in \cite{konecny2016fl}. There are also recent efforts in developing quantization and source coding to reduce the communication cost  \cite{zhu2020one,reisizadeh2019fedpaq,amiri2020federated,amiri2020federated2,Zheng2020jsac,chen2021twc,du2020high}. Nevertheless, they mostly do not consider the communication channel noise. 

\mypara{Communication design for FL. } Recent years have also seen increased effort in the communication algorithm and system design for FL. Trade-off between local model update and global model aggregation is studied in \cite{mo2020energy} to optimize the transmission power/rate and training time. Various radio resource allocation and client selection policies \cite{zeng2019energy,shi2019device,yang2019energy,yang2019scheduling,chen2020convergence,xu2020client} have been proposed to minimize the learning loss or the training time. Joint communication and computation is investigated \cite{yang2020federated,zhu2019broadband,zhu2020one,chen2019joint}. In particular, the analog aggregation design \cite{zhu2019broadband,amiri2020sp,zhu2020one} serves as one of our design examples in Section~\ref{sec:design}.

\mypara{FL with imperfect/noisy communications.} 
Existing literature is dominated by uplink-only noisy communications \cite{frey2020over,jiang2020iccc,yang2020federated,amiri2020federated,zhu2019broadband,xia2020fast,sery2020over,guo2021iot}. There is very limited study on downlink-only noisy communications for FL; \cite{amiri2020convergence} proposes and analyzes downlink digital and analog transmissions while assuming an error-free uplink. On the other hand, existing literature that consider both upload and download imperfect communications focus only on how to modify the ML model training method. In particular, \cite{Ang2020tc} changes the loss function of FL to accommodate the communication error. \cite{tang2019doublesqueeze,yu2018double,chen2021scalecom} propose to compress the gradients in order to tolerate both uplink and downlink bandwidth bottlenecks. Their methods are either error compensation, quantization or leveraging sparsity. None of these considers improving the communication design.

\section{System Model for Learning and Noisy Communication}
\label{sec:model}

We first introduce the FL problem formulation, and then describe the FL pipeline where both local model upload (uplink) and global model download (downlink) take place over noisy channels.



\subsection{FL Problem Formulation}
\label{sec:modelFL}
The federated learning problem setting studied in this paper mostly follows the standard model in the original paper \cite{mcmahan2017fl}. In particular, we consider a FL system with one central parameter server (e.g., base station) and a set of at most $N$ clients (e.g., mobile devices). Client $k \in [N] \triangleq \{1, 2, \cdots, N\}$ stores a local dataset $\mathcal{D}_k = \{\vect{z}_i\}_{i=1}^{D_k}$, with its size denoted by $D_k$, that never leaves the client. Datasets across clients are assumed to be non-IID and disjoint. The maximum data size when all clients participate in FL is $D_{\text{tot}} = \sum_{k=1}^N D_k$. 
Each data sample $\vect{z}$ is given as an input-output pair $\{\vect{x}, y\}$. 
The loss function $f(\vect{w}, \vect{z})$ measures how well a ML model with parameter $\vect{w} \in \mathbb{R}^d$ fits a single data sample $\vect{z}$. Without loss of generality, we assume that $\vect{w}$ has zero-mean and unit-variance elements\footnote{The parameter normalization and de-normalization procedure in wireless FL can be found in the Appendix in \cite{zhu2019broadband}. We further note that \emph{weight normalization} is widely adopted in training deep neural networks \cite{salimans2016nips}.}, i.e., $\expt||w_i||^2 = 1$, $\forall i \in [d]$. For the $k$-th client, its local loss function $F_k(\cdot)$ is defined by
   $ F_k(\vect{w}) \triangleq \frac{1}{D_k} \sum_{\vect{z}\in\mathcal{D}_k} f(\vect{w}, \vect{z})$,
and we further use $\nabla F_k (\vect{w}, \xi)$ to denote the SGD operation with model $\vect{w}$ and data sample $\xi$ at client $k$. 
The goal of FL is to learn a {global} machine learning (ML) model at the parameter server based on the distributed {local} datasets at the $N$ clients, by coordinating and aggregating the training processes at individual clients without allowing the server to access the raw data. Specifically, the global optimization objective over all $N$ clients is given by
$F(\vect{w}) \triangleq \sum_{k=1}^N \frac{D_{\zixiangb{k}}}{D_{\text{tot}}} F_{\zixiangb{k}} (\vect{w}) = \frac{1}{D_{\text{tot}}}\sum_{k=1}^N \sum_{\vect{z}\in\mathcal{D}_k} f(\vect{w}, \vect{z})$. 
The global loss function measures how well the model fits the entire corpus of data on average. The learning objective is to find the best model parameter $\vect{w}^*$ that minimizes the global loss function: $\vect{w}^* = \argmin_\vect{w} F(\vect{w}).$ Let $F^*$ and $F_k^*$ be the minimum value of $F$ and $F_k$, respectively. Then, \zixiangTCCN{$\Gamma = F^* - \sum_{k=1}^N \frac{D_{k}}{D_{\text{tot}}}F_k^*$} quantifies the degree of non-IID as shown in \cite{li2019convergence}. 



\subsection{FL over Noisy Uplink and Downlink Channels}
\label{sec:model_noise}

We study a generic FL framework where {\em partial} client participation and {\em non-IID} local datasets, two critical features that separate FL from conventional distributed ML, are explicitly captured. Unlike the existing literature, we focus on imperfect communications and consider that \emph{both} the upload and download transmissions take place over noisy communication channels. The overall system diagram is depicted in Fig.~\ref{fig:SystemDiagram}.  In particular, the  FL-over-noisy-channel pipeline works by iteratively executing the following steps at the $t$-th learning round, $\forall t \in [T]$. 

\begin{figure}[tb]
    \centering
    \includegraphics[width=1.0\textwidth]{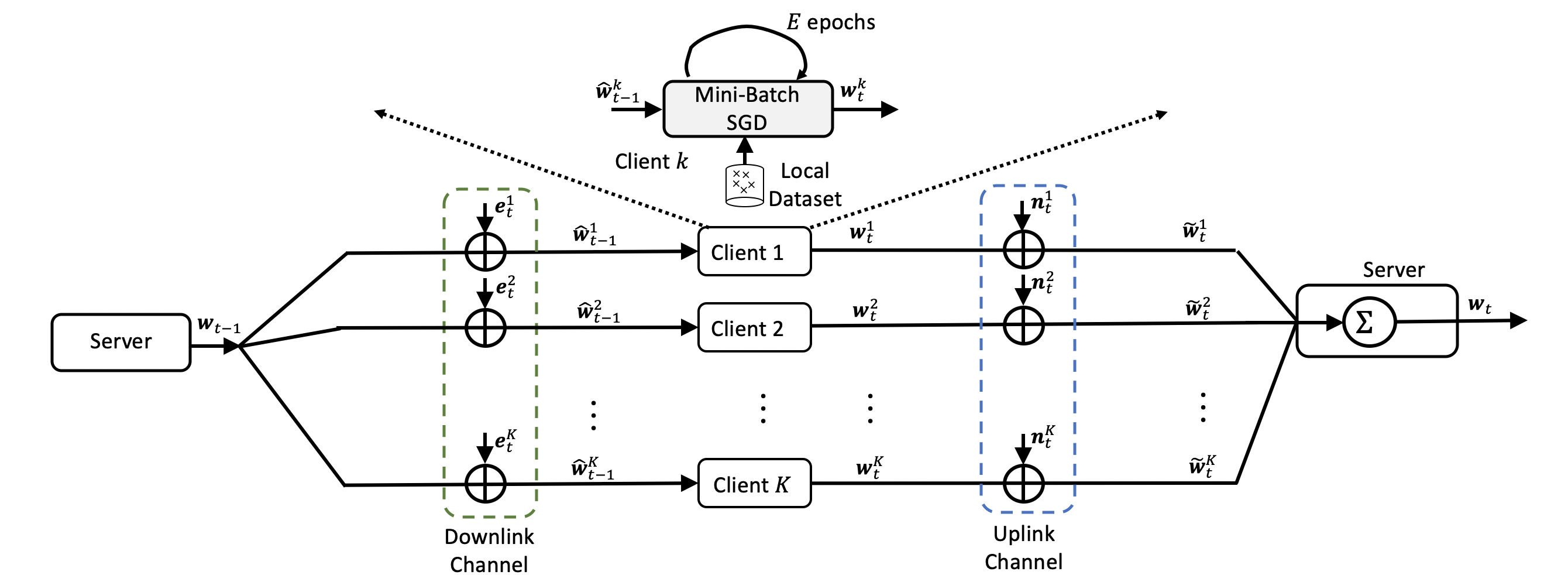}
    \caption{End-to-end FL system diagram in the $t$-th communication round. The impact of noisy channels in both uplink and downlink is captured.}
    \label{fig:SystemDiagram}
    \vspace{-0.1in}
\end{figure}

    \mypara{(1) Downlink communication for global model download.} The centralized server broadcasts the current global ML model, which is described by the latest weight vector $\vect{w}_{t-1}$ from the previous round, to a set of uniformly randomly selected clients\footnote{We note that for partial clients participation, we have $K<N$; in the case of full clients participation we have $K=N$.} denoted as $\mathcal{S}_t$ with $\vert \mathcal{S}_t \vert = K$.  Because of the imperfection introduced in communications, \zixiangTCCN{ e.g., channel noise, imperfect channel estimation, and detection or estimation error}, client $k$ receives a noisy version of $\vect{w}_{t-1}$, which is written as
    \begin{equation}
    \label{eqn:dlnoise}
     \vect{\hat w}_{t-1}^k = \vect{w}_{t-1} + \vect{e}_{t}^k,
    \end{equation}
    where $\vect{e}_{t}^k = [e_{t,1}^k, \cdots,  e_{t,d}^k]^T \in \mathbb{R}^d$ is the $d$-dimensional downlink {\em effective noise} vector at client $k$ and time $t$. We assume that $\vect{e}_{t}^k$ is a zero-mean random vector consisting of IID elements with  variance: 
    \begin{equation}
    \label{eqn:dlnoisevar}
    \mathbb{E}||e_{t,i}^k||^2=\zeta^2_{t,k} \quad \text{and} \quad \mathbb{E}||\vect{e}_t^k||^2=d\zeta^2_{t,k}, \; \forall t\in [T], k\in \mathcal{S}_t, i \in [d].
    \end{equation}
    \congr{
    \zixiangTCCN{{\em \underline{Effective noise and definition of SNR.}}} In order to keep the problem general, we do not specify a particular communication system for the actual downlink data transmission, and only use the effective noise model in \eqref{eqn:dlnoise}. The same approach applies to the uplink. This is a conscious choice to keep the problem general, and we want to focus on analyzing the impact of communication-induced noise and further controlling the resulting SNR to improve FL performance. In this way, $\vect{e}_{t}^k$ shall be interpreted as the effective noise that captures all the processing components in a downlink communication phase in addition to the natural channel noise\footnote{As a simple example, if the downlink communication is over a standard Additive White Gaussian Noise (AWGN) channel, then the actual received signal at client $k$ is $\vect{y}_{t-1}^k = \sqrt{P_{t-1}}\vect{w}_{t-1} + \vect{z}_{t}^k $ where $\vect{z}_{t}^k$ represents the AWGN and $P_{t-1}$ is the downlink broadcast transmit power. The effective channel noise becomes $\vect{e}_{t}^k = \frac{1}{\sqrt{P_{t-1}}} \vect{z}_{t}^k$.}. 
    Because we have normalized the variance of each scalar model parameter as described in Section~\ref{sec:modelFL}, the  (post-processing) receive SNR for the $k$-th client at the $t$-th communication round can be written as
    \begin{equation} \label{eqn:dlsnr}
    \ssf{SNR}_{t,k}^\text{L} = \frac{\expt||\vect{w}_{t-1}||^2}{\expt||\vect{e}_t^k||^2} = \frac{1}{\zeta^2_{t,k}}.
    \end{equation}
    Lastly we note that the noise assumption is very mild, because \eqref{eqn:dlnoisevar} only requires a \emph{bounded} variance of the random noise, but does not limit to any particular distribution. In addition, the downlink communication model is very general in the sense that the effective noise variances, $\{\zeta^2_{t,k}\}$,  can be different for different clients and at different rounds.
    }

    \mypara{(2) Local computation.} Each client uses its local data to train a local ML model improved upon the received global ML model. In this work, we assume that \zixiangTCCN{mini-batch SGD} is used in the model training. Note that this is the most commonly adopted training method in modern ML tasks, e.g., deep neural networks, but its analysis is more complicated than gradient descent (GD) when communication noise is present. 
    
    Specifically, mini-batch SGD operates by updating the weight iteratively (for $E$ steps in each learning round) at client $k$ as follows:
    \begin{align*}
        \text{Initialization: }& \qquad \vect{w}_{t,0}^k = \vect{\hat{w}}_{t-1}^k, \\
        \text{Iteration: }& \qquad \vect{w}_{t,\tau}^k = \vect{w}_{t,\tau-1}^k - \eta_t \nabla  F{_k}(\vect{w}_{{t},\tau-1}^k, \xi_{\tau}^k), \forall \tau=1, \cdots, E,\\
        \text{Output: }& \qquad \vect{w}_{t}^k =\vect{w}_{{t},E}^k,
    \end{align*}
    where $\xi_{\tau}^k$ is a batch of data points that are sampled independently and uniformly at random from the local dataset of client $k$ in the $\tau$-th iteration of mini-batch SGD.

    \mypara{(3) Uplink communication for local model upload.}  The $K$ participating clients upload their latest local models to the server. More specifically, client $k$ transmits a vector $\vect{x}_t^k$ to the server at the $t$-th round. We again consider the practical case where the server receives a noisy version of the individual weight vectors from each client in the uplink communications (\zixiangTCCN{e.g.,} channel noise, fading, transmitter and receiver distortion). The received vector for client $k$ can be written as
        \begin{equation}
        \label{eqn:ulnoise}
            \vect{\hat x}_{t}^k = \vect{x}_t^k + \vect{n}_t^k,
        \end{equation}
    where $\vect{n}_t^k \in \mathbb{R}^d$ is the $d$-dimensional uplink \emph{effective noise} vector for decoding client $k$'s model at time $t$. We assume that $\vect{n}_t^k$ is a zero-mean random vector consisting of IID elements with bounded variance:  
    \begin{equation}
    \label{eqn:ulnoisevar}
    \mathbb{E}||n_{t,i}^k||^2=\sigma^2_{t,k} \quad \text{and} \quad  \mathbb{E}||\vect{n}_{t}^k||^2=d\sigma^2_{t,k}, \; \forall t\in [T], k \in \mathcal{S}_t, i \in [d].
    \end{equation}
    We again note that the uplink communication model in \eqref{eqn:ulnoisevar} is very general in the sense that (1) only bounded variance is assumed as opposed to the specific noise distribution; and (2) the effective noise variances, $\{\sigma^2_{t,k}\}$, can be different for different clients and at different rounds.
    
    Unlike in the download phase where the model itself is transmitted to clients, two different choices of the vector $\vect{x}_t^k$ for model upload are considered in this paper.
    \begin{enumerate}
        \item 
        \textbf{Model Transmission (MT).} The $K$ participating clients upload the latest local models: $\vect{x}_t^k = \vect{w}_t^k$.
        Following \eqref{eqn:ulnoise}, the server receives the updated local model of client $k$ as
        \begin{equation}
        \label{eqn:ulmtnoise}
        \vect{\tilde w}_t^k = \vect{\hat x}_t^k = \vect{w}_t^k + \vect{n}_t^k.
        \end{equation}
        \item 
        \textbf{Model Differential Transmission (MDT).} The $K$ participating clients only upload the differences between the latest local model and the previously received (noisy) global model, i.e., $\vect{x}_t^k =\vect{d}_t^k \triangleq \vect{w}_t^k-\vect{\hat w}_{t-1}^k.$ 
        For MDT, the server uses $\vect{d}_t^k$ and the previously computed global model $\vect{w}_{t-1}$ to reconstruct the updated local model of client $k$ as
       
        \begin{equation}
        \label{eqn:ulmdtnoise}
            \vect{\tilde w}_{t}^k = \vect{w}_{t-1} + \vect{\hat x}_t^k = \vect{w}_{t-1} + \vect{d}_t^k + \vect{n}_t^k =  \vect{w}_t^k + \vect{n}_t^k - \vect{e}_t^k.
        \end{equation}
    \end{enumerate}
    The SNR for these two models, however, has to be defined slightly differently because we have normalized the ML model parameter $\vect{w}$ to have unit-variance elements in Section~\ref{sec:modelFL}. Thus, for MT, we can write the receive SNR  at the server for $k$-th client's signal as 
    \begin{equation}
    \label{eqn:MTulSNR}
        \ssf{SNR}_{t,k}^\text{S,MT} =\frac{\expt \norm{\vect{w}_t^k}^2}{\expt \norm{\vect{n}_t^k}^2} =\frac{1}{\sigma^2_{t,k}}.
    \end{equation} 
    For MDT, we keep the SNR expression general since the variance of model difference $\vect{d}_t^k$ is unknown \emph{a priori} and also changes over time. We have: 
    \begin{equation}
    \label{eqn:MDTulSNR}
        \ssf{SNR}_{t,k}^\text{S,MDT} = \frac{\expt \norm{\vect{d}_t^k}^2}{\expt \norm{\vect{n}_t^k}^2} =\frac{ \expt \norm{\vect{d}_t^k}^2}{d\sigma^2_{t,k}}.
    \end{equation} 
    \zixiangTCCN{\underline{\em Differences between MT and MDT, and why they are both considered.}}
    The different choices of MT and MDT are not considered in most of the literature because with a perfect communication assumption, there is no difference between them from a pure learning perspective -- as long as the server can reconstruct $\vect{w}_t^k$, this aspect does not impact the learning performance \cite{mcmahan2017fl}. However, the choice becomes significant when communication noises are present. From a practical system point of view, both schemes can be useful in different use cases. For example, MDT in the uplink relies on the server keeping the previous global model $\vect{w}_{t-1}$, from which the new local models can be reconstructed. This, however, may not always be true if the server deletes intermediate model aggregation (after broadcast) for privacy preservation \cite{bonawitz2019towards}, which makes reconstruction from the model differential infeasible. 
    
    We also note that the {\em download} phase, on the other hand, does not have these two choices -- we always transmit the global model $\vect{w}_{t-1}$ itself. This is because we have partial (and random) clients participation, where the set of clients participating the $t$-th round can be totally different from the $(t-1)$-th round, and they do not have the previous global model to reconstruct based on the model difference.

    \zixiangTCCN{\em \underline{Noise propagation.}} Both uplink and downlink channel noises collectively impact the received local models at the server. This {noise propagation} effect is more prominent in MDT (\eqref{eqn:ulmdtnoise} explicitly has both noise terms). However, this effect in fact exists in both cases, because the local model is trained using the previously received global model, which contains the downlink noise.

    \mypara{(4) Global aggregation.}  The server aggregates the received local models to generate a new global ML model, following the standard \textsc{FedAvg} \cite{mcmahan2017fl}:
	$
        \vect{w}_{t} = \sum_{k \in \mathcal{S}_t } \frac{D_k}{\sum_{i \in \mathcal{S}_t} D_i}\vect{\tilde w}_t^k.
        	$
    The server then moves on to the $(t+1)$-th round. 
    For ease of exposition and to simply the analysis, we assume in the remainder of the paper that the local dataset sizes at all clients are the same\footnote{We emphasize that all the results of this paper can be extended to handle different local dataset sizes.}: $D_i = D_j$, $\forall i, j \in [N]$, which leads to the following simplifications.

    \begin{enumerate}
        \item \textbf{MT.} The aggregation can be simplified as
            \begin{equation}
            \label{eqn:glbMT}
                \vect{w}_{t} =  \frac{1}{K} \sum_{k \in \mathcal{S}_{t}} \vect{\tilde w}_{t}^k = \frac{1}{K} \sum_{k \in \mathcal{S}_{t}} \vect{\hat x}_{t}^k {= \frac{1}{K} \sum_{k \in \mathcal{S}_{t}} \left( \vect{w}_{t}^k + \vect{n}_{t}^k \right)}.
            \end{equation}
        \item \textbf{MDT.} The aggregation can be written as
        \begin{equation}
        \label{eqn:glbMDT}
            \vect{w}_{t} =\frac{1}{K} \sum_{k \in \mathcal{S}_{t}} \vect{\tilde w}_{t}^k =\vect{w}_{t-1} + \frac{1}{K} \sum_{k \in \mathcal{S}_{t}} \vect{\hat x}_{t}^k {= \frac{1}{K} \sum_{k \in \mathcal{S}_{t}} \left( \vect{w}_{t}^k + \vect{n}_{t}^k - \vect{e}_{t}^k \right).}
        \end{equation}
    \end{enumerate}
    For the case of MT, the SNR for the global model (after aggregation) can be written as
    \begin{equation}
    \label{eqn:snrgt1}
        \ssf{SNR}_{t}^G= \frac{\expt||\sum_{k \in \mathcal{S}_{t}}\vect{w}_{t}^k||^2}{\expt||\sum_{k \in \mathcal{S}_{t}}\vect{n}_{t}^k||^2}= \frac{\expt||\sum_{k \in \mathcal{S}_{t}}\vect{w}_{t}^k||^2 }{ d \sigma^2_{t}},
    \end{equation}
    and for MDT, the SNR for the global model can be written as
    \begin{equation}
    \label{eqn:snrgt2}
        \ssf{SNR}_{t}^G= \frac{\expt||\sum_{k \in \mathcal{S}_{t}}\vect{w}_{t}^k||^2}{\expt||\sum_{k \in \mathcal{S}_{t}} (\vect{n}_{t}^k- \vect{e}_{t}^k)||^2}= \frac{\expt||\sum_{k \in \mathcal{S}_{t}}\vect{w}_{t}^k||^2 }{ d( \sigma^2_{t} + \zeta^2_t )},
    \end{equation}
    where $\sigma_{t}^2\triangleq\sum_{k \in \mathcal{S}_{t}}\sigma^2_{t,k}$ and $\zeta_{t}^2\triangleq\sum_{k \in \mathcal{S}_{t}}\zeta^2_{t,k}$ denote the total uplink and downlink effective noise power for participating clients, respectively. 
    
    In general, $\{\vect{w}_{t}^k\}$ are correlated across clients because the local model updates all start from (roughly) the same global model. Intuitively, once FL convergences, these models will largely be the same, leading to a signal power term of $dK^2$ for the numerator. On the other hand, if we assume that these local models are independent across clients, which is reasonable in the early phases of FL with large local epochs, where the (roughly) same starting point has diminishing impact due to the long training period and non-IID nature of the data distribution, we can have a signal power term of $dK$. Nevertheless, since the SNR control can be realized by adjusting the effective noise power levels, we focus on the impact of $\sigma^2_{t}$ and $\zeta^2_t$ on the FL performance in Section~\ref{sec:convergence}. 
    


\zixiangTCCN{
In this paper, we mainly focus on analog communication for FL, where model parameters are transmitted in an analog manner. Therefore, digital communication processing such as source coding, channel coding and modulation are not incorporated. The adopted zero-mean bounded random noise assumption is reasonable for this setting, because it does not require any specific distribution and thus can be applicable to a broad range of analog communication systems.}

\section{Convergence Analysis of FL over Noisy Channels}
\label{sec:convergence}

\subsection{Convergence Analysis for Model Transmission for Full Clients Participation}
\label{sec:conv_MT_full}

We first analyze the convergence of \textsc{FedAvg} in the presence of both uplink and downlink communication noise when direct model transmission (MT) is adopted for local model upload: $\vect{x}_t^k = \vect{w}_t^k$. To \zixiangb{simplify} the analysis and highlight the key techniques in deriving the convergence rate, we assume $K=N$ in this subsection (i.e., \textit{full} clients participation with $\mathcal{S}_t = [K] = [N]$), and leave the case of partial clients participation to Section~\ref{sec:conv_MT_part}. 

We make the following standard assumptions that are commonly adopted in the convergence analysis of \textsc{FedAvg} and its variants; see \cite{li2019convergence,jiang2018nips,stich2018local,reisizadeh2019fedpaq,Zheng2020jsac}.  In particular, Assumption \ref{as:F}-2) indicates that we focus on strongly convex $F_k(\cdot)$, which represents a category of loss functions that are widely studied in the literature.  
\begin{assumption}\label{as:F}
\begin{enumerate}[leftmargin=12pt,topsep=0pt, itemsep=0pt,parsep=0pt]
    \item[1)] \textbf{$L$-smooth:} $\forall~\vect{v}$ and $\vect{w}$, $F_k(\vect{v}) \leq F_k(\vect{w}) + (\vect{v} - \vect{w})^T \nabla F_k(\vect{w}) + \frac{L}{2} \norm{\vect{v} - \vect{w}}^2$.
    
    \item[2)] \textbf{$\mu$-strongly convex:} $\forall~\vect{v}$ and $\vect{w}$, $F_k(\vect{v}) \geq F_k(\vect{w}) + (\vect{v} - \vect{w})^T \nabla F_k(\vect{w}) + \frac{\mu}{2} \norm{\vect{v} - \vect{w}}^2$.
    
    \item[3)] \textbf{Bounded variance for unbiased mini-batch SGD:}  
    The mini-batch SGD is unbiased: $\expt [ \nabla F_k (\vect{w}, \xi) ] = \nabla F_k (\vect{w})$, and the variance of stochastic gradients is bounded: $\expt \norm{\nabla F_k (\vect{w}, \xi) - \nabla F_k (\vect{w})}^2 \leq \delta_k^2$,  for mini-batch data $\xi$ at client $k \in [N]$. 
    
    \item[4)] \textbf{Uniformly bounded gradient:} $\expt \norm{\nabla F_k (\vect{w}, \xi)}^2 \leq H^2$ for mini-batch data $\xi$ at client $k \in [N]$.
\end{enumerate}
\end{assumption}

We present the main convergence result of MT with full clients participation in Theorem \ref{thm:MTfull}.



\begin{theorem}
\label{thm:MTfull}
    Define $\phi = {L}/{\mu}$, $\gamma = \max\{8\phi, E\}$. Set learning rate as $\eta_t={2}/(\mu(\gamma+t))$ and adopt a SNR control policy that scales the effective uplink and downlink noise power over $t$  such that: 
    \begin{align}
    \sigma_{t}^2 &\leq \frac{4N^2}{\mu^2(\gamma+t-1)^2} \sim \mathcal{O}\left(\frac{1}{t^2} \right)  \label{eqn:thmMTfull_ul} \\
    \zeta_{t}^2 &\leq \frac{4N^2}{\mu^2(\gamma+t)(\gamma+t-2)}\sim\mathcal{O}\left(\frac{1}{t^2}\right). \label{eqn:thmMTfull_dl}
    \end{align}
    where $\sigma_{t}^2\triangleq\sum_{k \in [N]}\sigma^2_{t,k}$ and $\zeta_{t}^2\triangleq\sum_{k \in [N]}\zeta^2_{t,k}$ denote the {total} uplink and downlink effective noise power, respectively. 
    Then, under Assumption \ref{as:F}, the convergence of \textsc{FedAvg} with non-IID datasets and full clients participation satisfies
        \begin{equation}
        \label{eqn:conv_res}
            \expt \norm{\vect{w}_T - \vect{w}^*}^2 \leq \frac{8L + \mu E}{\mu(T+\gamma)} \norm{\vect{w}_0 - \vect{w}^*}^2 +   \frac{4D}{\mu^2 (T+\gamma)} 
        \end{equation}
    with $D = \sum_{k=1}^{N}{\delta_k^2}/{N^2} + 6L \Gamma + 8(E-1)^2H^2 + 2d$.
\end{theorem}


A few remarks about Theorem~\ref{thm:MTfull} and its proof are now in order.

\begin{remark} \normalfont
A complete proof of Theorem~\ref{thm:MTfull} can be found in Appendix \ref{app:proof_thm1}. The core technique utilized in Appendix \ref{app:proof_thm1} is the \emph{perturbed iterate framework} that was pioneered in \cite{mania2017siam}, especially the virtual sequence construction that have been widely adopted in the distributed SGD analysis \cite{stich2018local,li2019convergence,reisizadeh2019fedpaq,Zheng2020jsac}. The {unique challenge} of this proof, however, is how to handle \emph{simultaneous} uplink and downlink noises, which cannot be isolated from the SGD iterations. Not only do we have to incorporate more virtual sequences in the proof, but they also have the ``coupling'' effect in that downlink noise is present in the SGD steps and further in the new local model for uplink, while the uplink noise is present in the next-round downlink model. A careful manipulation of these coupled noise components in the various virtual sequences is a key analytical novelty of the proof.
\end{remark}

\begin{remark} \normalfont
We make an important clarification that although the requirement of Theorem~\ref{thm:MTfull} is presented in terms of the effective noise power, what ultimately matters is the SNR defined in Section~\ref{sec:model_noise}. Controlling the effective noise power to scale as $\mathcal{O}(1/t^2)$ is equivalent to scaling the SNR as $\mathcal{O}(t^2)$, and can be implemented by either increasing the signal power (e.g., transmit power control) or reducing the post-processing noise power (e.g., receive diversity combining) while satisfying a fixed total resource budget constraint. We discuss design examples that realize the requirement of Theorem~\ref{thm:MTfull}  in Section~\ref{sec:design}. 
\end{remark}

\begin{remark} \normalfont
\congr{It is not surprising to see that Theorem~\ref{thm:MTfull} requires the SNR to increase, which gradually suppresses the noise effect as the FL process converges. There are, however, two unique characteristics about this theorem:
\begin{enumerate}[leftmargin=*]\itemsep=0pt
\item It characterizes a sufficient condition for the SNR \textbf{scaling law} as $\mathcal{O}(t^2)$. As we will see in Section~\ref{sec:experiment}, choosing a SNR scaling that is slower than $\mathcal{O}(t^2)$ degrades the FL performance.
\item This $\mathcal{O}(t^2)$ scaling law can be realized under a fixed total budget constraint. In other words, the benefit of Theorem~\ref{thm:MTfull} does not come from using more communication resources, but rather is due to a more judicious allocation (following the scaling law) of the same \zixiangb{resource budget}.
\end{enumerate}
}
\end{remark}

\begin{remark} \normalfont
Theorem~\ref{thm:MTfull} guarantees that even under {\em simultaneous} uplink and downlink noisy communications, the same $\mathcal{O}(1/T)$ convergence rate of \textsc{FedAvg} with {perfect} communications can be achieved {if we control the effective noise power of both uplink and downlink to scale at rate $\mathcal{O}(1/t^2)$ and choose the learning rate at $\mathcal{O}(1/t)$ over $t$}. We note that the choice of $\eta_t$ to scale as $\mathcal{O}(1/t)$ is well-known in distributed and federated learning \cite{stich2018local,wang2018cooperative,jiang2018nips,li2019convergence}, which essentially controls the ``SGD noise'' that is inherent to the {stochastic} process in SGD to gradually shrink as the FL process converges. \zixiangTCCN{We also note that for other learning rate choices in SGD, the fundamental insight of Theorem~\ref{thm:MTfull}, i.e.,  \emph{controlling the ``effective channel noise'' to not dominate the ``SGD noise''}, is still valid. We will investigate the convergence requirement that clients adopt different learning rates in future research.
}   
\end{remark}

\begin{remark} \normalfont
Lastly, we note that the scaling law in Theorem~\ref{thm:MTfull} should be viewed as an \emph{average} SNR requirement that changes over learning rounds. The time scale of changing the average SNR is on the order of learning rounds, which is much slower\footnote{This is particularly true when there are a large number of clients participating in the FL process, as the length of learning rounds is often dominated by the ``straggler'' \cite{reisizadeh2020straggler}.} than the time scale of the time-varying wireless channel. 
Furthermore, the SNR scaling law can be used in conjunction with other ``faster'' resource allocation mechanisms, such as inner-loop power control, to handle wireless dynamics under the average SNR budget decided from Theorem~\ref{thm:MTfull}. This will become clear in Section~\ref{sec:tpc}.
\end{remark}



\subsection{Convergence Analysis for Model Transmission for Partial Clients Participation}
\label{sec:conv_MT_part}

We now generalize the convergence analysis for full clients participation to \emph{partial} clients participation, where we have a given $K<N$ and uniformly randomly select a set of clients $\mathcal{S}_{t}$ at round $t$ to carry out the FL process. In this section, we mostly follow the FL system model described in Section~\ref{sec:model_noise}, with the only simplification that we consider \emph{homogeneous} noise power levels at the uplink and downlink, i.e., we assume 
\begin{equation}
\label{MT_pt_noise}
    \sigma_{t,k}^2 = \bar{\sigma}_t^2, \quad \text{and} \quad \zeta_{t,k}^2 = \bar{\zeta}_t^2, \quad \forall t \in [T], k \in [N].
\end{equation}
The main reason to introduce this simplification is due to the time-varying randomly participating clients: since $\mathcal{S}_{t}$ changes over $t$, the total power levels also vary over $t$ if we insist on heterogeneous noise power for different clients. Furthermore, since clients are randomly selected, the  total power level becomes a random variable as well, which significantly complicates the convergence analysis. Making this assumption would allow us to focus on the challenge with respect to the model update from {partial} clients participation.

\begin{theorem}
\label{thm:MTpart}
    Let $\phi$, $\gamma$ and $\eta_t$ be the same as in Theorem~\ref{thm:MTfull}. Adopt a SNR control policy that scales the effective uplink and downlink noise power over $t$ such that: 
    \begin{align}
    \bar{\sigma}_t^2 &\leq \frac{4K}{\mu^2(\gamma+t-1)^2} \sim \mathcal{O}\left(\frac{1}{t^2} \right)  \label{eqn:thmMTpart_ul} \\
    \bar{\zeta}_t^2 &\leq \frac{4N}{\mu^2(\gamma+t)(\gamma+t-2)}\sim\mathcal{O}\left(\frac{1}{t^2}\right). \label{eqn:thmMTpart_dl}
    \end{align}
    where $\bar{\sigma}_t^2$ and $\bar{\zeta}_t^2$ represent the individual client effective noise in the uplink and downlink, respectively, which are defined in \eqref{MT_pt_noise}. 
    Then, under Assumption \ref{as:F}, the convergence of \textsc{FedAvg} with non-IID datasets and partial clients participation has the same convergence rate expression as \eqref{eqn:conv_res}, with $D$ being replaced as $D = \sum_{k=1}^N {\delta_k^2}/{N^2} + {4(N-K)E^2 H^2}/(K(N-1)) +  6L \Gamma + 8(E-1)^2 H^2 + 2d$.
\end{theorem}

The proof of Theorem \ref{thm:MTpart} is given in Appendix \ref{app:proof_thm_MTpart}. We can see that partial clients participation does not fundamentally change the behavior of FL in the presence of uplink and downlink communication noises. 
\zixiangTCCN{However, unlike full client participation, the uplink effective noise depends on the number of active users, which makes it harder to satisfy \eqref{eqn:thmMTpart_ul} compared with \eqref{eqn:thmMTpart_dl}. The reason behind this difference is that, in partial client participation, the downlink process remains the same as the fully client participation, while the number of participants in the uplink process reduces from $N$ to $K$. Therefore, the effective uplink noise can only be controlled by $K$ rather than $N$ participants, which implies that each user needs to allocate more transmission power than the fully client participation case to achieve the desired noise-free convergence rate of FL. We will provide a practical example to handle this tighter upper bound in Section \ref{sec:diversity}.
}

\subsection{Convergence Analysis for Model Differential Transmission}
\label{sec:conv_MDT}

In this section, we consider the model different transmission (MDT) scheme when the clients upload model parameters. Since only model differential is transmitted, the receiver must possess a copy of the ``base'' model to reconstruct the updated model. This precludes using MDT in the downlink for partial clients participation, because participating clients differ from round to round, and a newly participating client does not have the ``base'' model of the previous round to reconstruct the new global model. We thus only focus on MDT in the uplink and MT in the downlink with partial clients participation. 


\begin{theorem}
\label{thm:MDTpart}
    Let $\phi$, $\gamma$ and $\eta_t$ be the same as in Theorem~\ref{thm:MTfull}, and the effective noise follows \eqref{MT_pt_noise}. Adopt a SNR control policy that maintains a {constant} uplink SNR at each client over $t$:
    \begin{equation}
    \ssf{SNR}_{t,k}^\text{S,MDT} = \nu \sim\mathcal{O}\left(1\right), \label{eqn:thmMDTpart_ul}
    \end{equation}    
    and scales the effective downlink noise power at each client over $t$ such that:
    \begin{equation}
    \label{eqn:thmMDTpart_dl}
    \bar{\zeta}_t^2 \leq \frac{\frac{4}{\mu^2}}{  \frac{1}{N} (\gamma+t)(\gamma+t-2) + \frac{1}{K} \left( 1+\frac{1}{\nu} \right) (\gamma+t)^2 }   \sim\mathcal{O}\left(\frac{1}{t^2}\right).
    \end{equation}
    Then, under Assumption \ref{as:F}, the convergence of \textsc{FedAvg} with non-IID datasets and partial clients participation for uplink MDT and downlink MT has the same convergence rate expression as \eqref{eqn:conv_res}, with $D$ being replaced as $D = \sum_{k=1}^N {\delta_k^2}/{N^2} + {4(N-K)E^2 H^2}/(K(N-1))  + {4E^2H^2}/(K\nu) + 6L \Gamma + 8(E-1)^2 H^2  + d$.
\end{theorem}
The complete proof of {Theorem \ref{thm:MDTpart}} can be found in Appendix \ref{app:proof_thmMDTpart}. 
It is instrumental to note that unlike direct model transmission, only transmitting model differentials in the uplink allows us to remove the corresponding SNR scaling requirement. Instead, one can keep a \emph{constant} SNR in uplink throughout the entire FL process. Intuitively, this is because the ``scaling'' already takes place in the model differential $\vect{d}_t^k$, which is the difference between the updated local model at client $k$ after $E$ epochs of training and the starting local model. As FL gradually converges, this differential becomes smaller. Thus, by keeping a constant communication SNR, we essentially scales down the effective noise power at the server.

\congr{Lastly, we note that the constant SNR requirement of Theorem \ref{thm:MDTpart} enables very simple implementation given the MDT SNR expression in \eqref{eqn:MDTulSNR}. The signal power in the numerator of \eqref{eqn:MDTulSNR} is unknown and varying over learning rounds. However, a constant SNR requirement means one can fix the transmit power and ``scales'' individual $\vect{d}_t^k$ to have the desired power, without prior knowledge of its true variance.}


\section{Communication Design Examples for FL in Noisy Channels}
\label{sec:design}

An immediate engineering question following the previous analyses is how we can realize the {effective noise power} (or equivalently the SNR) specified in the theorems. A natural approach is \emph{transmit power control}, which has the flexibility of controlling the average receive SNR (and thus the effective noise power) while satisfying a total power constraint. Specially, for a FL task with $T$ total communication rounds and a given total power budget of $P$ over all rounds, it is straightforward to compute that 
\begin{equation} \label{eqn:pwralloc1}
    P_t = 6Pt^2/(T(T+1)(2T+1)), \forall t=1, \cdots, T,
\end{equation}
where $P_t$ is the desired average transmit power of the communication round $t$. 

Since we consider analog communication and aggregation for FL, we also need to take the wireless channel fading into account. To combat the influence of channel fading on received power, we now propose two design examples to demonstrate how the proposed $\mathcal{O}(t^2)$-power increased strategy could be adopted in both continuous and discrete average power allocation schemes.



\subsection{Design Example I: Transmit Power Control for Analog Aggregation}
\label{sec:tpc}

We first design a power control policy for the analog aggregation FL framework in \cite{zhu2019broadband,amiri2020federated,amiri2020sp}, as an example to demonstrate the system design for FL tasks in the presence of communication noise. 

\mypara{The analog aggregation method in \cite{zhu2019broadband,amiri2020federated,amiri2020sp}.} Consider a communication system where several narrowband orthogonal channels \zixiangTCCN{(e.g., sub-carriers in orthogonal frequency-division multiplexing (OFDM), time slots in time division multiple access (TDMA))} are shared by $K$ random selected clients in an uplink model upload phase of a communication round. Each element in the transmitted model $\vect{w}\in \mathbb{R}^d$ is allocated and transmitted in a narrowband channel and aggregated automatically over the air. Denote the received signal of each element $i = 1, \cdots, d$ in the $t$-th communication round as
\begin{equation*}
    y_{t,i} = \frac{1}{K}\sum_{k \in \mathcal{S}_{t}} r_{t,k}^{-\alpha/2} h_{t,k,i} \sqrt{p_{t,k,i}} w_{t,k,i} + n_{t,i} \;\;\forall k \in \mathcal{S}_{t}, 
\end{equation*}
where $r_{t,k}^{-\alpha/2}$ and $h_{t,k,i} \zixiangTCCN{ \sim\mathcal{CN}(0,1)}$ are the large-scale and small-scale fading coefficients of the channel, respectively, \zixiangTCCN{$n_{t,i} \sim\mathcal{CN}(0,1)$} is the additive Gaussian white noise of the channel, and $p_{t,k,i}$ denotes the transmit power determined by the power control policy. We assume perfect channel state information at the transmitters (CSIT). Due to the aggregation requirement of federated learning, the channel inversion rule is used in \cite{zhu2019broadband}, which leads to the following \textit{instantaneous} transmit power of user $k$ at time $t$ for model weight element $i$:
\begin{equation}\label{eqn:channelinversion}
    p_{t,k,i} = \frac{{\rho_{t}^{\text{UL}}}}{r_{t,k}^{-\alpha} |h_{t,k,i}|^2},
\end{equation}
where $\rho_{t}^{\text{UL}}$ is a scalar that denotes the uplink average transmit power, which is to be optimized. Hence, the receive SNR of the global model can be written as
\begin{equation}\label{eqn:snrg_tpc}
    \ssf{SNR}_{t}^G = \mathbb{E}\norm{\frac{1}{K} \sum_{i = 1}^d\frac{\sqrt{\rho_{t}^{\text{UL}}}\sum_{k \in \mathcal{S}_{t}} w_{t,k,i}}{n_{t,i}}}^2 = \frac{\rho_{t}^{\text{UL}}\expt||\sum_{k \in \mathcal{S}_{t}}\vect{w}_{t}^k||^2 }{dK^2}.
\end{equation}

\mypara{Transmit power control.} The original analog aggregation framework in \cite{zhu2019broadband} assumes that $\rho_t^{\text{UL}}$ is a constant over time $t$. However, our theoretical analysis in Section~\ref{sec:convergence} suggests that this can be improved.  
Specifically, if we take partial clients participation and MT as an example, and further assume IID weight elements, we have
\zixiangTCCN{
\begin{equation}
\label{eqn:ulpwr}
    \rho_{t}^{\text{UL}} = \frac{K}{\bar \sigma_t^2} \geq \frac{\mu^2(\gamma + t -1)^2}{4} \sim \mathcal{O}(t^2),
\end{equation}}
by plugging in Theorem~\ref{thm:MTpart}, which implies that $\rho_t$ should be increased at the rate $\mathcal{O}(t^2)$ in the uplink to ensure the convergence of \textsc{FedAvg}. Similar policy can be derived for MDT and/or full clients participation, by invoking the corresponding theorems. 

In the downlink case, when the server broadcasts the global model to $K$ randomly selected clients, the receive signal of the $i$-th element for the $n$-th user in the $t$-th communication round is
\begin{equation*}
    y_{t,n,i} = r_{t,n}^{-\alpha/2} h_{t,n,i} \sqrt{\rho_{t}^{\text{DL}}} w_{t,i} + e_{t,n,i}\;\; \forall n = 1\cdots K,
\end{equation*}
where $e_{t,n,i}\in \mathcal{CN}\sim(0,1)$ is the additive Gaussian white noise, and $\rho_{t}^{\text{DL}}$ is the transmitted power at the server. 
The downlink SNR for the $n$-th user is 
\zixiangTCCN{
\begin{equation}
\label{eqn:snrl_tpc}
    \ssf{SNR}_{t,n,i}^L =  r_{t,n}^{-\alpha}|h_{t,n,i}|^2\rho_{t}^{\text{DL}}.
\end{equation}}
Instead of keeping $\rho_t^{\text{DL}}$ as a constant, we derive the following policy based on Theorem~\ref{thm:MDTpart} to guarantee the convergence of \textsc{FedAvg}: 
\zixiangTCCN{
\begin{equation}
\label{eqn:dlpwr}
    \rho_{t}^{\text{DL}} \geq \frac{ r_{t,k}^{\alpha} \mu^2(\gamma+t)(\gamma+t-2)}{4N|h_{t,n,i}|^2} \sim \mathcal{O}(t^2).
\end{equation}}
Finally, by applying the power control policy defined in Eqns.~\eqref{eqn:ulpwr} and \eqref{eqn:dlpwr}, FL tasks are able to achieve better performances under the same energy budget. This is also numerically validated in the experiment.

\mypara{Remarks.} We note that the proposed transmit power control only changes the \emph{average} transmit power at learning rounds. Such method is often referred to as the \emph{outer-loop power control} (OLPC) \cite{G:05}, which operates at a very slow time scale and only relies on the large-scale, stationary information of the wireless FL system. In fact, this method can be used in conjunction with a faster \emph{inner-loop power control}, such as the channel inversion power in \eqref{eqn:channelinversion} or any other methods that handle the fast fading component or interference, to determine the instantaneous transmit power of the sender. 
Another minor note is that the pathloss component appears in \eqref{eqn:dlpwr} but not in \eqref{eqn:ulpwr}. This is due to the broadcast nature of download. For upload, the pathloss is absorbed in the channel inversion expression \eqref{eqn:channelinversion}. 


\subsection{Design Example II: Receive Diversity Combining for Analog Aggregation}
\label{sec:diversity}

Another technique that can benefit from our theoretical results is to control the \emph{diversity order} of a receiver combining scheme, such as using multiple receive antennas, multiple time slots, or multiple frequency resources. Essentially we are leveraging the repeated transmissions to reduce the effective noise power via receive diversity combining, and by only activating sufficient diversity branches as we progress over the learning rounds, resources can be more efficiently utilized.


\mypara{Uplink diversity requirement.} We assume the uploaded local model is independently received $L_t$ times (over time, frequency, space, or some combination of them) in the $t$-th round. Reusing the notations and the channel inversion rule in \eqref{eqn:channelinversion}, the $L_t$ received signals for the $i$-th element can be denoted as
\begin{equation*}
    y_{t,i,l} = \frac{1}{K}\sum_{k=1}^K\sqrt{\rho_{t,l}}w_{t,k,i}+n_{t,i,l}\;\;\forall k \in \mathcal{S}_{t},\;\;\forall l=1\cdots L_t.
\end{equation*}
For simplicity, we fixed the average transmit power for each branch:  $\rho_{t,l} = \rho_{0}$, but this can be easily extended to incorporate power allocation over diversity branches \cite{G:05}. The receive SNR of the global model after the diversity combining can be written as 
\begin{equation*}
    \ssf{SNR}_{t}^G = \mathbb{E}\norm{\sum_{i = 1}^d\frac{\sum_{l=1}^{L_t}\frac{\sqrt{\rho_{t,l}}}{K}\sum_{k\in \mathcal{S}_{t}} w_{t,k,i}}{\sum_{l=1}^{L_t} n_{t,i,l}}}^2 = L_t\frac{ \rho_{0}\expt||\sum_{k \in \mathcal{S}_{t}}\vect{w}_{t}^k||^2}{dK^2}.
\end{equation*}
Compared with the SNR of the power control policy in (\ref{eqn:snrg_tpc}), we can derive the diversity requirement as
\begin{equation}
\label{eqn:ul_div}
    L_t = \left \lceil\rho_t^{\text{UL}}/\rho_0 \right \rceil,
\end{equation}
where $\lceil a \rceil$ denotes the ceiling operation on $a$.

\mypara{Downlink diversity requirement.} The server broadcasts the global weight for $Q_t$ times (again it can be over time, frequency, space, or some combination of them) in the $t$-th  round and each client combines the multiple independent copies of the received signals to achieve a higher SNR (i.e., lower effective noise power). The receive signal at client $k$ can be written as
\begin{equation*}
    y_{t,k,i,q} = r_{t,k,i,q}^{-\alpha/2} h_{t,k,i,q} \sqrt{\rho_{t,q}} w_{t,i} + e_{t,k,i,q}\;\; \forall q = 1\cdots Q_t,
\end{equation*}
where $\rho_{t,q} = \rho_1$ is the (constant) transmit power at the server. The downlink SNR for the $k$-th user is
$    \ssf{SNR}_{t,k}^L =  r_{t,k}^{-\alpha}Q_t\rho_1.$ 
Similarly, compared with the local SNR in (\ref{eqn:snrl_tpc}), we can derive the diversity requirement as
\begin{equation}
\label{eqn:dl_div}
    Q_t = \left \lceil{\rho_t^{\text{DL}}}/{\rho_1} \right \rceil.
\end{equation}
By applying the combining rules in Eqns.~\eqref{eqn:ul_div} and \eqref{eqn:dl_div}, we have the complete design for receive diversity combining that can guarantee the convergence of FL at rate $\mathcal{O}(1/T)$, under the transmit power constraints at both clients and server.

\mypara{Remarks.}  Receive diversity combining is not as flexible as power control, because it can only achieve \emph{discrete} effective noise power levels. This is also observed in the experiments. However, it can be useful in situations where adjusting the average transmit power is not feasible, e.g., no change at the transmitter is allowed. In addition, one can combine the transmit power control in Section~\ref{sec:tpc} with the receive diversity combining in Section~\ref{sec:diversity} in a straightforward manner. We also note that there are other methods, such as increasing the precision of Analog-to-Digital Converters (ADC), to implement the SNR control policy. The general design principles in Theorems~\ref{thm:MTfull} to \ref{thm:MDTpart} can be similarly realized. 

\section{Experiment Results}
\label{sec:experiment}

\subsection{Experiment Setup}
We consider noisy uplink and downlink communications to support various FL tasks. 
For simplicity, we assume that every channel use has the same noise level, and we also assume that both uplink and downlink have the same total energy budget $P = \sum_{t=1}^T P_t$, where $P_t$ is the transmission power of the $t$-th round, $t=1,\cdots, T$. However, we note that the downlink energy is consumed only by the server (i.e., $P_t$), while the uplink budget is equally shared among all clients (i.e., $P_t/N$ per transmitter), resulting in significantly smaller uplink transmit power per transmitter than the downlink. 
In each round of FL, the updated (locally or globally) ML model (or model differential when applicable) is transmitted over the noisy channel as described in Section~\ref{sec:model_noise}. 
We consider the following four schemes in the experiments.

    \begin{enumerate}[leftmargin=*]\itemsep=0pt
        \item \textbf{Noise free.} This is the {ideal} case with no noise in either uplink or downlink. The accurate model parameters are perfectly received at the server and clients. This represents the best-case performance. 
        \item \textbf{Equal power allocation.} 
        This corresponds to $P_t = P/T, \forall t=1, \cdots, T$, as used in \cite{zhu2019broadband}. \zixiangTCCN{We adopte a normalized transmitted power $P_t = 1$ and the receive SNR of the model parameters is set as $10$ dB in the experiments.} 
        
        \item \textbf{$\mathcal{O}(t^2)$-increased power control policy.} Transmit power increases at the rate of $\mathcal{O}(t^2)$ with the round $t$ but the overall energy consumption is kept constant as other methods, i.e., the receive SNR is increased and the effective noise of the signal is decreased with the progress of FL. With the total budget $P$, \eqref{eqn:pwralloc1} gives the power allocation solution. 
        \item \textbf{$\mathcal{O}(t^2)$-increased diversity combining policy.} The transmit power in both downlink and uplink remains the same as 2). However, the final models at the server and clients of each communication round are obtained by multiple repeated transmissions and the subsequent combining. The number of the repeated transmissions increases at the rate of $\mathcal{O}(t^2)$. For simple discretization, we use $1$, $4$, $9$, $16$ and $25$ orders of receive diversity combining in both uplink and downlink model transmissions for $1$st to $9$th, $10$th to $45$th, $46$th to $125$th, $125$th to $270$th, and $270$th to $500$th communication round, respectively, of a $500$-round task. Note that the total energy budget remains the same as the previous two methods.
    \end{enumerate}
We use the standard image classification and natural language processing FL tasks to evaluate the performances of these schemes. The following three standard datasets are used in the experiments, which are commonly accepted as the benchmark tasks to evaluate the performance of FL.
    \begin{enumerate}[leftmargin=*]\itemsep=0pt
        \item \textbf{MNIST.} The training sets contains $60000$ examples. For the full clients participation case, the training sets are evenly distributed over $N = K = 10$ clients. For the partial clients participation case, the training sets are evenly partitioned over $N=2000$ clients each containing 30 examples, and we set $K=20$ per round ($1\%$ of total users). For the {IID} case, the data is shuffled and randomly assigned to each client, while for the {non-IID} case the data is sorted by labels and each client is then randomly assigned with 1 or 2 labels. The CNN model has two $5 \times 5$ convolution layers, a fully connected layer with 512 units and $\ssf{ReLU}$ activation, and a final output layer with softmax. The first convolution layer has 32 channels while the second one has 64 channels, and both are followed by $2 \times 2$ max pooling. The following parameters are used for training: local batch size $BS=5$, the number of local epochs $E=1$, and learning rate $\eta = 0.065$.
        \item \textbf{CIFAR-10.} We set $N=K=10$ for the full clients participation case while $N=100$ and $K=10$ for the partial clients participation case. We train a CNN model with two $5 \times 5$ convolution layers (both with 64 channels), two fully connected layers (384 and 192 units respectively) with $\ssf{ReLU}$ activation and a final output layer with softmax. The two convolution layers are both followed by $2 \times 2$ max pooling and a local response norm layer. The training parameters are: (a) {IID}: $BS=50$, $E=5$, learning rate initially sets to $\eta=0.15$ and decays every 10 rounds with rate 0.99; (b) {non-IID}: $BS=100$, $E=1$, $\eta=0.1$ and decay every round with rate 0.992.
        \item \textbf{Shakespeare.} This dataset is built from \emph{The Complete Works of William Shakespeare} and each speaking role is viewed as a client. Hence, the dataset is naturally unbalanced and {non-IID} since the number of lines and speaking habits of each role vary significantly. There are totally $1129$ roles in the dataset \cite{caldas2018leaf}. We randomly pick $300$ of them and build a dataset with $794659$ training examples and $198807$ test examples. We also construct an {IID} dataset by shuffling the data and redistribute evenly to $300$ roles and set $K = 10$. The ML task is the next-character prediction, and we use a classifier with an 8D embedding layer, two LSTM layers (each with $256$ hidden units) and a softmax output layer with $86$ nodes. The training parameters are: $BS = 20$, $E = 1$, learning rate initially sets to $\eta = 0.8$ and decays every $10$ rounds with rate $0.99$.
    \end{enumerate}
We compare the test accuracies and training losses as functions of the communication rounds for all the aforementioned configurations. All of the reported results are obtained by averaging over 5 independent runs. We also report the final test accuracy, which is averaged over the last 10 rounds, as the performance of the final global model. 
    
\subsection{Experiment Results for Transmit Power Control}

The focus of the experiment is on partial clients participation under both MT and MDT, but we first report the results for full clients participation in CIFAR-10, to highlight some common observations across all experiments.

\begin{figure}
    \centering
    \subfigure{
        \includegraphics[width=0.23\textwidth]{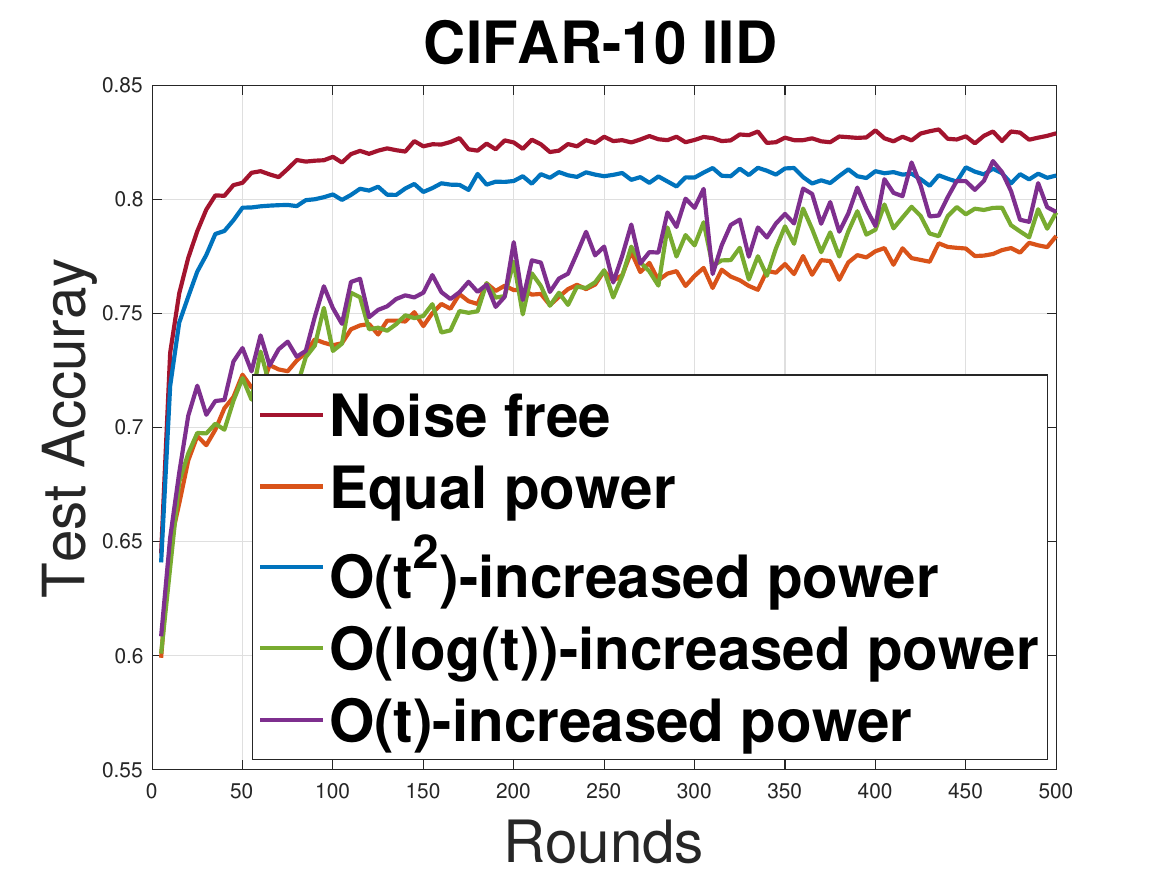}}
    \subfigure{
        \includegraphics[width=0.23\textwidth]{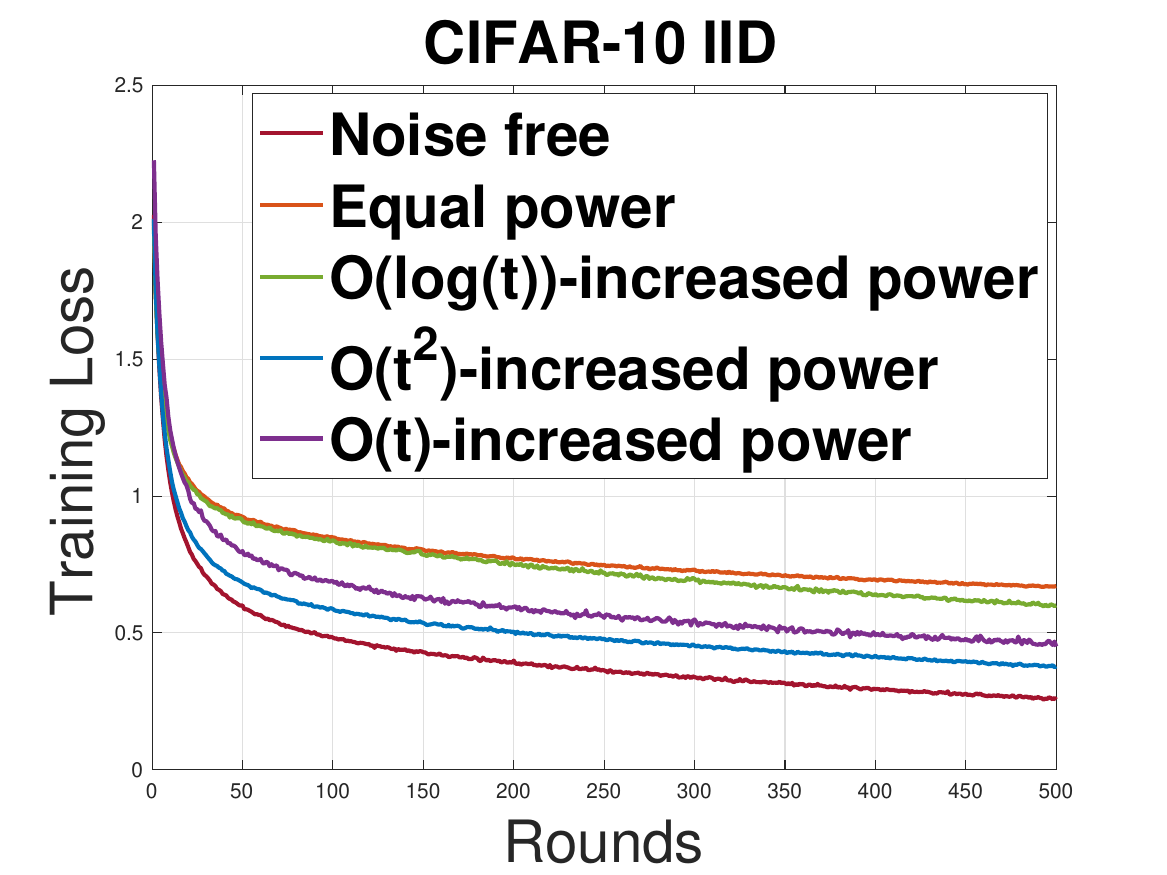}}
    \subfigure{
        \includegraphics[width=0.23\textwidth]{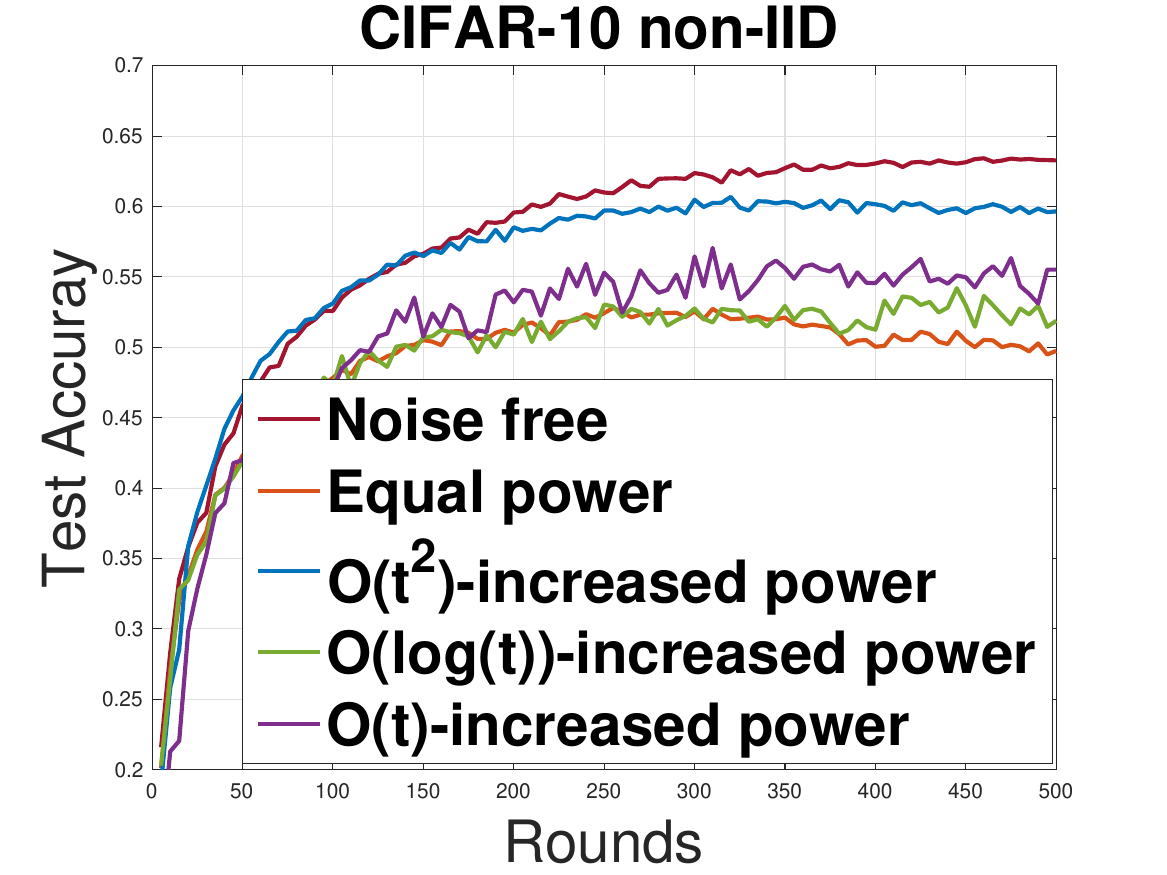}}
    \subfigure{
        \includegraphics[width=0.23\textwidth]{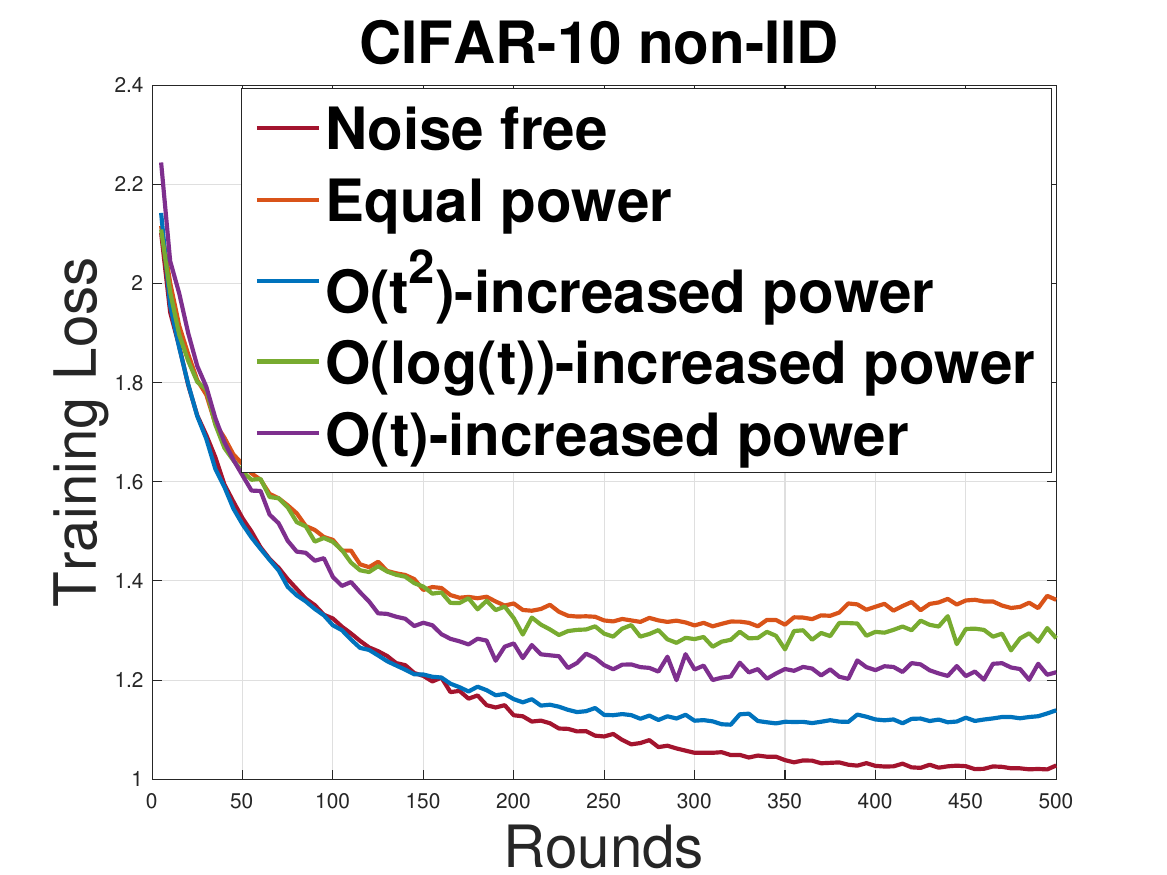}}
    \caption{\zixiangTCCN{Comparing the performance of transmit power control to the baselines with full clients participation, model transmission, and both IID (left two) and non-IID (right two) FL on the CIFAR-10 dataset.}}
    \label{fig:CIFAR10full}
    \vspace{-0.15in}
\end{figure}

\mypara{Full clients participation.} 
We see from Fig.~\ref{fig:CIFAR10full} that under the same total power budget, the $\mathcal{O}(t^2)$ power control policy performs better than the equal power allocation scheme and is very close to the noise-free ideal case. Specifically, $\mathcal{O}(t^2)$ power control policy achieves $81.1\%$ and $59.6\%$ final test accuracy in IID and non-IID data partitions on CIFAR-10, which is $2.6\%$ and $9.8\%$ better than that of the equal power allocation scheme. Note that the training loss (test accuracy) of equal power allocation scheme increases (decreases) during the late rounds ($350$th to $500$th) in the non-IID case, implying that a non-increasing SNR may occur deterioration in the convergence of FL for more difficult ML tasks. 

To further validate the $\mathcal{O}(t^2)$ scaling, we also carry out experiments where power is increased as a slower rate of $\mathcal{O}(\log(t)) $ \zixiangTCCN{ and $\mathcal{O}(t)$}. The resulting performance is much worse than the $\mathcal{O}(t^2)$ scaling, and in fact has only very limited improvement over the equal power allocation.

\congr{Lastly, we note that the early rounds of all methods have very similar performance. This is because although $\mathcal{O}(t^2)$ power control allocates less power than the equal power policy, both are dominated by the noise of SGD in early rounds and thus their performances are similar. This phenomenon is also observed in other experiments, which again highlights the benefits of adaptively ``flying under the radar'', to only allocate sufficient-but-not-excessive transmit power in each round. All of the aforementioned observations carry over to other tasks and different FL configurations.}


\begin{figure}
    \centering
    \subfigure{
        \includegraphics[width=0.23\textwidth]{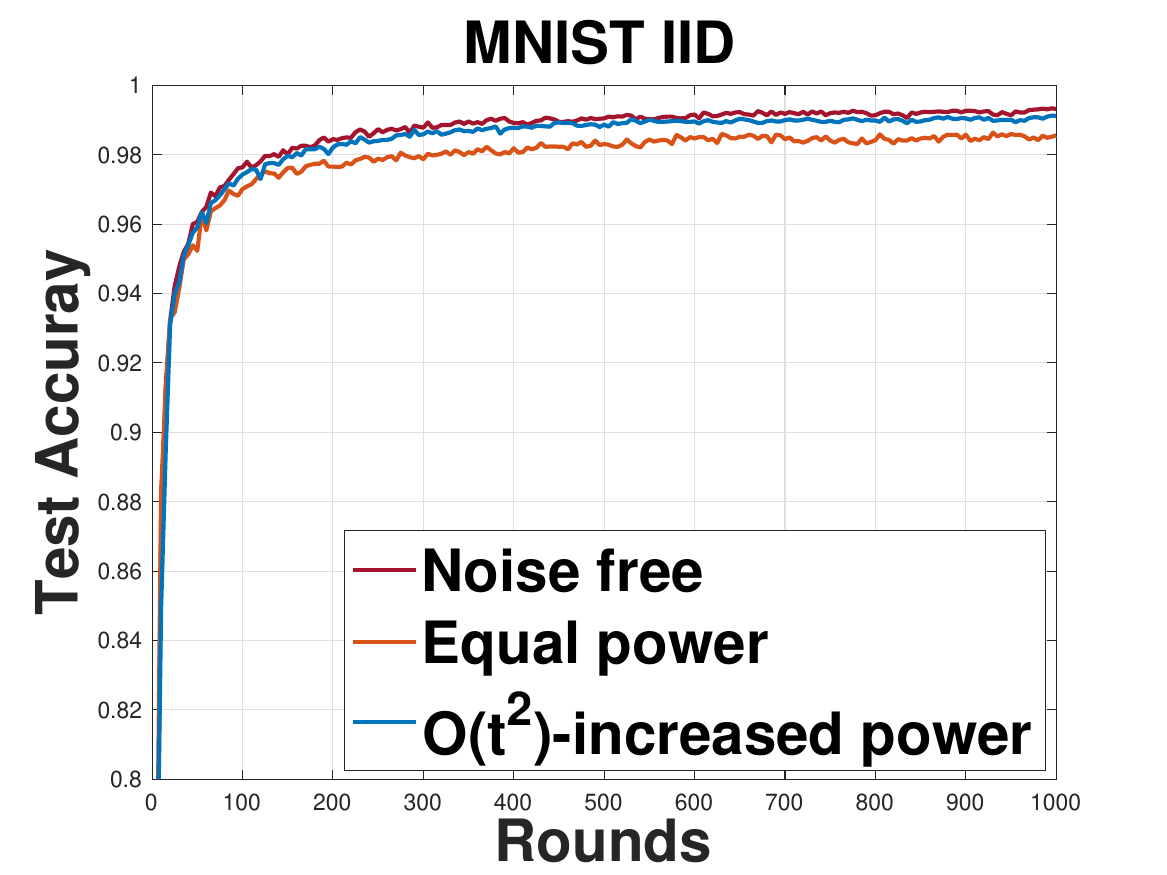}}
    \subfigure{
        \includegraphics[width=0.23\textwidth]{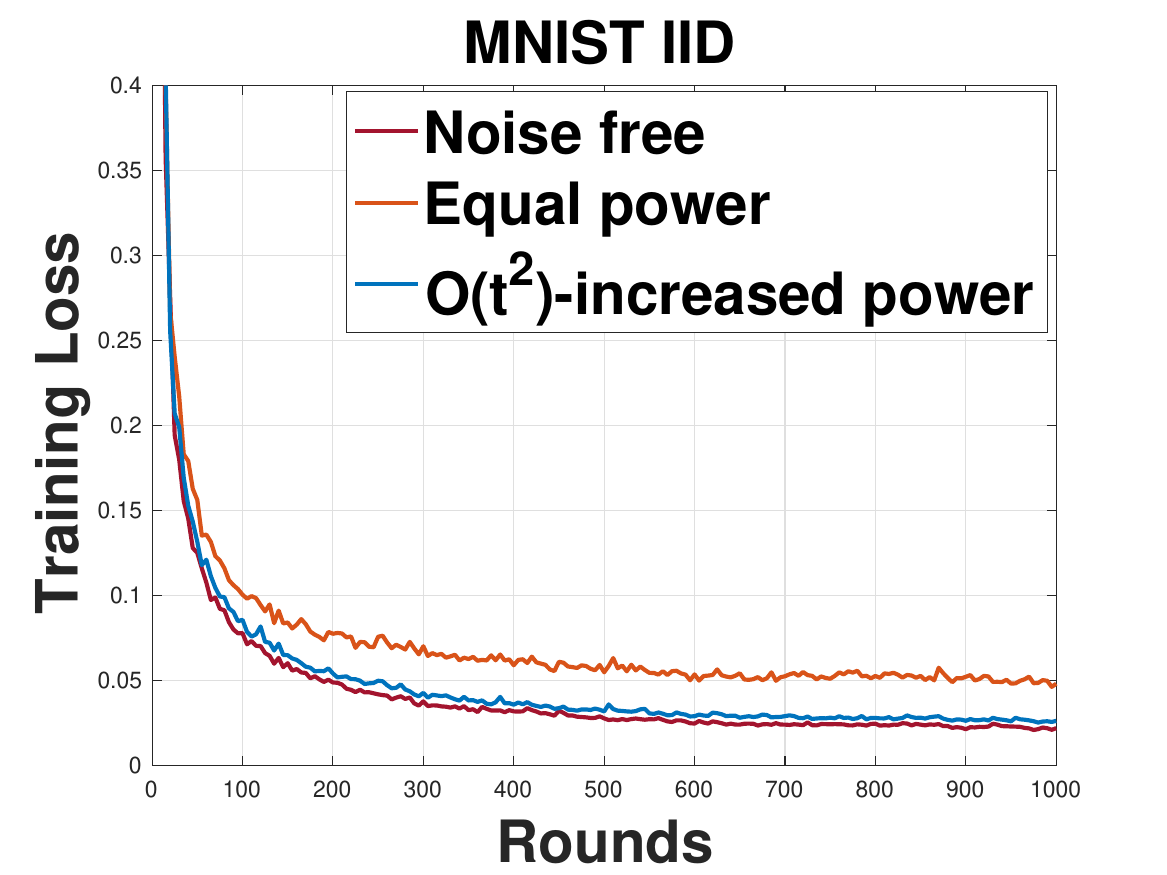}}
    \subfigure{
        \includegraphics[width=0.23\textwidth]{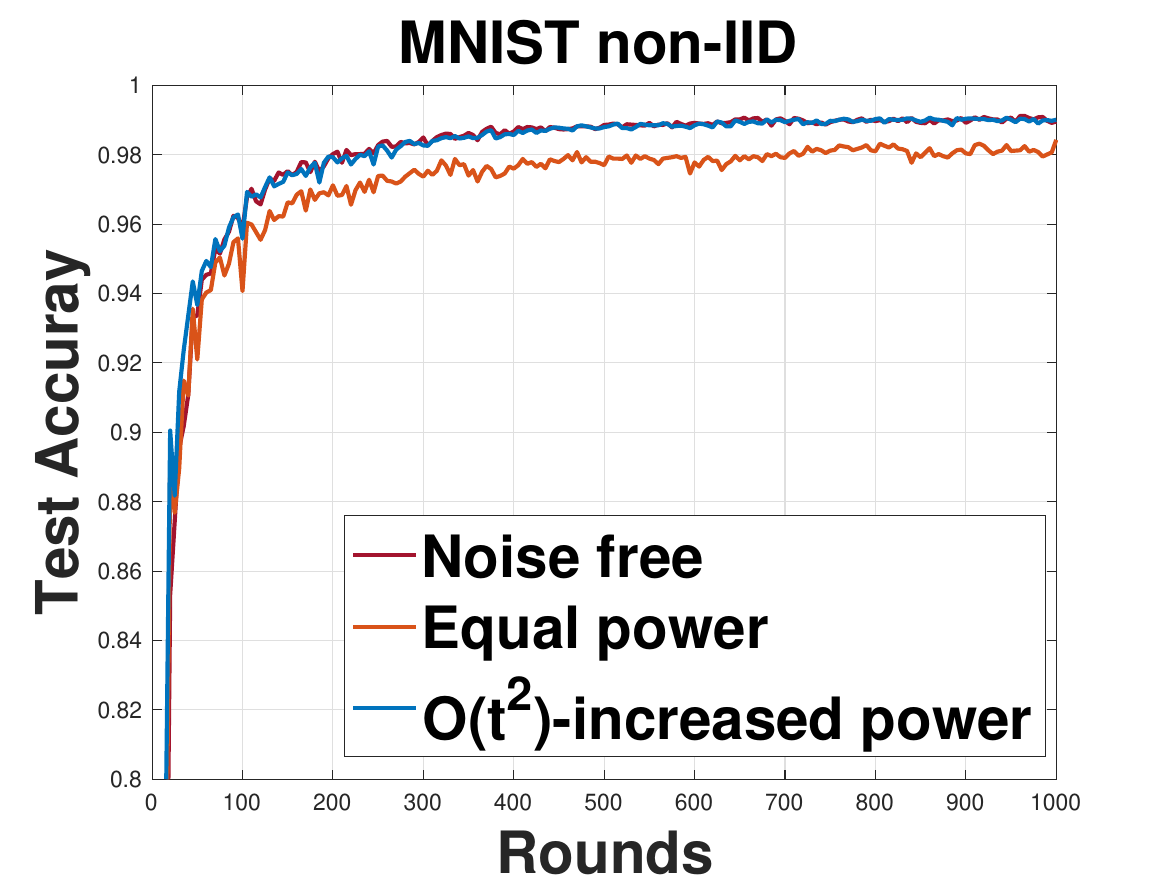}}
    \subfigure{
        \includegraphics[width=0.23\textwidth]{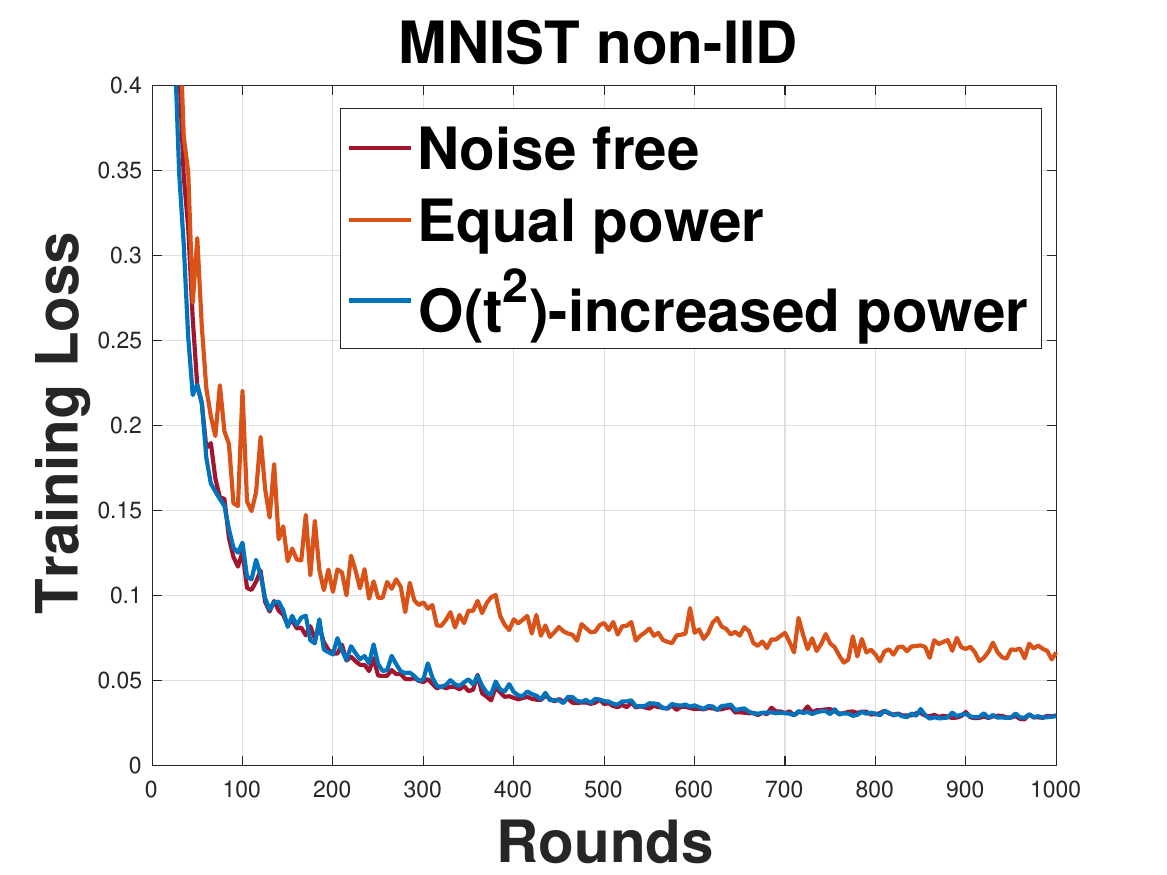}}
    \caption{Comparing the performance of transmit power control to the baselines with partial clients participation, model transmission, and both IID (left two) and non-IID (right two) FL on the MNIST dataset.}
    \label{fig:MNIST}
    \vspace{-0.15in}
\end{figure}
\begin{figure}
    \centering
    \subfigure{
        \includegraphics[width=0.23\textwidth]{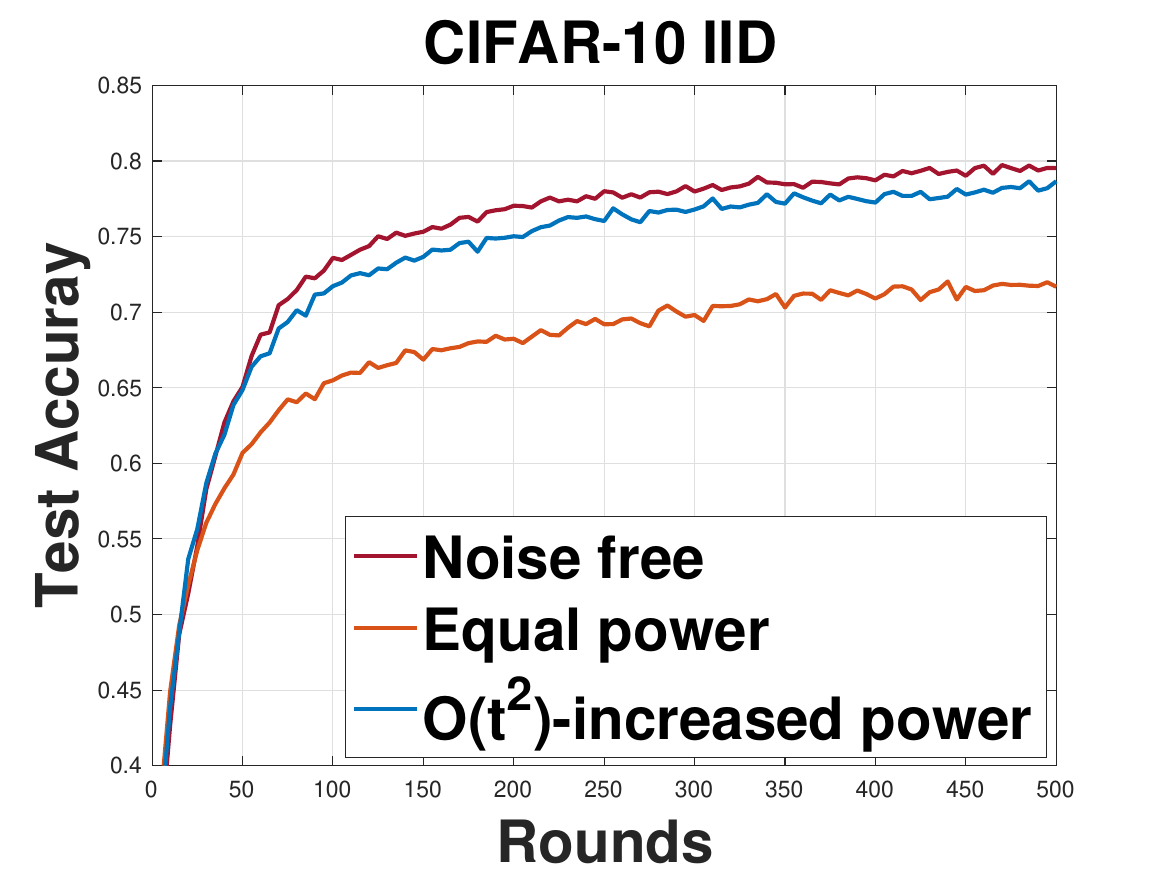}}
    \subfigure{
        \includegraphics[width=0.23\textwidth]{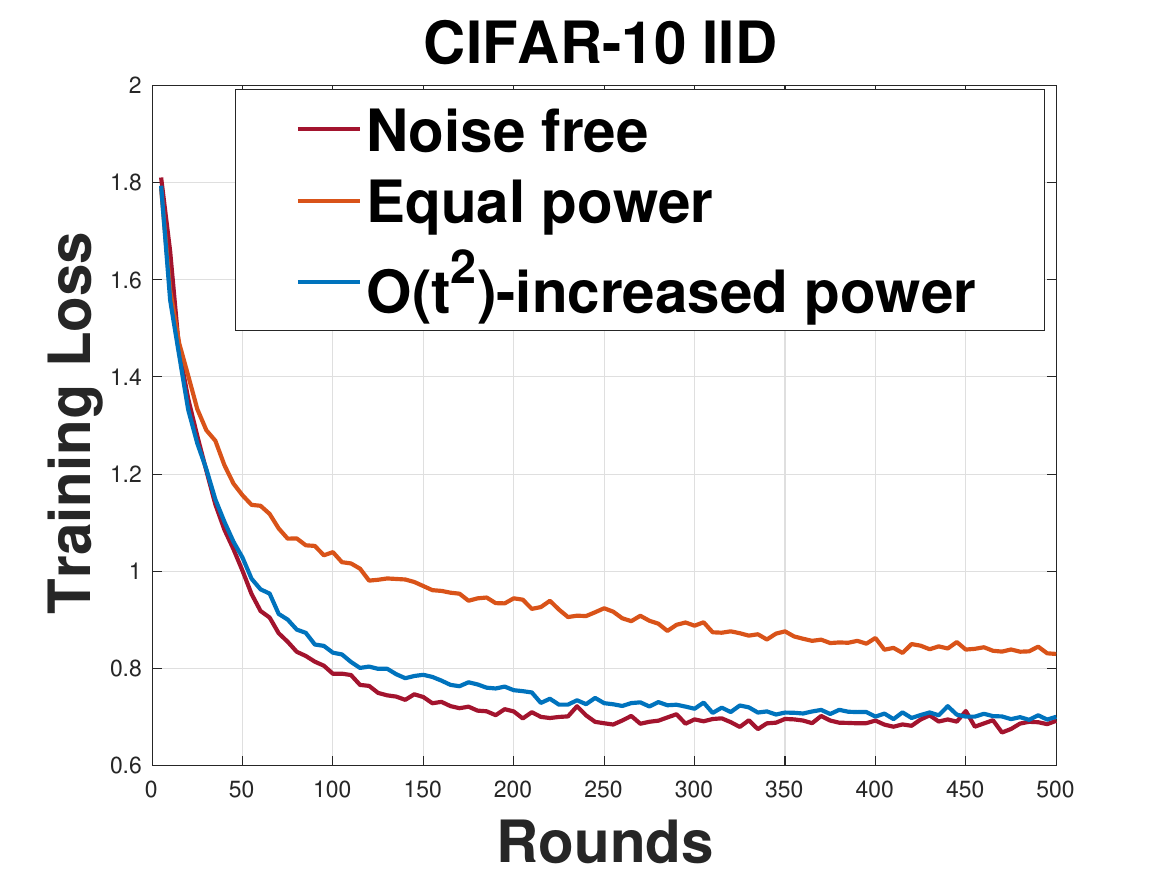}}
    \subfigure{
        \includegraphics[width=0.23\textwidth]{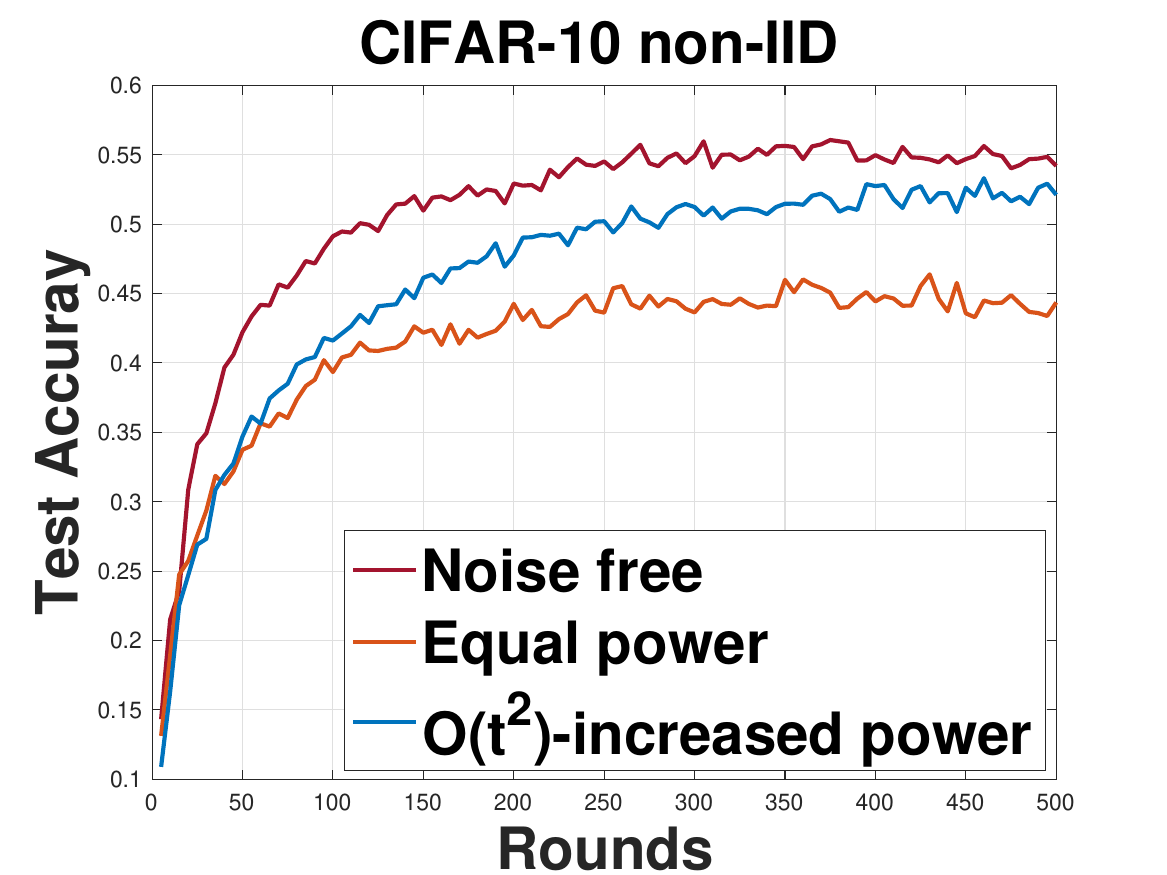}}
    \subfigure{
        \includegraphics[width=0.23\textwidth]{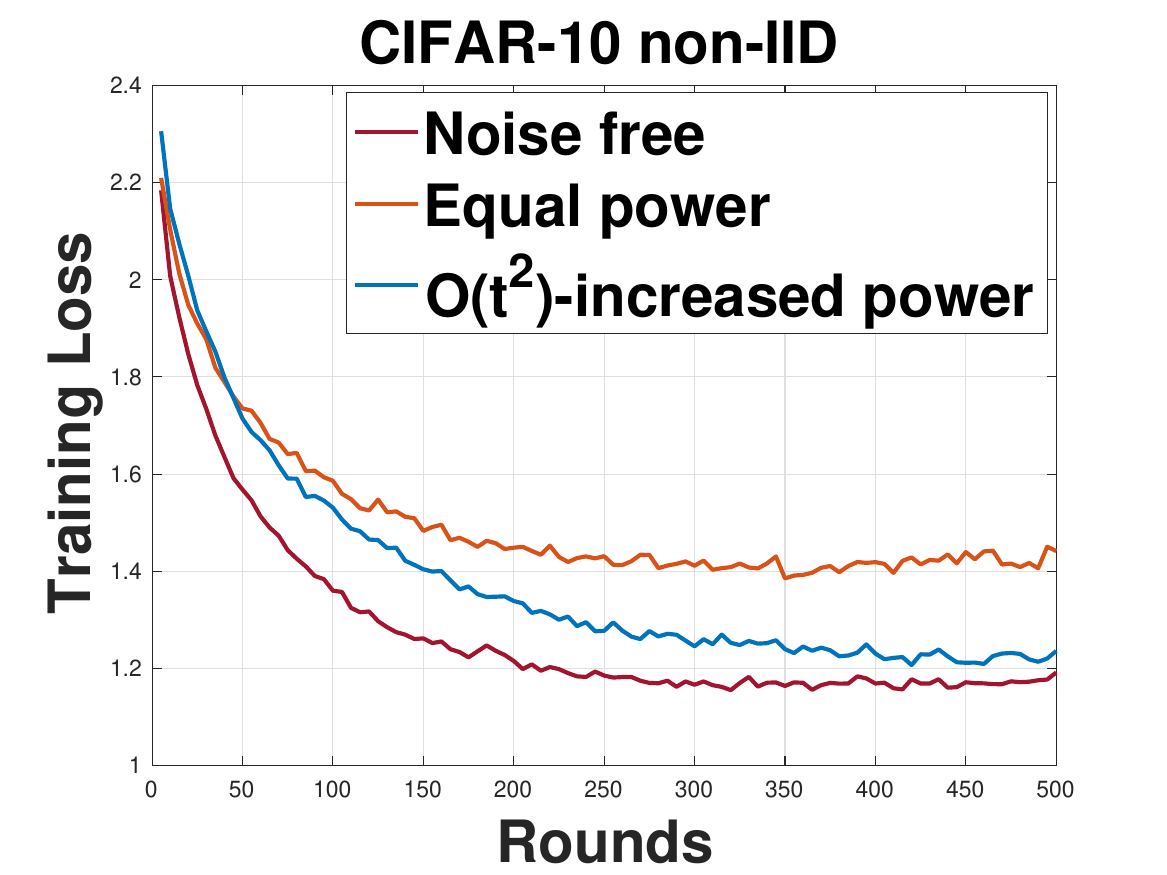}}
    \caption{Comparing the performance of transmit power control to the baselines with partial clients participation, model transmission, and both IID (left two) and non-IID (right two) FL on the CIFAR-10 dataset.}
    \label{fig:CIFAR10}
    \vspace{-0.15in}
\end{figure}
\begin{figure}
    \centering
    \subfigure{
        \includegraphics[width=0.23\textwidth]{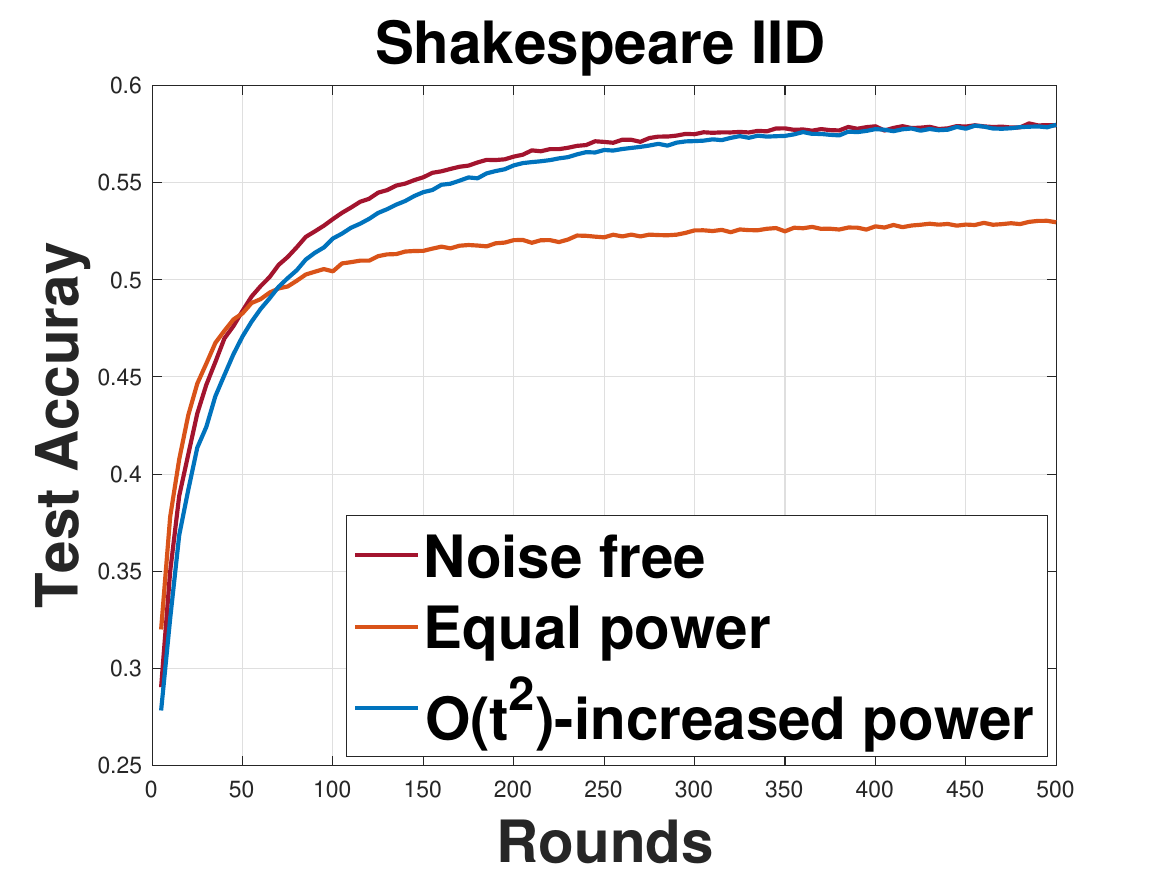}}
    \subfigure{
        \includegraphics[width=0.23\textwidth]{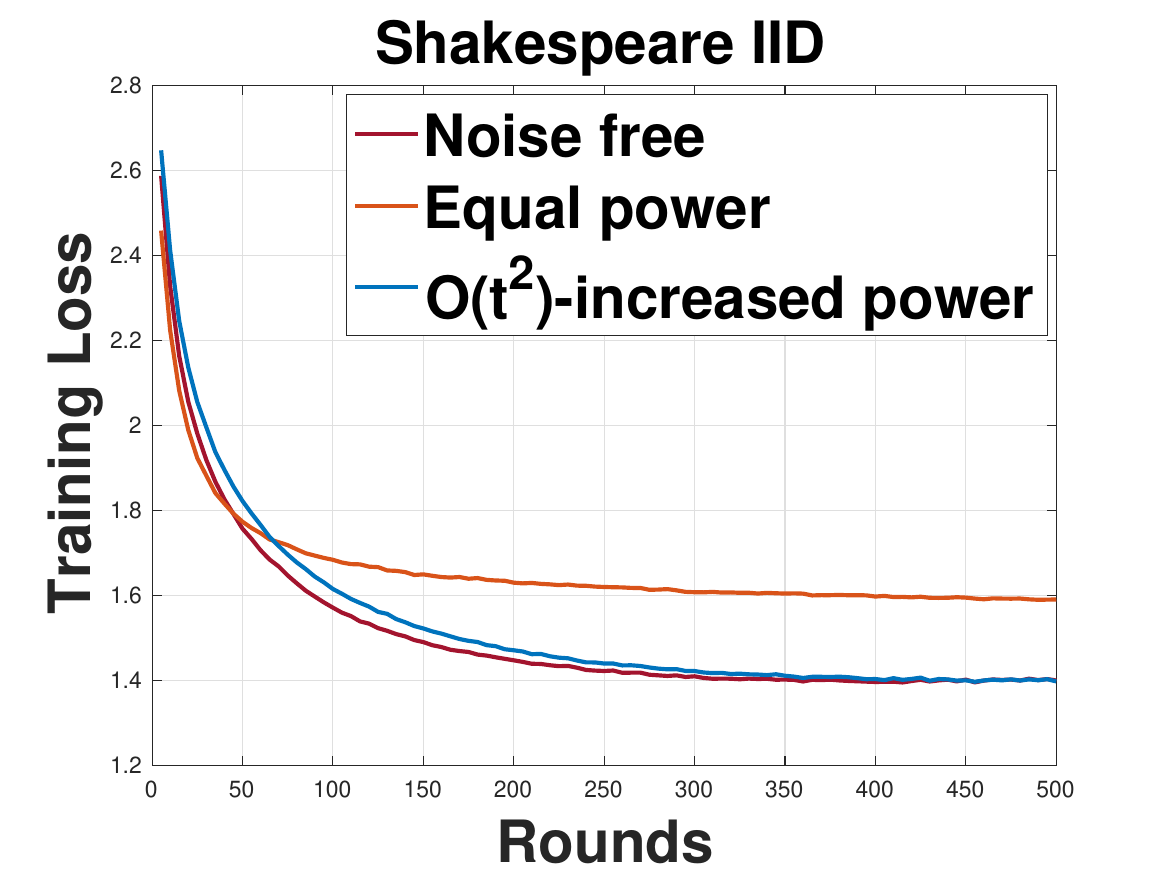}}
    \subfigure{
        \includegraphics[width=0.23\textwidth]{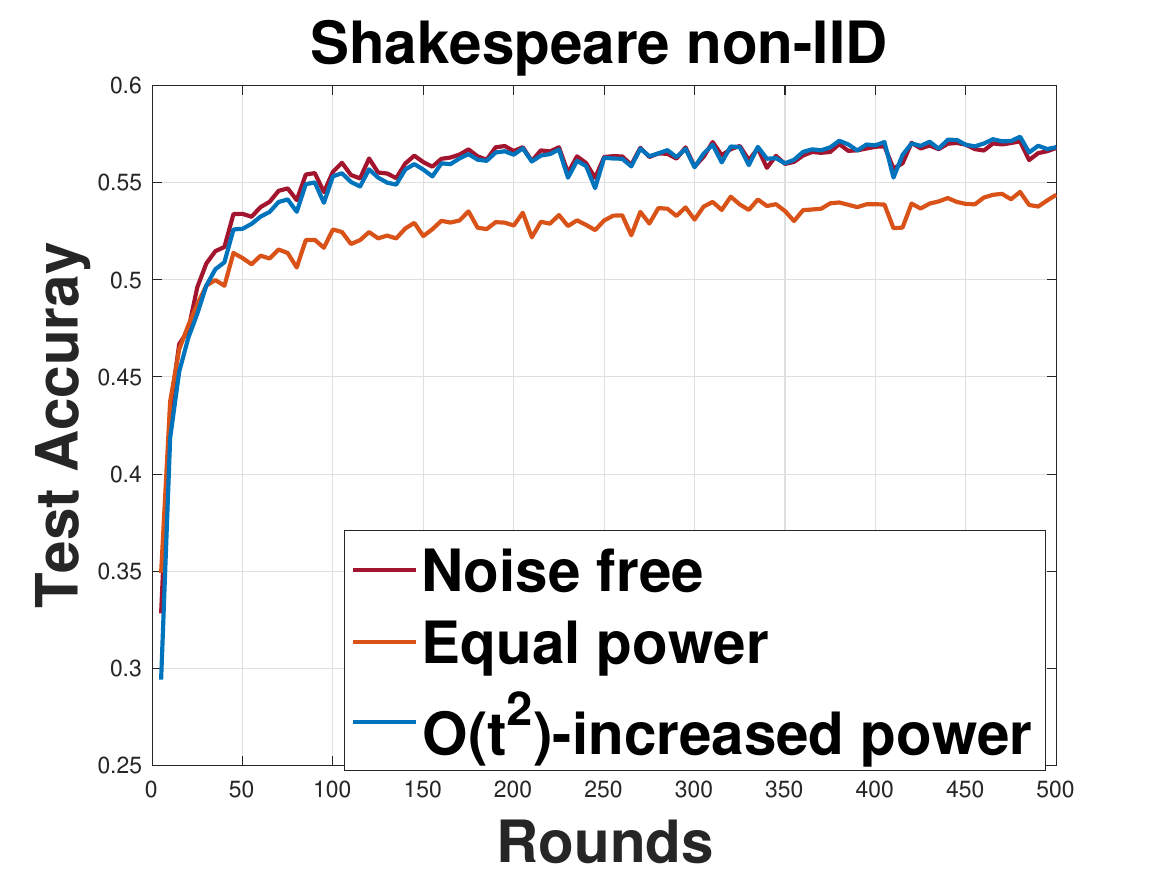}}
    \subfigure{
        \includegraphics[width=0.23\textwidth]{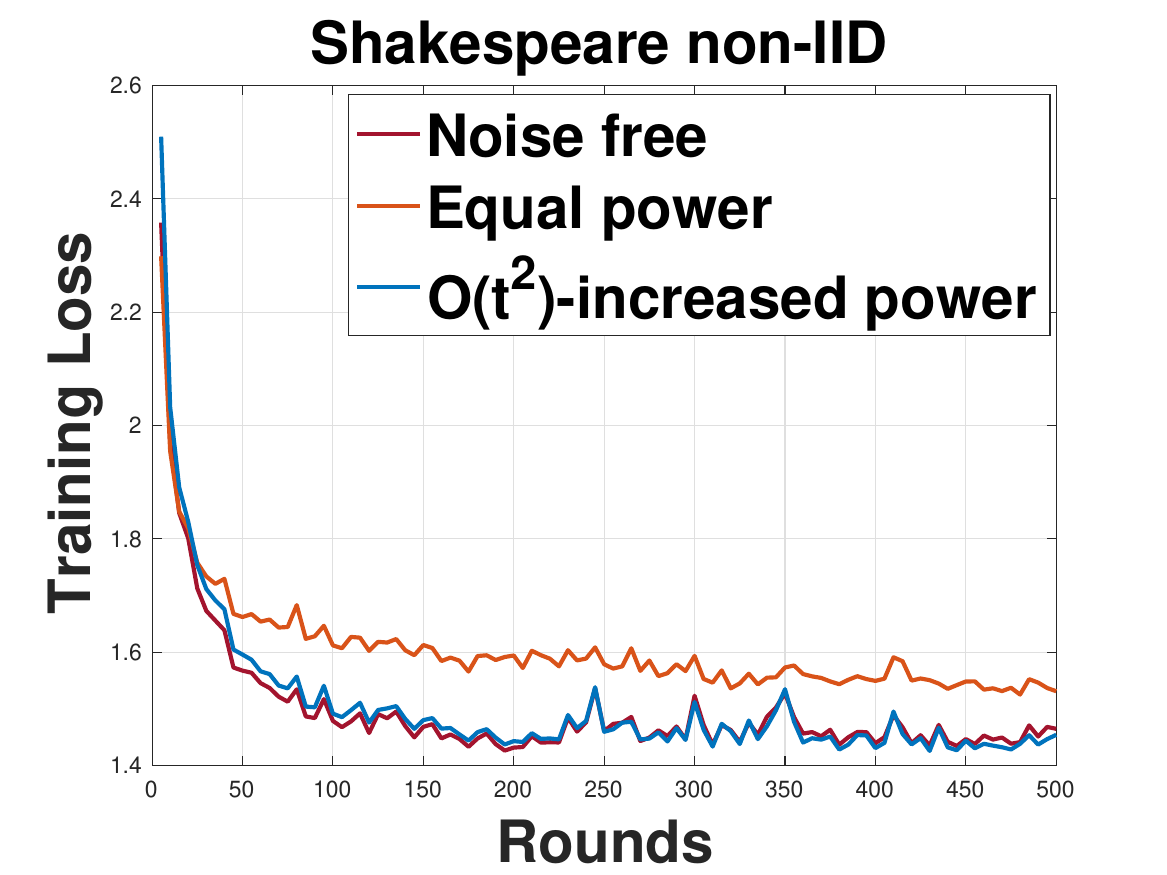}}
    \caption{Comparing the performance of transmit power control to the baselines with partial clients participation, model transmission, and both IID (left two) and non-IID (right two) FL on the Shakespeare dataset.}
    \label{fig:Shakespeare}
    \vspace{-0.15in}
\end{figure}

\mypara{Partial clients participation.} The performance comparisons of the three schemes on MNIST, CIFAR-10 and Shakespeare datasets in both IID and non-IID configurations and MT are reported in Figs.~\ref{fig:MNIST}, \ref{fig:CIFAR10}, and \ref{fig:Shakespeare}, respectively. Their final model accuracies (after $T$ rounds of FL are complete) are also summarized in Table \ref{table:accuracy}. First, we see from Fig.~\ref{fig:MNIST} that the proposed $\mathcal{O}(t^2)$-increased power allocation scheme achieves higher test accuracy and lower train loss than the equal power allocation scheme under the same energy budget on MNIST. In particular, $\mathcal{O}(t^2)$-increased power allocation scheme achieves $0.6\%$ higher test accuracy than that of equal power allocation scheme in both IID and non-IID data partitions, respectively. It may seem that the gain is insignificant, but the reason is mostly due to that MNIST classification is a very simple task. In fact, the gain of power control is much more notable under the challenging CIFAR-10 and Shakespeare tasks as shown in Figs.~\ref{fig:CIFAR10} and Fig.~\ref{fig:Shakespeare}, respectively. Compared with the equal power allocation scheme, which achieves $90.2\%$ and $81.6\%$ of the ideal (noise free) test accuracy in IID and non-IID data partitions under CIFAR-10 dataset respectively, the proposed $\mathcal{O}(t^2)$-increased power allocation achieves $99.2\%$ (IID) and $95.9\%$ (non-IID) of the ideal (noise free) test accuracy respectively after $T = 500$ communication rounds. Similarly, under Shakespeare dataset, the equal power allocation scheme achieves $91.5\%$ (IID) and $95.8\%$ (non-IID) of the ideal (noise free) test accuracy, while the proposed method improves $8.5\%$ and $3.5\%$, respectively. 

\begin{figure}
    \centering
    \subfigure{
        \includegraphics[width=0.23\textwidth]{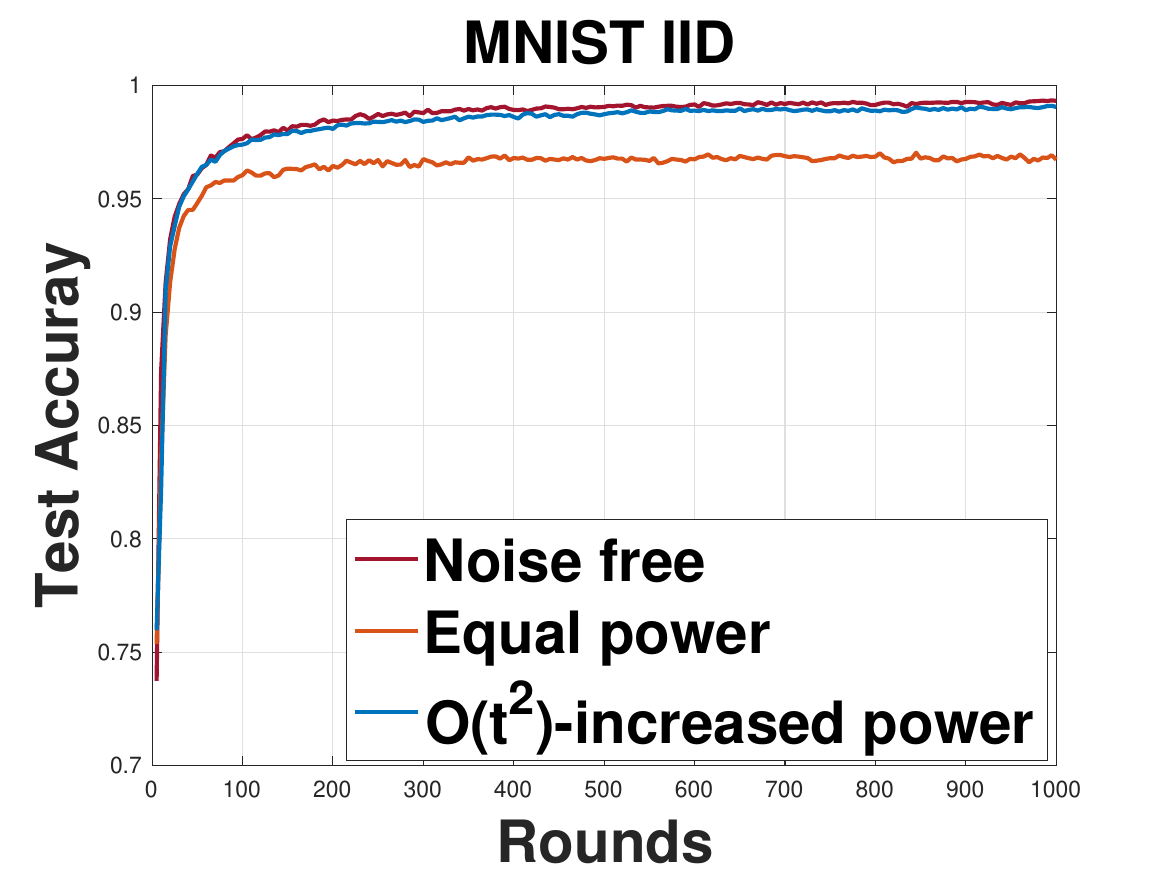}}
    \subfigure{
        \includegraphics[width=0.23\textwidth]{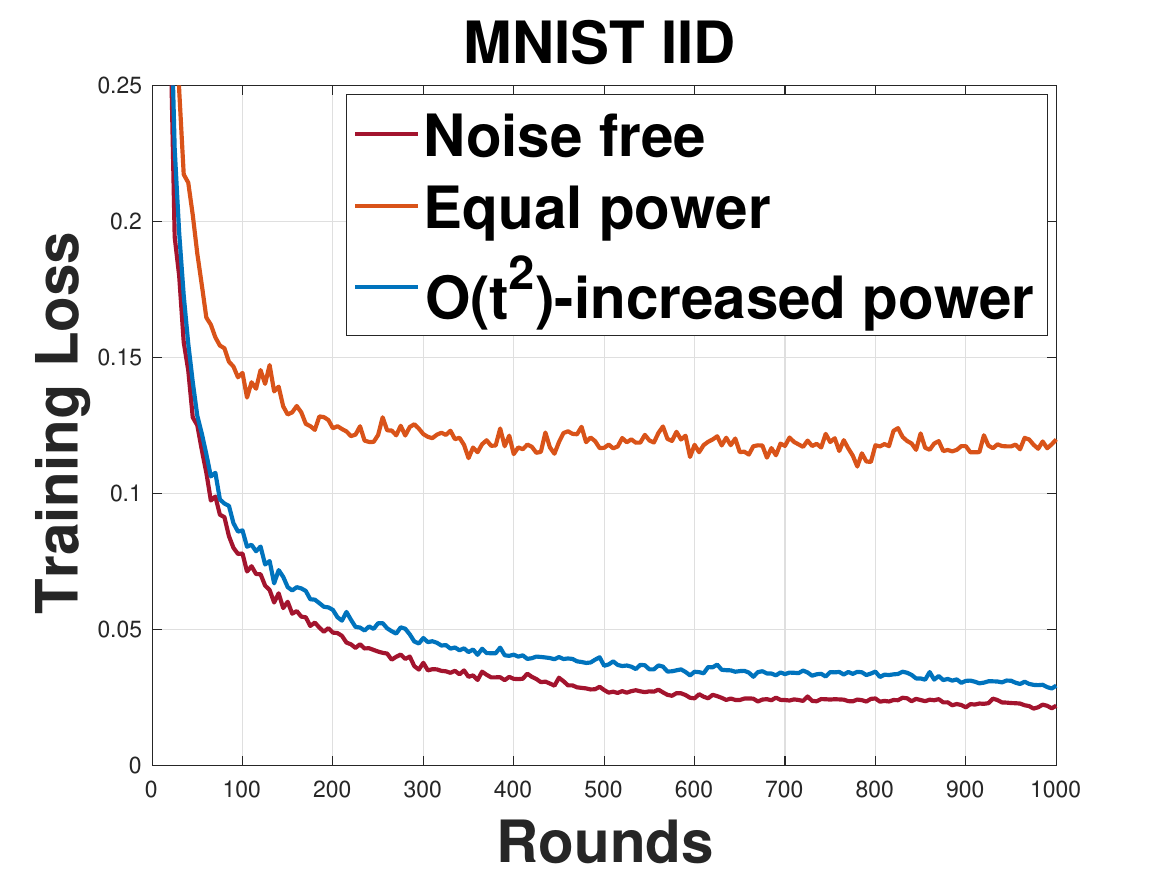}}
    \subfigure{
        \includegraphics[width=0.23\textwidth]{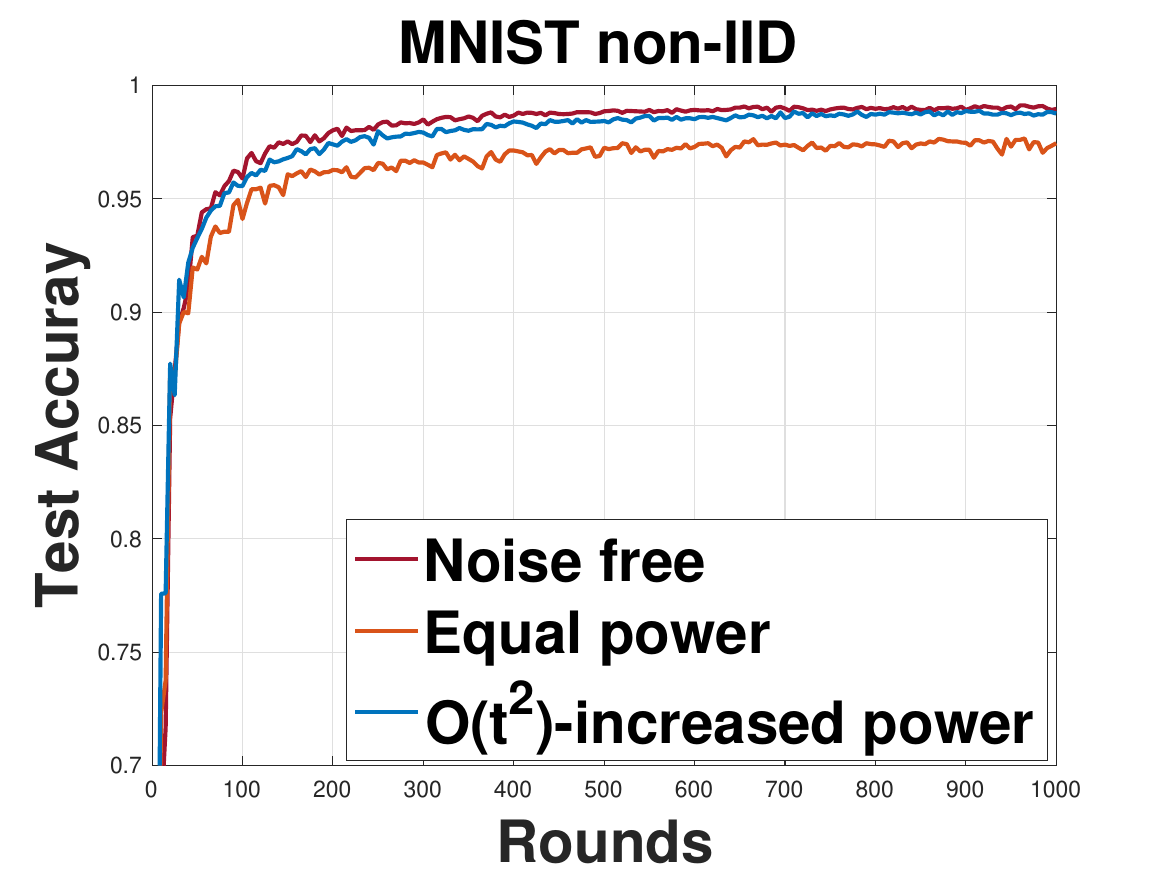}}
    \subfigure{
        \includegraphics[width=0.23\textwidth]{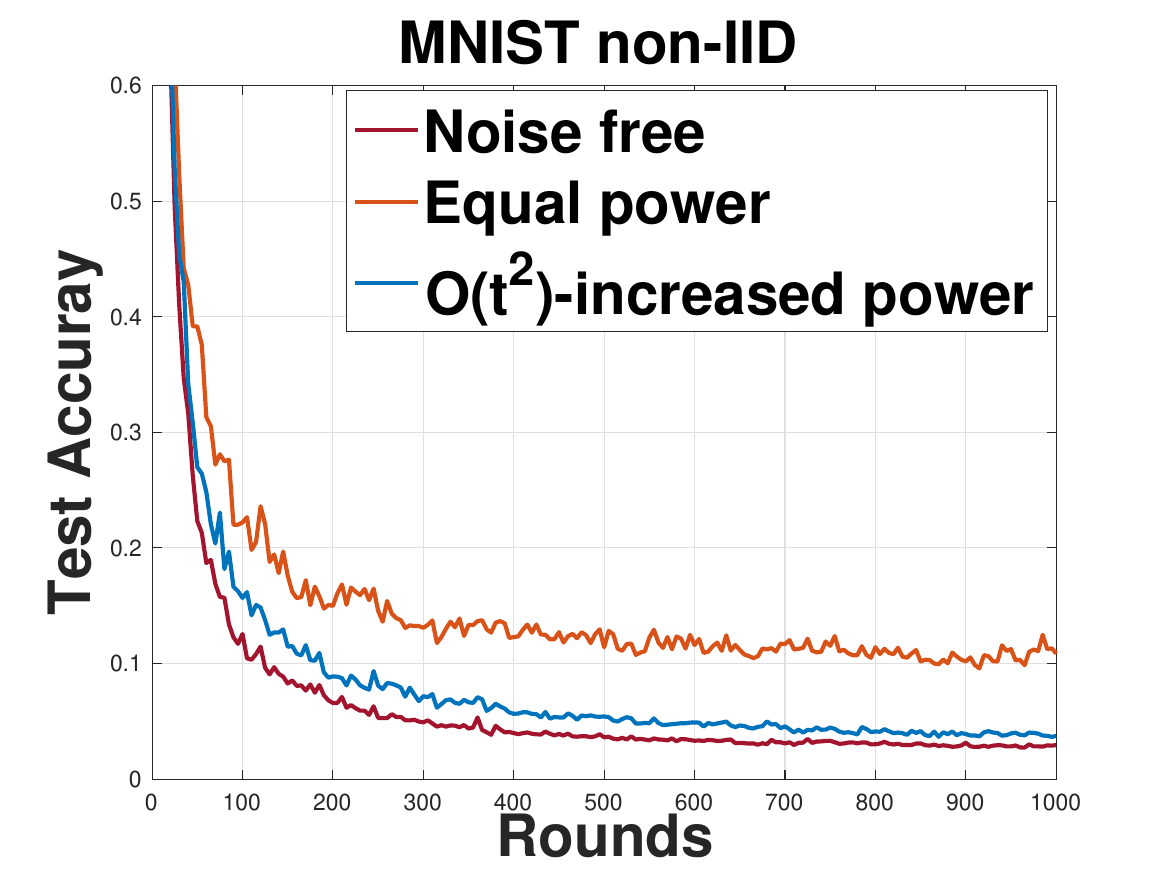}}
    \caption{Comparing the performance of transmit power control to the baselines with partial clients participation, model differential transmission, and both IID (left two) and non-IID (right two) FL on the MNIST dataset.}
    \label{fig:MNIST-DT}
    \vspace{-0.15in}
\end{figure}
\begin{figure}
    \centering
    \subfigure{
        \includegraphics[width=0.23\textwidth]{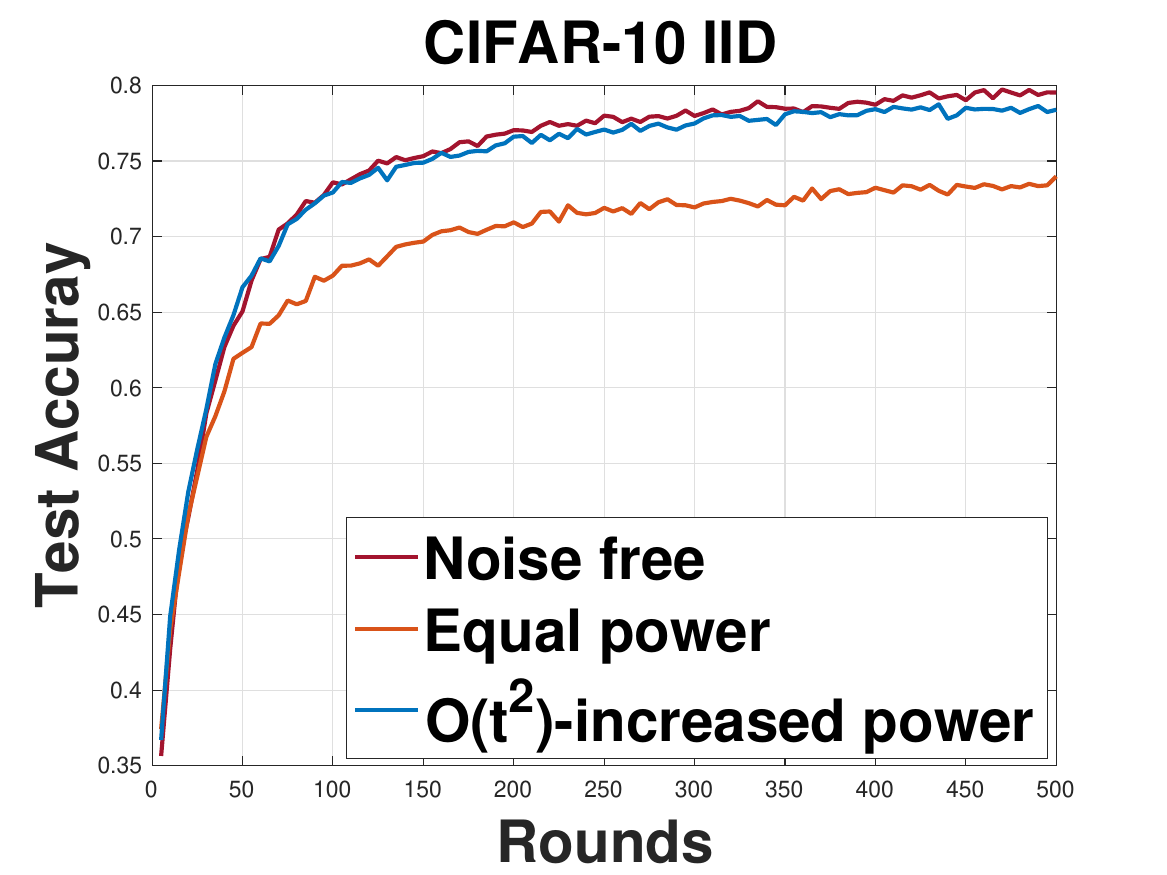}}
    \subfigure{
        \includegraphics[width=0.23\textwidth]{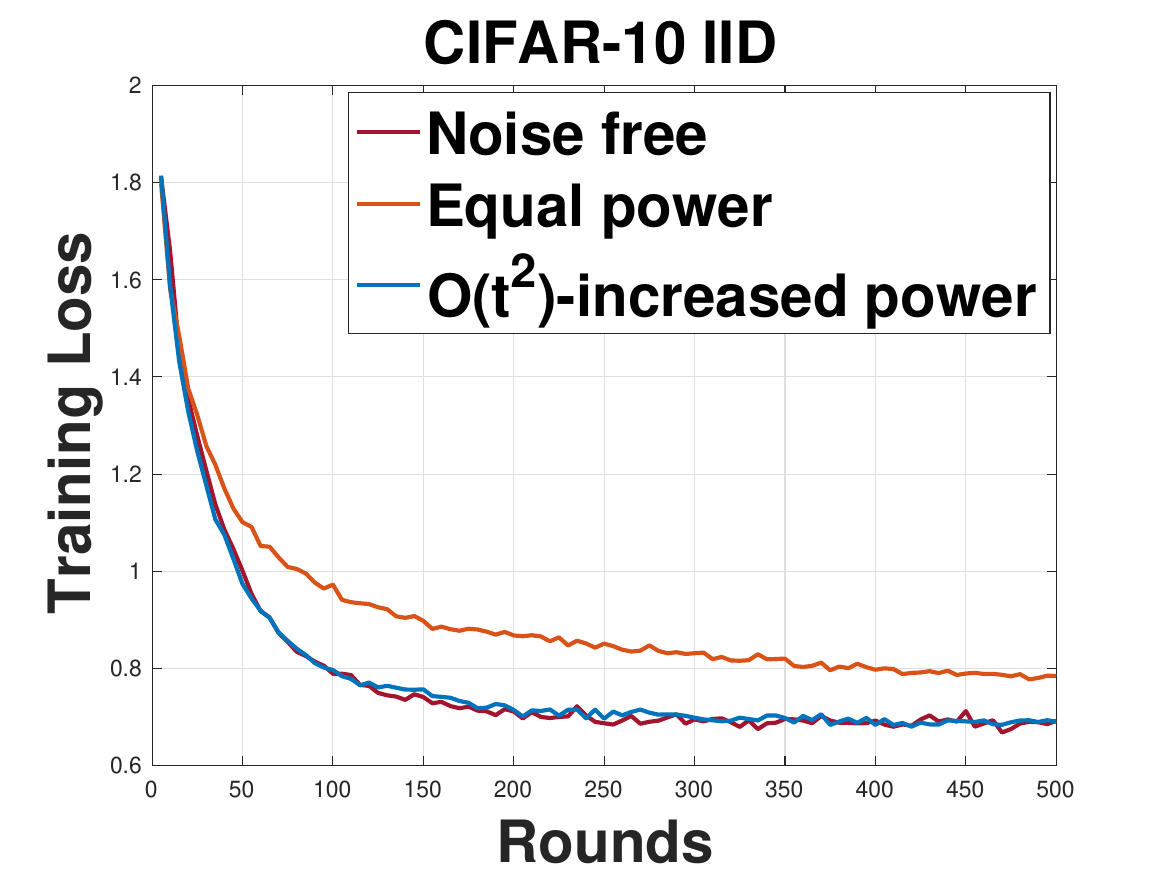}}
    \subfigure{
        \includegraphics[width=0.23\textwidth]{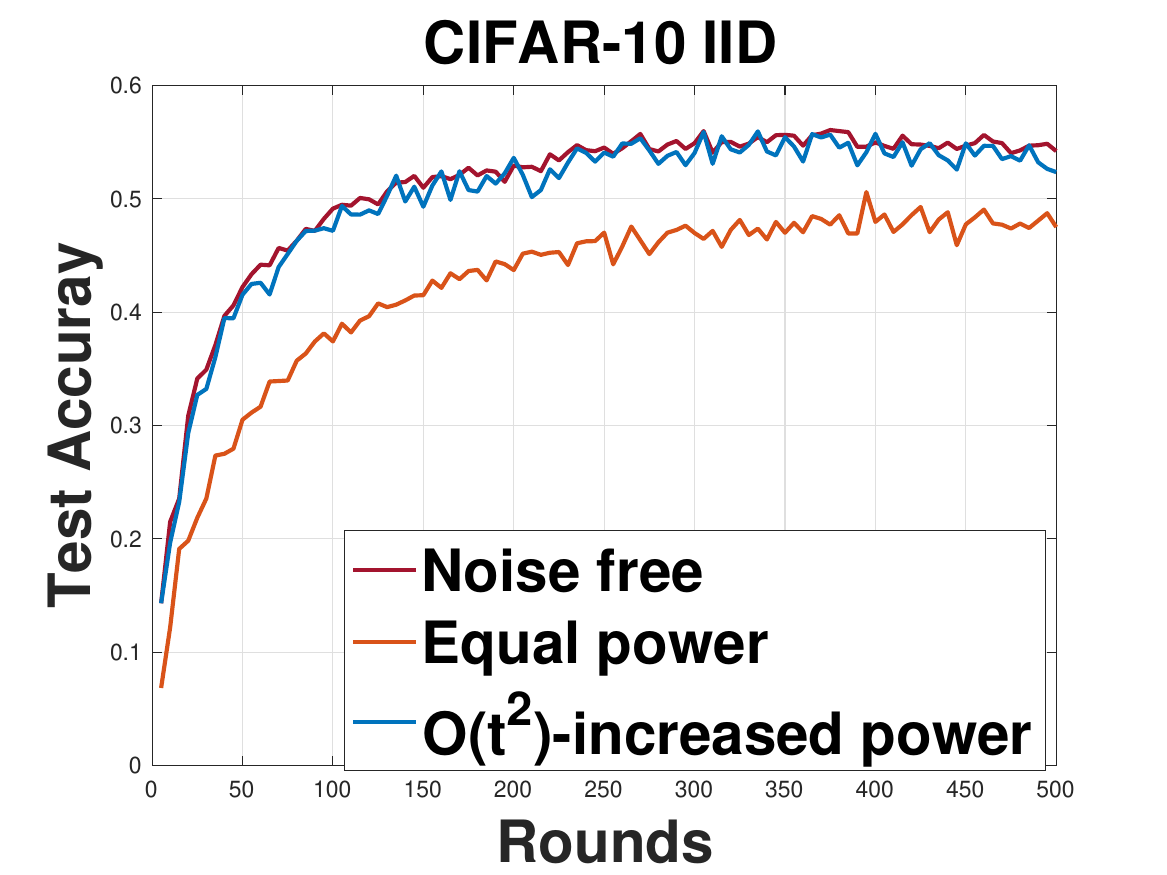}}
    \subfigure{
        \includegraphics[width=0.23\textwidth]{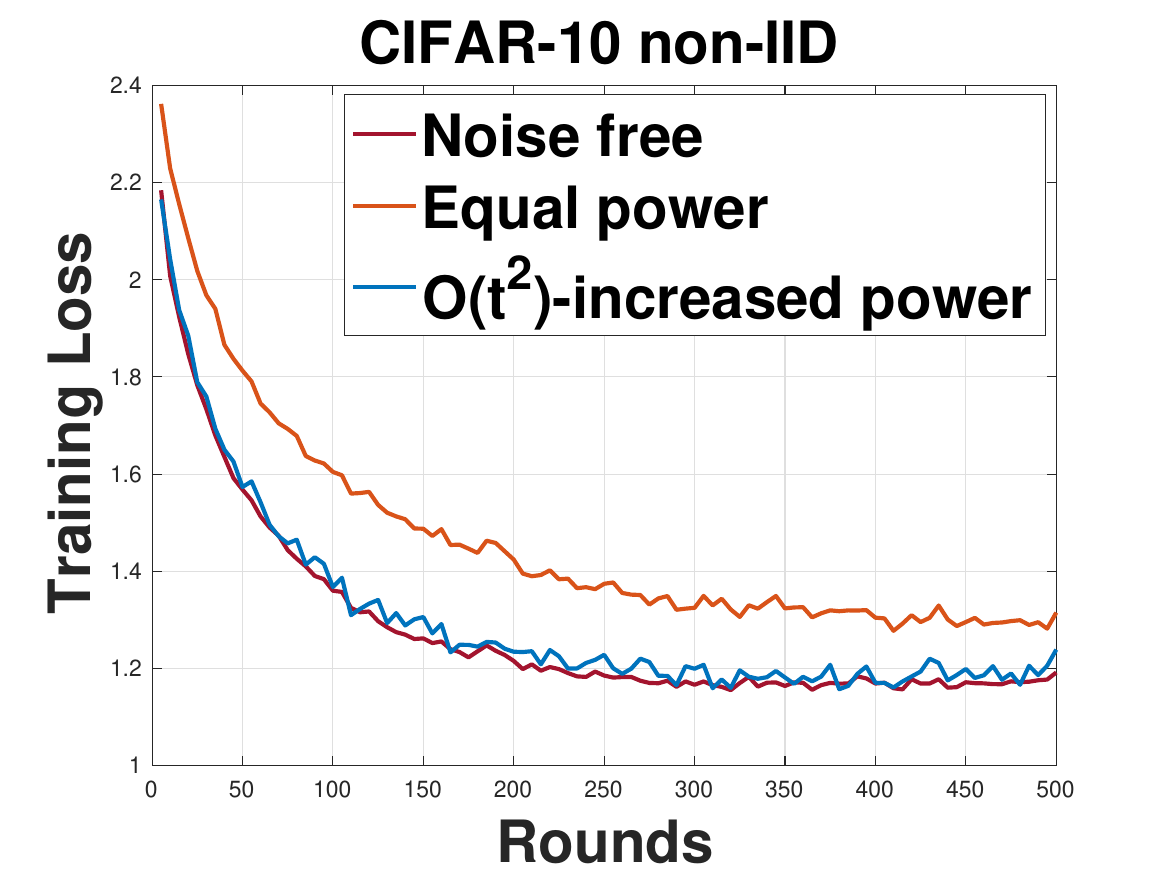}}
    \caption{Comparing the performance of transmit power control to the baselines with partial clients participation, model differential transmission, and both IID (left two) and non-IID (right two) FL on the CIFAR-10 dataset.}
    \label{fig:CIFAR10-DT}
    \vspace{-0.15in}
\end{figure}
\begin{figure}
    \centering
    \subfigure{
        \includegraphics[width=0.23\textwidth]{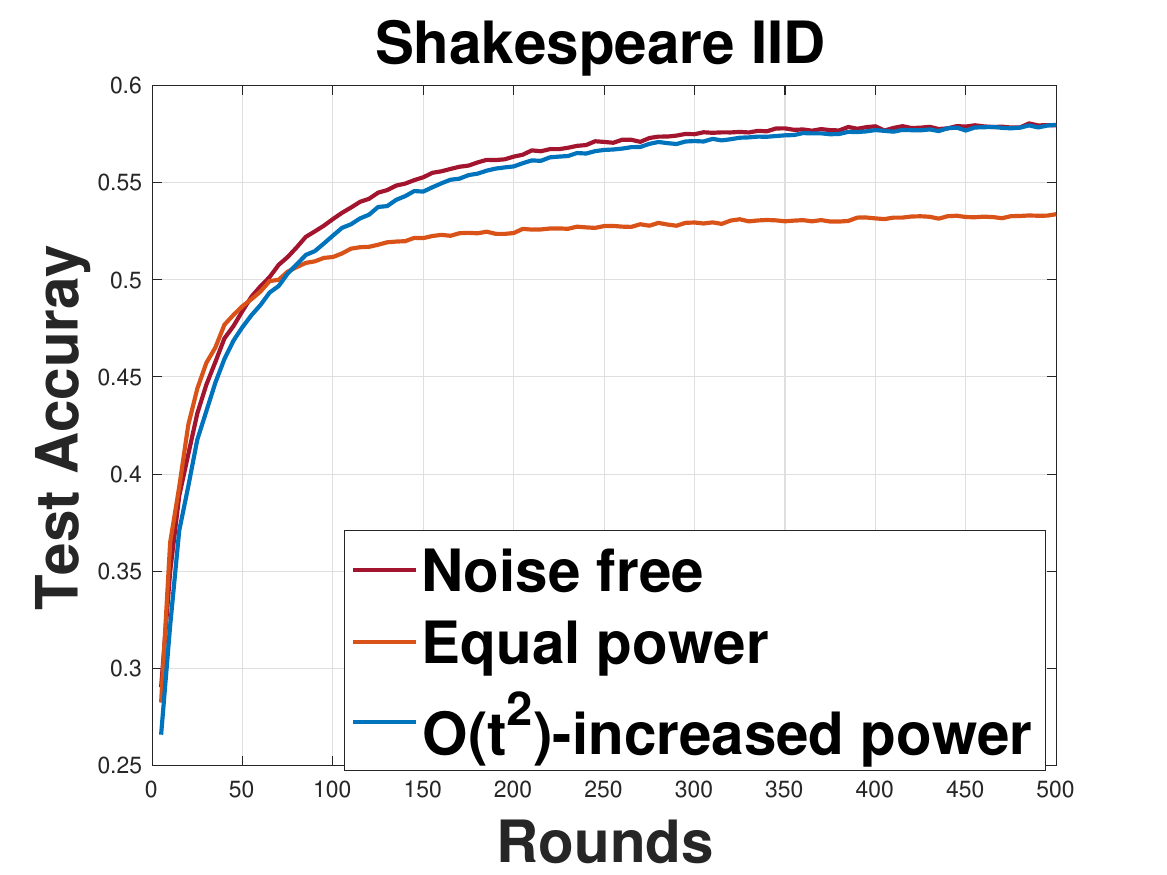}}
    \subfigure{
        \includegraphics[width=0.23\textwidth]{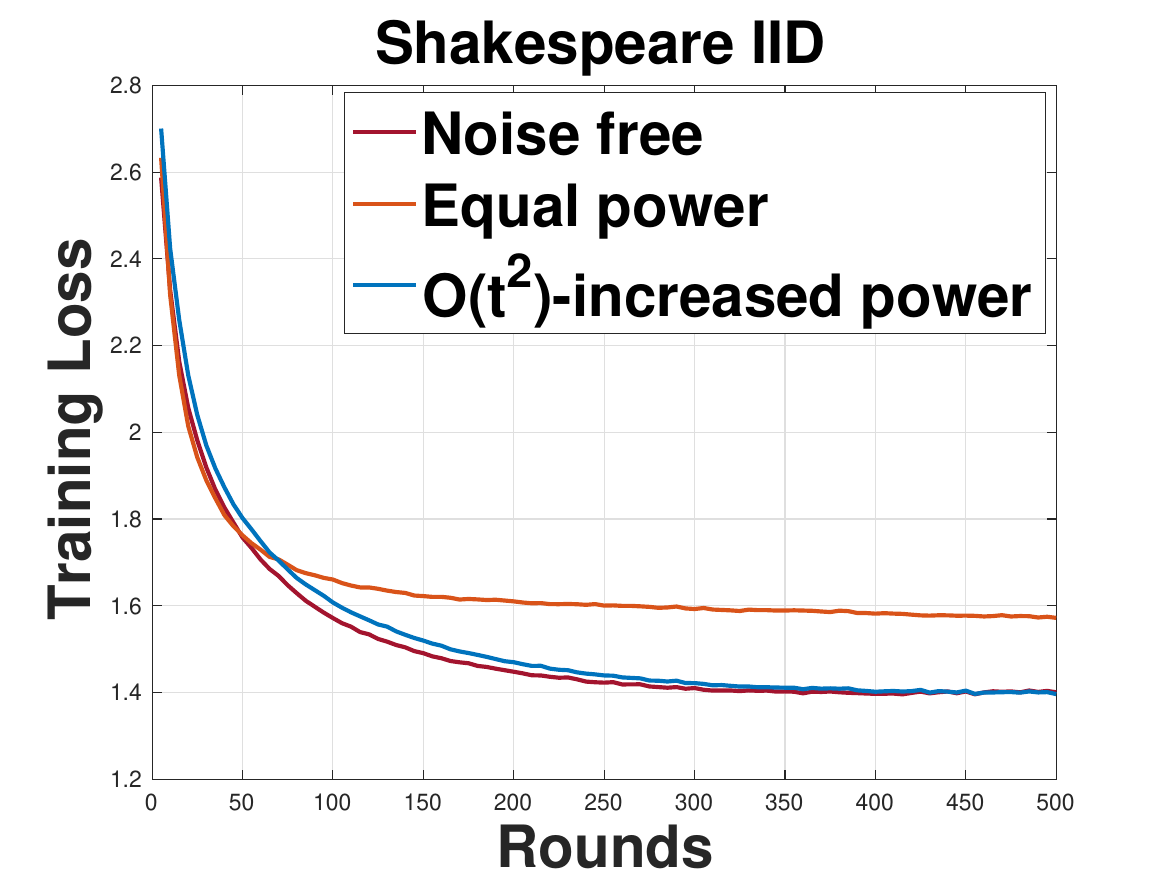}}
    \subfigure{
        \includegraphics[width=0.23\textwidth]{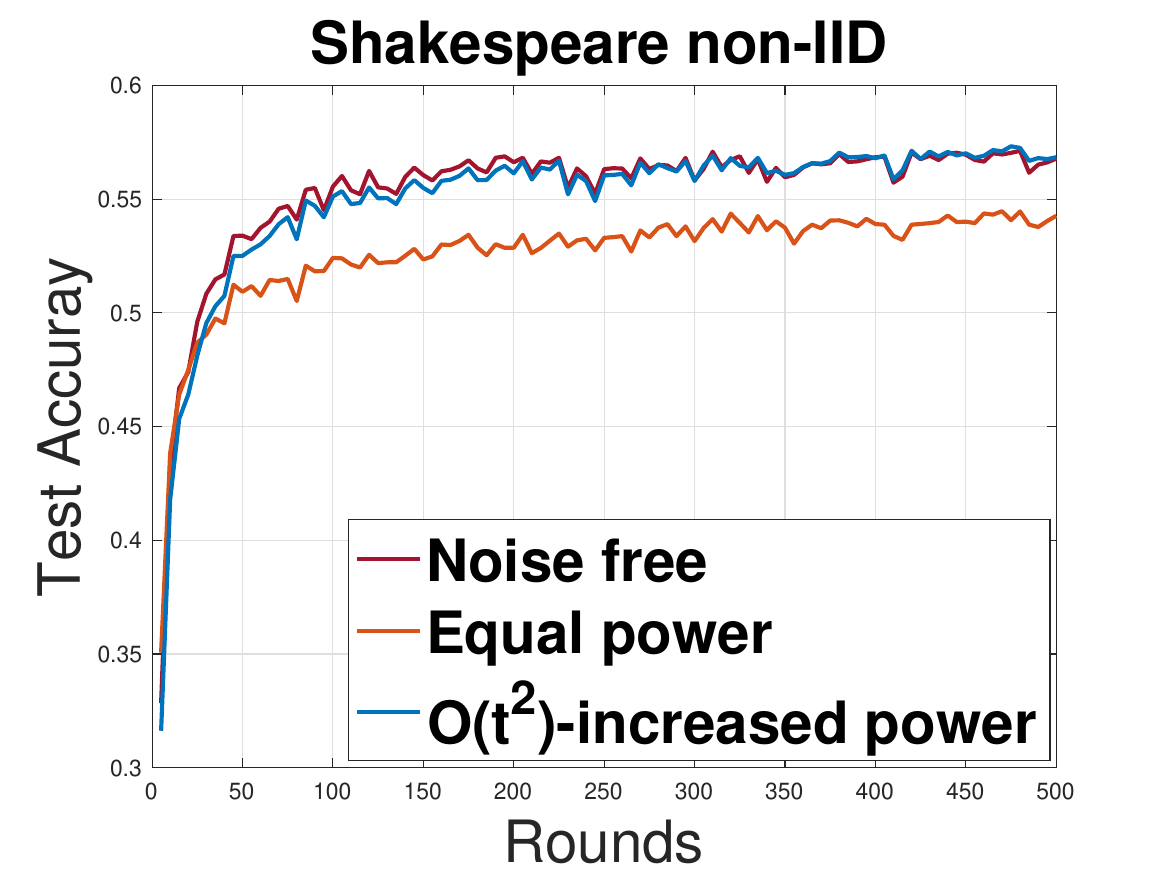}}
    \subfigure{
        \includegraphics[width=0.23\textwidth]{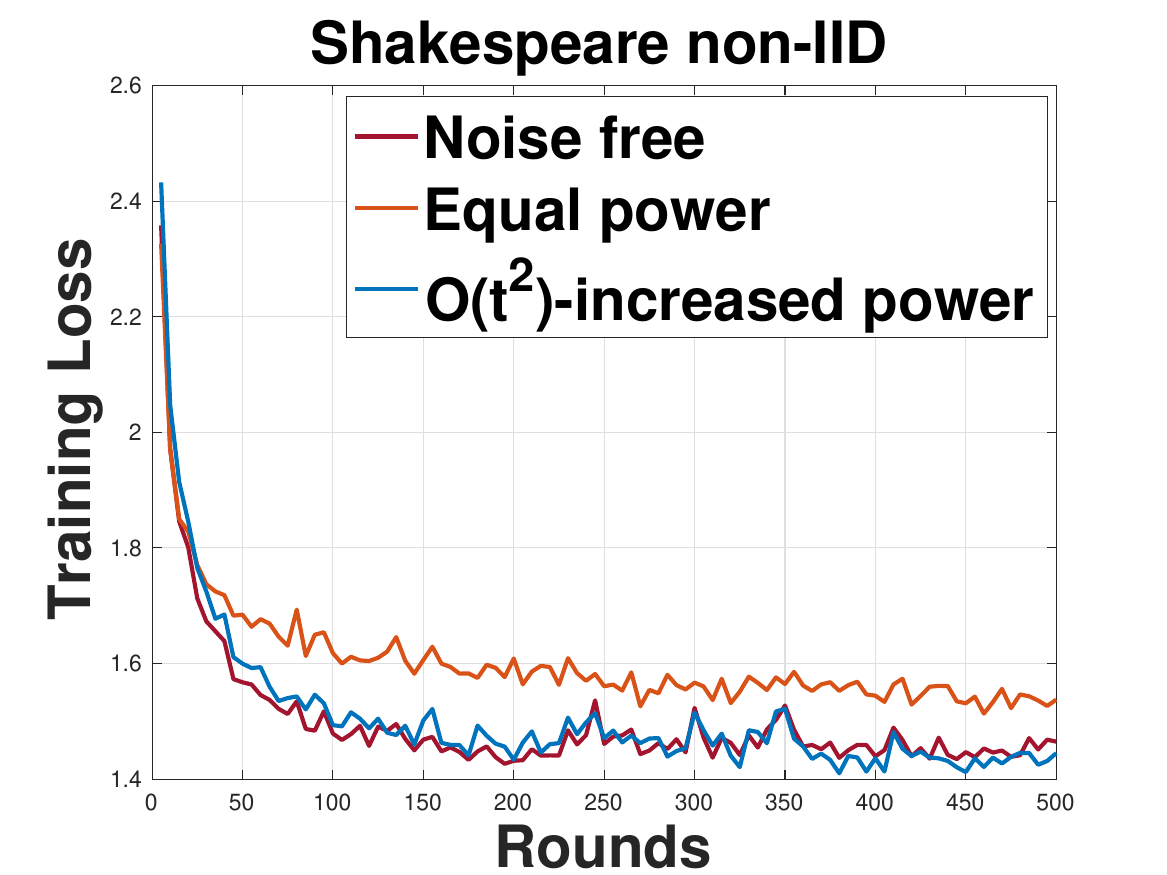}}
    \caption{Comparing the performance of transmit power control to the baselines with partial clients participation, model differential transmission, and both IID (left two) and non-IID (right two) FL on the Shakespeare dataset.}
    \label{fig:Shakespeare-DT}
    \vspace{-0.15in}
\end{figure}

\mypara{MDT.} We next present the experiment results of model differential transmission. Note that, by applying MDT, the uplink transmission power of the proposed scheme remains constant (recall that SNR is set as 10dB) while the downlink transmission power still increases at the rate of $\mathcal{O}(t^2)$. Figs.~\ref{fig:MNIST-DT}, \ref{fig:CIFAR10-DT} and \ref{fig:Shakespeare-DT} illustrate the test accuracies and training losses with MDT under MNIST, CIFAR-10 and Shakespeare datasets and the final model accuracies of the three schemes are summarized in Table \ref{table:accuracy-dt}. We see that the proposed power control policy achieves $99.7\%$ ($99.7\%$), $99.2\%$ ($98.0\%$) and $100\%$ ($98.9\%$) of the ideal test accuracy in IID (non-IID) data setting under MNIST, CIFAR-10 and Shakespeare datasets, respectively, which significantly outperforms the baseline equal power allocation scheme.


\begin{table}
\parbox{.45\linewidth}{
\caption{\small \zixiangTCCN{Performance Summary of MT}}
\label{table:accuracy}
\centering
\resizebox{\linewidth}{!}{%
\begin{tabular}{cccccc}
\hline
Dataset     & Scheme & Accuracy & Percentage* & Accuracy & Percentage*        \\ \hline
\multicolumn{2}{c}{}& \multicolumn{2}{c}{IID} & \multicolumn{2}{c}{non-IID}                                                       \\ \hline
MNIST       & Noise free          & 99.3\%  & 100\%        & 99.1\%    & 100\%\\
            & Increased power     & 99.1\%  & 99.8\%      & 99.0\%    & 99.9\%\\
            & Equal power         & 98.5\%  & 99.2\%      & 98.4\%    & 99.3\% \\ \hline
CIFAR-10    & Noise free          & 79.5\%  & 100\%        & 54.3\%    & 100\%\\
            & Increased power     & 78.9\%  & 99.2\%      & 52.1\%    & 95.9\%\\
            & Equal power         & 71.7 \% & 90.2\%      & 44.3\%    & 81.6\%\\ \hline
Shakespeare & Noise free          & 57.8\%  & 100\%        & 56.8\%    & 100\%\\
            & Increased power     & 57.8\%  & 100\%      & 56.4\%    & 99.3\%\\
            & Equal power         & 52.9 \% & 91.5\%      & 54.4\%    & 95.8\%\\ \hline
\end{tabular}
}%
}
\hfill
\parbox{.45\linewidth}{
\caption{\small \zixiangTCCN{Performance Summary of MDT}}
\label{table:accuracy-dt}
\centering
\resizebox{\linewidth}{!}{%
\begin{tabular}{cccccc}
\hline
Dataset     & Scheme & Accuracy & Percentage* & Accuracy & Percentage*        \\ \hline
\multicolumn{2}{c}{}& \multicolumn{2}{c}{IID} & \multicolumn{2}{c}{non-IID}                                                       \\ \hline
MNIST       & Noise free          & 99.3\%  & 100\%        & 99.1\%    & 100\%\\
            & Increased power     & 99.0\%  & 99.7\%      & 98.8\%    & 99.7\%\\
            & Equal power         & 96.7\%  & 97.4\%      & 97.5\%    & 98.4\% \\ \hline
CIFAR-10    & Noise free          & 79.5\%  & 100\%        & 54.3\%    & 100\%\\
            & Increased power     & 78.9\%  & 99.2\%      & 53.2\%    & 98.0\%\\
            & Equal power         & 73.9 \% & 93.0\%      & 47.7\%    & 87.8\%\\ \hline
Shakespeare & Noise free          & 57.8\%  & 100\%        & 56.8\%    & 100\%\\
            & Increased power     & 57.8\%  & 100\%      & 56.2\%    & 98.9\%\\
            & Equal power         & 53.3 \% & 92.2\%      & 54.3\%    & 95.6\%\\ \hline
\end{tabular}
}%
}
\end{table}
\subsection{Experiment Results for Receive Diversity Combining}

\begin{figure}
    \centering
    \subfigure{
        \includegraphics[width=0.23\textwidth]{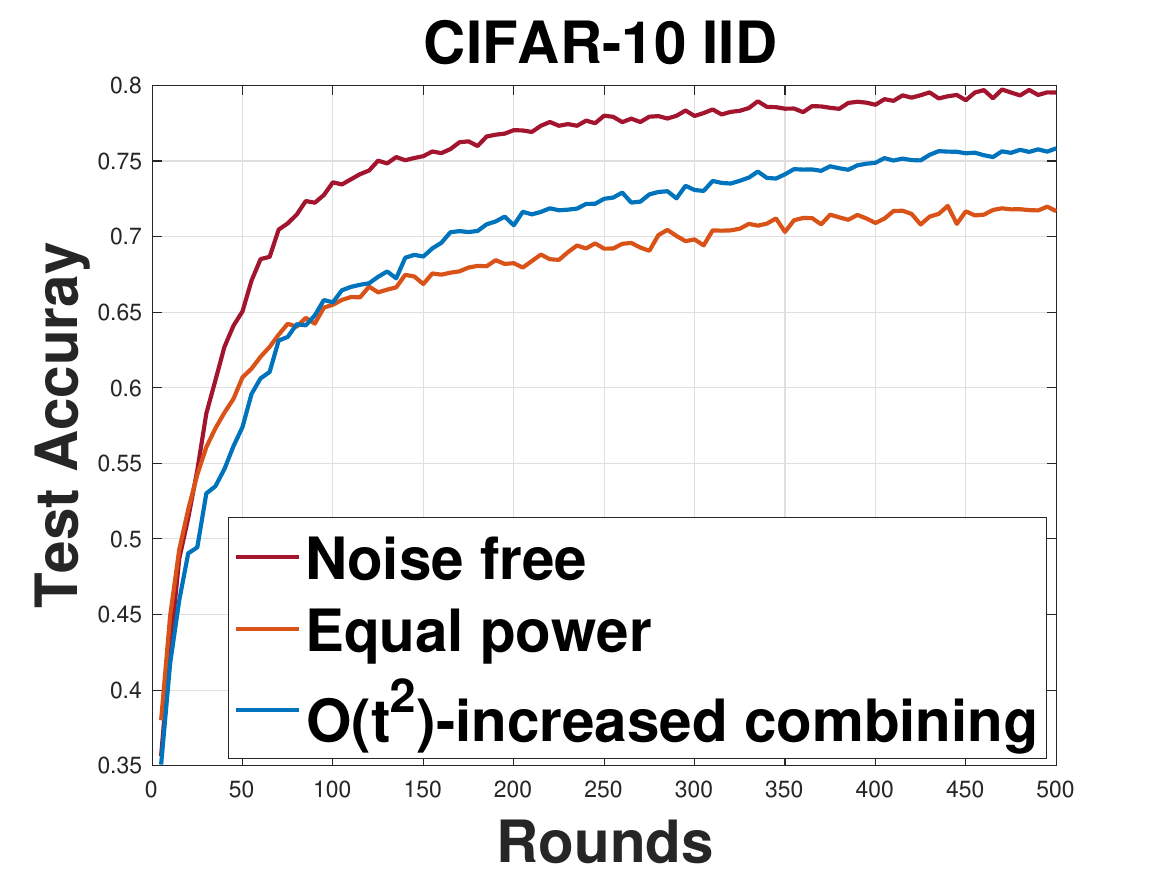}}
    \subfigure{
        \includegraphics[width=0.23\textwidth]{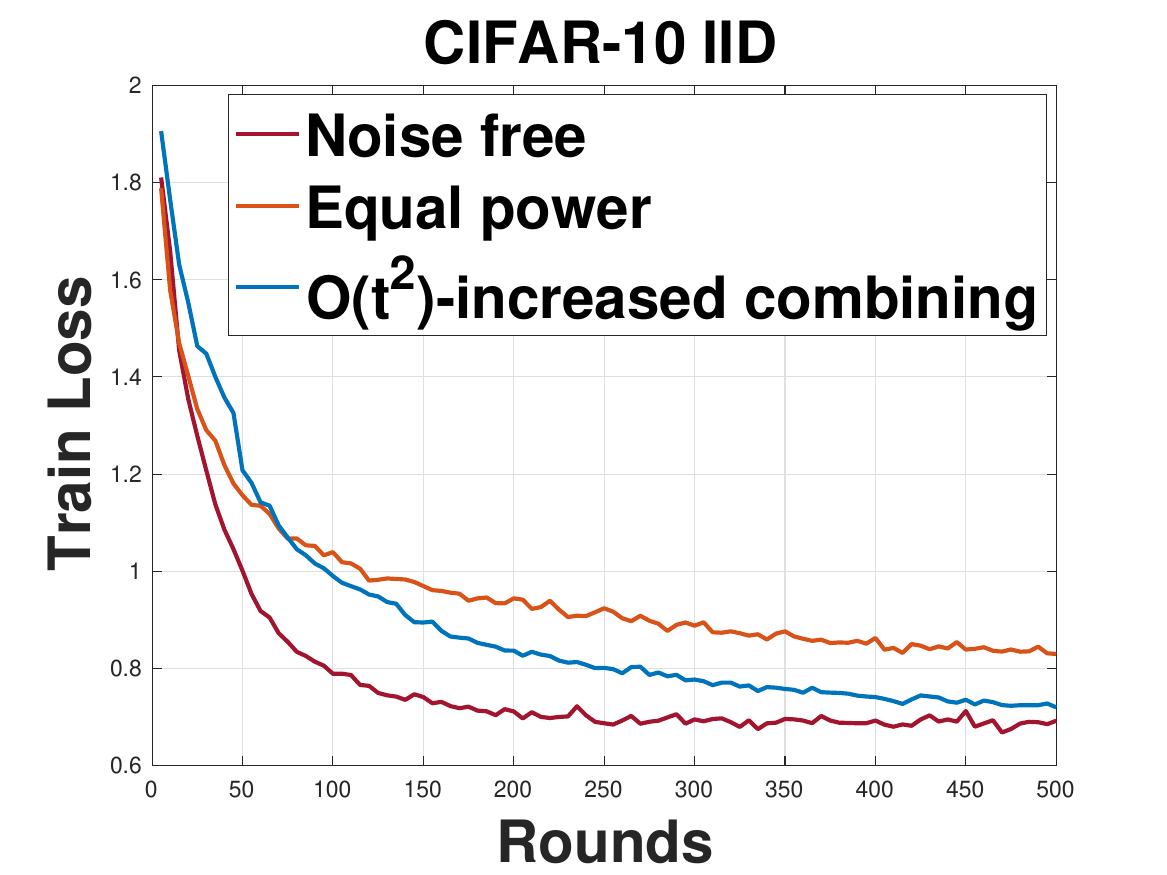}}
    \subfigure{
        \includegraphics[width=0.23\textwidth]{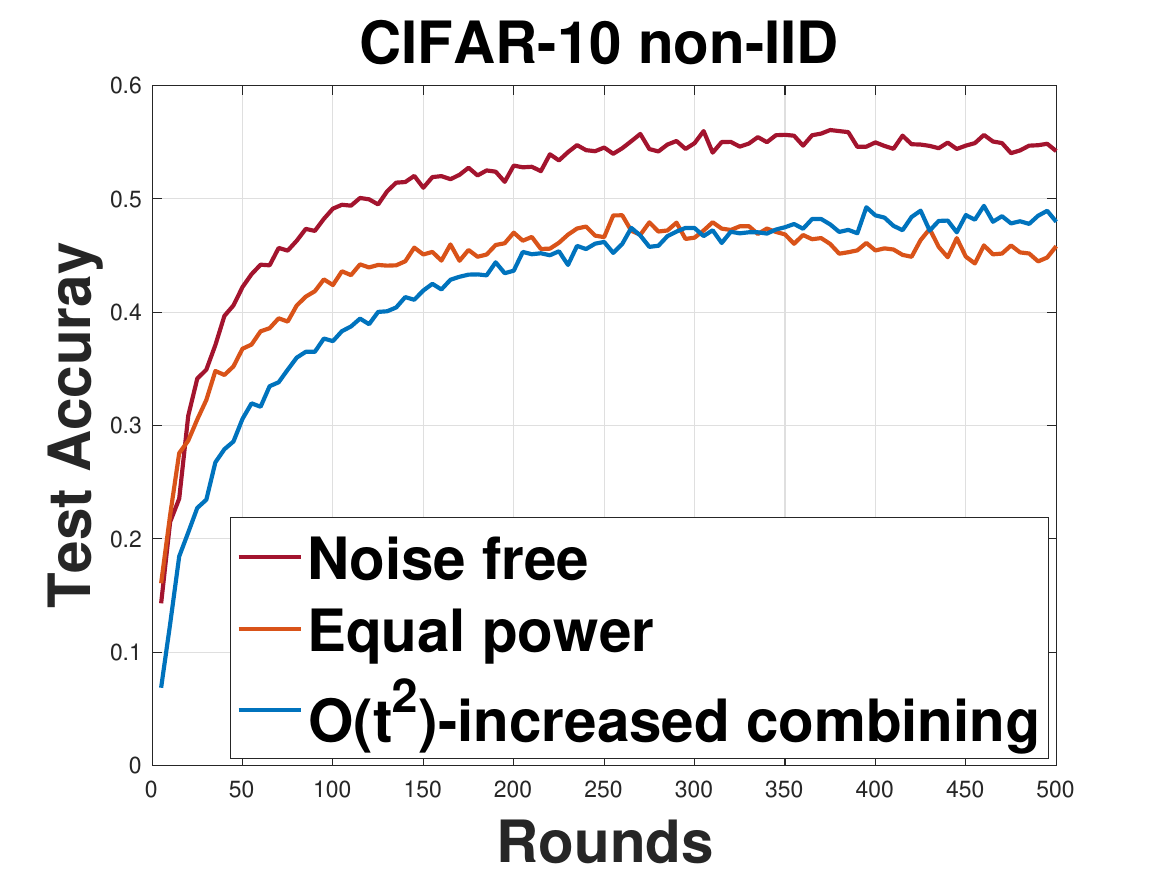}}
    \subfigure{
        \includegraphics[width=0.23\textwidth]{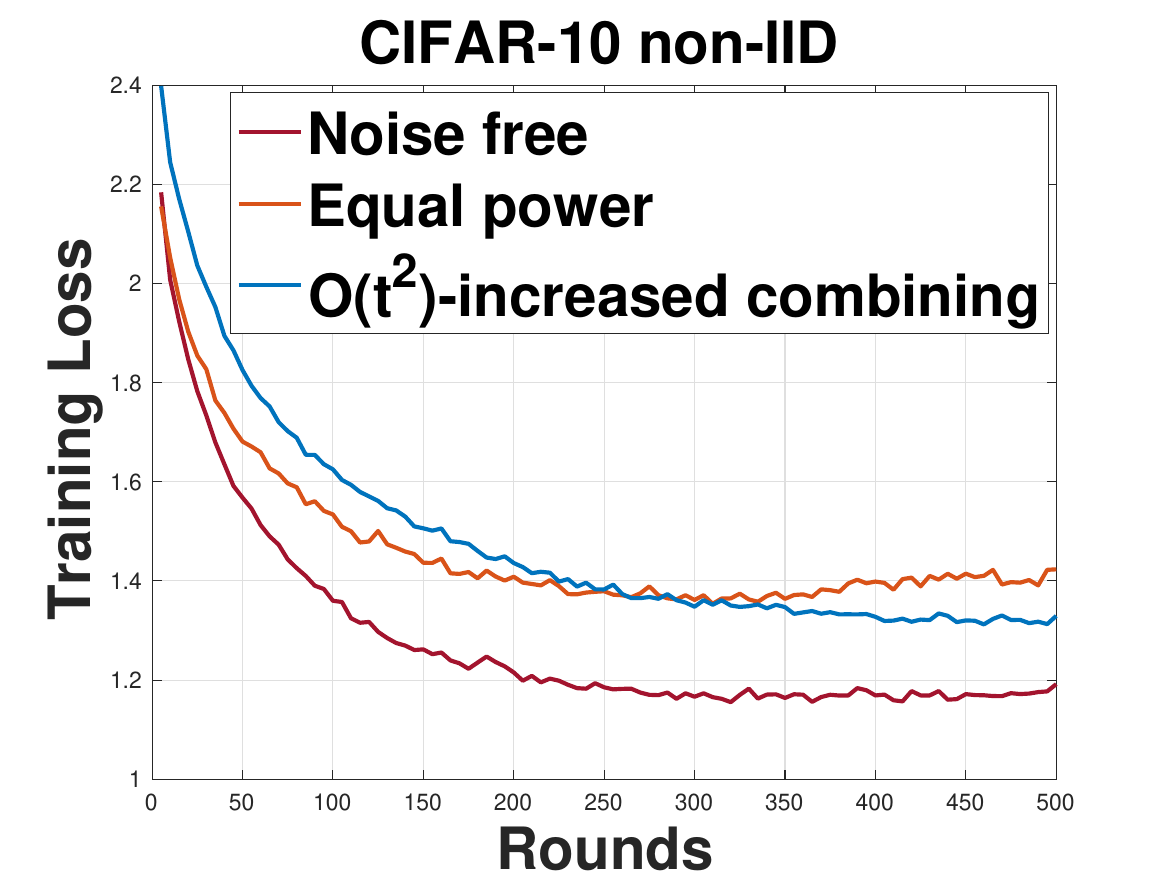}}
    \caption{Comparing the performance of receive diversity combining to the baselines with partial clients participation, model differential transmission, and both IID (left two) and non-IID (right two) FL on the CIFAR-10 dataset.}
    \label{fig:CIFAR10-Combine}
    \vspace{-0.15in}
\end{figure}

We next evaluate the performance of receive diversity combining. Due to space limitation, we only report the result for CIFAR-10, but similar conclusions hold for other tasks. Fig.~\ref{fig:CIFAR10-Combine} captures the test accuracies and training losses of receive diversity combining together with noise free and equal power allocation schemes. Although receive diversity combining is less flexible than the (continuous) transmit power control policy, we can see that it still outperforms the baseline method and approaches the noise-free ideal case. We notice that the training losses of receive diversity combining are larger than those of the equal power allocation scheme at the beginning stage of convergence, but as the diversity branches increase, the training losses eventually reduce and the model converges to a better global one. Particularly, receive diversity combining achieves $75.6\%$ and $47.8\%$ test accuracies for IID and non-IID data partitions, which is $3.9\%$ and $3.4\%$ better than the equal power allocation scheme. 


\section{Conclusion}\label{sec:conclusion}
In this paper, we have investigated federated learning over noisy channels, where a \textsc{FedAvg} pipeline with both uplink and downlink communication noises was studied. By theoretically analyzing the model training convergence, we have proved that the same $\mathcal{O}({1}/{T})$ convergence rate of \textsc{FedAvg} under perfect (noise-free) communications can be maintained if the uplink and downlink SNRs are controlled as $\mathcal{O}(t^2)$ over noisy channels for direct model transmission, and $\mathcal{O}(1)$ for model differential transmission. We have showcased two widely used communication methods -- transmit power control and receive diversity combining -- to implement these theoretical results. Extensive experimental results have corroborated the theoretical analysis and demonstrated the performance superiority of the advanced designs over baseline methods under the same total energy budget. 
Future research directions include relaxing the assumption of strongly convex loss functions to a broader class (e.g., convex, non-convex), and removing the fixed $T$ assumption to develop an \emph{any-time} version of the proposed design.



\appendices

\section{Proof of Theorem~\ref{thm:MTfull}}
\label{app:proof_thm1}

\subsection{Preliminaries}
\label{sec:notation}

With a slight abuse of notation, we change the timeline to be with respect to the overall SGD iteration time steps instead of the communication rounds, i.e., $$t=\underbrace{1, \cdots, E}_{\text{round 1}}, \underbrace{E+1, \cdots, 2E}_{\text{round 2}}, \cdots, \cdots, \underbrace{(T-1)E+1, \cdots, TE}_{\text{round $T$}}.$$ 
Note that the (noisy) global model $\vect{w}_{t}$ is only accessible at the clients for specific $t \in \mathcal{I}_E$, where $\mathcal{I}_E=\{nE ~|~n=1,2,\dots\}$, i.e., the time steps for communication.  The notations for $\eta_t$, $\sigma_t$ and $\zeta_t$ are similarly adjusted to this extended timeline, but their values remain constant inside the same round.


As mentioned in Section~\ref{sec:conv_MT_full}, the key technique in the proof is the \textbf{perturbed iterate framework}  in \cite{mania2017siam}. In particular, We first define the following variables

\zixiangTCCN{
\begin{equation*}
 \vect{u}_{t+1}^k = \frac{1}{N} \sum_{i \in [N]} \vect{v}_{t+1}^i, \qquad 
 \vect{p}_{t+1}^k = \vect{u}_{t+1}^k+\frac{1}{N}\sum_{i \in [N] }\vect{n}_{t+1}^i, \qquad
 \text{and~ } \vect{w}_{t+1}^k = \vect{p}_{t+1}^k+\vect{e}_{t+1}^k,
 \end{equation*}
 to summarize the aforementioned steps:
\begin{align*}
    \vect{v}_{t+1}^k & \triangleq \vect{w}_t^k - \eta_t \nabla F_k(\vect{w}_t^k, \xi_t^k); \\
    \vect{u}_{t+1}^k & \triangleq \begin{cases}
        \vect{v}_{t+1}^k & \text{if~} t+1 \notin \mathcal{I}_E, \\
        \frac{1}{N} \sum_{i \in [N]} \vect{v}_{t+1}^i & \text{if~} t+1 \in \mathcal{I}_E; 
    \end{cases} \\
     \vect{p}_{t+1}^k & \triangleq \begin{cases}
        \vect{v}_{t+1}^k & \text{if~} t+1 \notin \mathcal{I}_E, \\
        \vect{u}_{t+1}^k+\frac{1}{N}\sum_{i \in [N] }\vect{n}_{t+1}^i & \text{if~} t+1 \in \mathcal{I}_E. 
    \end{cases}\\
    \vect{w}_{t+1}^k & \triangleq \begin{cases}
         \vect{v}_{t+1}^k & \text{if~} t+1 \notin \mathcal{I}_E, \\
        \vect{p}_{t+1}^k+\vect{e}_{t+1}^k & \text{if~} t+1 \in \mathcal{I}_E. 
    \end{cases}
 \end{align*}

}
Then, we construct the following \textit{virtual sequences}:
\begin{equation}\label{eqn:vir_seq}
\avgvect{v}_t=\frac{1}{N}\sum_{k=1}^N \vect{v}_t^k, \qquad 
\avgvect{u}_{t} = \frac{1}{N}\sum_{k=1}^N \vect{u}_t^k, \qquad
\avgvect{p}_t=\frac{1}{N}\sum_{k=1}^N \vect{p}_t^k, \qquad 
\text{and~ }\avgvect{w}_{t} = \frac{1}{N}\sum_{k=1}^N \vect{w}_t^k.
\end{equation}
We also define $\avgvect{g}_{t} = \frac{1}{N}\sum_{k=1}^N \nabla F_k(\vect{w}_t^k)$ and $\vect{g}_{t} = \frac{1}{N}\sum_{k=1}^N \nabla F_k(\vect{w}_t^k, \xi_t^k)$ for convenience. Therefore, $\avgvect{v}_{t+1} = \avgvect{w}_t - \eta_t \vect{g}_t$ and $\expt \squab{\vect{g}_t} = \avgvect{g}_t$. 
After some manipulation, we can also write the specific formulations of these virtual sequences when $t+1 \in \mathcal{I}_E$ as follows: 
\begin{equation}
    \label{eqn:avg_pw}
    \avgvect{u}_{t+1} = \frac{1}{N} \sum_{i \in [N]} \vect{v}_{t+1}^i, \quad
    \avgvect{p}_{t+1}  = \avgvect{u}_{t+1} + \frac{1}{N}\sum_{i \in [N] }\vect{n}_{t+1}^i,  \quad
    \avgvect{w}_{t+1}  = \avgvect{p}_{t+1} + \frac{1}{N}\sum_{k=1}^N \vect{e}_{t+1}^k.
\end{equation}
Note that for $t+1 \notin \mathcal{I}_E$, all these virtual sequences are the same. 
In addition, the global model (at the server) $\avgvect{p}_{t+1}$ is meaningful only at $t+1 \in \mathcal{I}_E$. We emphasize that when $t+1 \in \mathcal{I}_E$, Eqns.~\eqref{eqn:avg_pw} and \eqref{eqn:glbMT} indicate that $\avgvect{p}_{t+1} = \vect{w}_{t+1}$. Thus it is sufficient to analyze the convergence of $\norm{\avgvect{p}_{t+1}-\vect{w}^*}^2$.


\subsection{Lemmas and proofs}
\label{sec:key_lemma1}

\begin{lemma}
\label{lemma:sgd}
    Let \zixiangTCCN{Assumption \ref{as:F}}  hold, $\eta_t$ is non-increasing, and $\eta_t \leq 2 \eta_{t+E}$ for all $t \geq 0$. If $\eta_t \leq {1}/(4L)$, we have
            $\expt \norm{\avgvect{v}_{t+1} - \vect{w}^*}^2 \leq (1-\eta_t \mu) \expt \norm{\avgvect{w}_t - \vect{w}^*}^2 + \eta_t^2 \left( \sum_{k=1}^{N}{\delta_k^2}/{N^2} + 6 L \Gamma + 8(E-1)^2H^2\right)$.
\end{lemma}
Lemma~\ref{lemma:sgd} establishes a bound for the one-step SGD. This result only concerns the local model update and is not impacted by the noisy communication. The derivation is similar to the technique in \cite{stich2018local}. 

\begin{lemma}
\label{lemma:quan}
    We have
    \begin{equation}\label{eqn:unbiased}
    \expt\squab{\avgvect{p}_{t+1}} = \avgvect{u}_{t+1}, \expt\norm{\avgvect{u}_{t+1} - \avgvect{p}_{t+1}}^2 =  \frac{d\sigma_{t+1}^2}{N^2}; \quad \expt\squab{\avgvect{w}_{t+1}} = \avgvect{p}_{t+1}, \expt\norm{\avgvect{w}_{t+1} - \avgvect{p}_{t+1}}^2 =  \frac{d\zeta_{t+1}^2}{N^2}
    \end{equation}
    for $t+1 \in \mathcal{I}_{E}$, where $\sigma_{t+1}^2\triangleq\sum_{k \in [N] }\sigma^2_{t+1,k}$ and $\zeta_{t+1}^2\triangleq\sum_{k \in [N] }\zeta^2_{t+1,k}$.
\end{lemma}

\begin{proof}
    For $t+1 \in \mathcal{I}_E$, we have
$  \expt \left[\avgvect{p}_{t+1} - \avgvect{u}_{t+1} \right] = \frac{1}{K}\sum_{k \in [N]} \expt [\vect{n}^k_{t+1}]=0$ and
$
\expt \norm{\avgvect{p}_{t+1} - \avgvect{u}_{t+1}}^2  = \frac{1}{N^2}\expt \norm{\sum_{k \in [N]} \vect{n}^k_{t+1}}^2= \frac{1}{N^2} \sum_{k \in [N]}\expt \norm{\vect{n}^k_{t+1}}^2 = \frac{d\sigma^2_{t+1}}{N^2}
$
from \eqref{eqn:avg_pw}, because $\{\vect{n}^k_{t+1}, \forall k\}$ are independent variables. 
Similarly, according to (\ref{eqn:avg_pw}), we have
$  \expt \left[\avgvect{w}_{t+1} - \avgvect{p}_{t+1} \right]
             = \frac{1}{N}\sum_{k \in [N]} \expt [\vect{e}^k_{t+1}]=0$ 
   and
    $
  \expt \norm{\avgvect{w}_{t+1} - \avgvect{p}_{t+1}}^2  
             = \frac{1}{N^2}\expt \norm{\sum_{k \in [N]} \vect{e}^k_{t+1}}^2 =
             \frac{1}{N^2} \sum_{k \in [N] }\expt \norm{\vect{e}^k_{t+1}}^2
             = \frac{\sum_{k\in [N]}d\zeta^2_{t+1,k}}{N^2}=\frac{d\zeta^2_{t+1}}{N^2}.
             $
\end{proof}

\subsection{Proof of Theorem}
\label{sec:compl_proof1}
We need to consider four cases for the analysis of the convergence of $\expt\norm{\avgvect{p}_{t+1}-\vect{w}^*}^2$.

1) If $t \notin \mathcal{I}_E$ and $t+1 \notin \mathcal{I}_E$, $\avgvect{w}_{t} = \avgvect{p}_{t}$ and $\avgvect{v}_{t+1} = \avgvect{p}_{t+1}$. Using Lemma \ref{lemma:sgd}, we have:
\begin{equation*}
    \expt  \norm{\avgvect{p}_{t+1} - \vect{w}^*}^2 = \expt \norm{\avgvect{v}_{t+1} - \vect{w}^*}^2 \leq (1-\eta_t \mu) \expt \norm{\avgvect{p}_t - \vect{w}^*}^2 + \eta_t^2 \squab{\sum_{k=1}^{N}\frac{\delta_k^2}{N^2} + 6L \Gamma + 8(E-1)^2 H^2}.
\end{equation*}

2) If $t \in \mathcal{I}_E$ and $t+1 \notin \mathcal{I}_E$, we still have $\avgvect{v}_{t+1} = \avgvect{p}_{t+1}$. With $\avgvect{w}_{t} = \avgvect{p}_{t} + \frac{1}{N}\sum_{k=1}^N \vect{e}_{t}^k$, we have:
\begin{equation*}
    \norm{\avgvect{w}_t - \vect{w}^*}^2 = \norm{\avgvect{w}_t - \avgvect{p}_t + \avgvect{p}_t - \vect{w}^*}^2 = \norm{\avgvect{p}_t - \vect{w}^*}^2 + \underbrace{\norm{\avgvect{w}_t - \avgvect{p}_t}^2}_{A_1} + \underbrace{2\dotp{\avgvect{w}_t - \avgvect{p}_t}{\avgvect{p}_t - \vect{w}^*}}_{A_2}.
\end{equation*}
 We first note that the expectation of $A_2$ over the noise randomness is zero since we have $\expt \squab{\avgvect{w}_{t} - \avgvect{p}_{t}} = \vect{0}$ (from \eqref{eqn:unbiased}). Second, the expectation of $A_1$ can be bounded using Lemma \ref{lemma:quan}. We then have
\begin{align}
    & \expt  \norm{\avgvect{p}_{t+1} - \vect{w}^*}^2 = \expt \norm{\avgvect{v}_{t+1} - \vect{w}^*}^2 \nonumber \\
    & \leq (1-\eta_t \mu) \expt \norm{\avgvect{p}_t - \vect{w}^*}^2 +(1-\eta_t \mu) \expt\norm{\avgvect{w}_t - \avgvect{p}_t}^2  + \eta_t^2 \squab{\sum_{k=1}^{N}\frac{\delta_k^2}{N^2} + 6L \Gamma + 8(E-1)^2 H^2 } \nonumber \\
    & \leq (1-\eta_t \mu) \expt \norm{\avgvect{p}_t - \vect{w}^*}^2 +(1-\eta_t \mu)\frac{d\zeta_{t}^2}{N^2} + \eta_t^2 \squab{\sum_{k=1}^{N}\frac{\delta_k^2}{N^2} + 6L \Gamma + 8(E-1)^2 H^2 }. \label{eqn:case2}
\end{align}
        
3) If $t \notin \mathcal{I}_E$ and $t+1 \in \mathcal{I}_E$, then we still have $\avgvect{w}_{t} = \avgvect{p}_{t}$. For $t+1$, we need to evaluate the convergence of $\expt\norm{\avgvect{p}_{t+1}-\vect{w}^*}^2$. We have
 \begin{equation} \label{eqn:depart0}
    \begin{split}
        \norm{\avgvect{p}_{t+1} - \vect{w}^*}^2 & = \norm{\avgvect{p}_{t+1} - \avgvect{u}_{t+1} + \avgvect{u}_{t+1}- \vect{w}^*}^2 \\
        & = \underbrace{\norm{\avgvect{p}_{t+1} - \avgvect{u}_{t+1}}^2}_{B_1} + \underbrace{\norm{\avgvect{u}_{t+1}- \vect{w}^*}^2}_{B_2} + \underbrace{2\dotp{\avgvect{p}_{t+1} - \avgvect{u}_{t+1}}{\avgvect{u}_{t+1}- \vect{w}^*}}_{B_3}.
    \end{split}
\end{equation}
We first note that the expectation of $B_3$ over the noise is zero since we have $\expt \squab{\avgvect{u}_{t+1} - \avgvect{p}_{t+1}} = \vect{0}$ (from \eqref{eqn:unbiased}). Second,  the expectation of $B_1$ can be bounded using Lemma \ref{lemma:quan}. Noticing that $\avgvect{u}_{t+1}=\avgvect{v}_{t+1}$ for $B_2$ and applying Lemma~\ref{lemma:sgd}, we have:
   \begin{align}
        & \expt\norm{\avgvect{p}_{t+1}-\vect{w}^*}^2  \leq \expt\norm{\avgvect{v}_{t+1}- \vect{w}^*}^2 + \frac{d\sigma_{t+1}^2}{K^2} \nonumber  \\
        & \leq (1-\eta_t \mu) \expt \norm{\avgvect{p}_t - \vect{w}^*}^2 + \frac{d\sigma_{t+1}^2}{N^2} + \eta_t^2 \squab{\sum_{k=1}^N \frac{\delta_k^2}{N^2} + 6L \Gamma + 8(E-1)^2 H^2}.  \label{eqn:case3}
   \end{align}

4) If $t \in \mathcal{I}_E$ and $t+1 \in \mathcal{I}_E$, $\avgvect{v}_{t+1}\neq \avgvect{p}_{t+1}$ and $\avgvect{w}_{t}\neq \avgvect{p}_{t}$. (Note that this is possible only for $E=1$.) Combining the results from the previous two cases, we have 
\begin{equation}
\label{eqn:case4}
    \expt\norm{\avgvect{p}_{t+1}-\vect{w}^*}^2   \leq (1-\eta_t \mu) \expt \norm{\avgvect{p}_t - \vect{w}^*}^2 + (1-\eta_t \mu)\frac{d\zeta_{t}^2}{N^2} + \frac{d\sigma_{t+1}^2}{N^2} + \eta_t^2 \squab{\sum_{k=1}^N\frac{\delta_k^2}{N^2} + 6L \Gamma + 8(E-1)^2 H^2}.
\end{equation}

Finally, we have that inequality~\eqref{eqn:case4} holds for all four cases. Denote $\Delta_t = \expt \norm{\avgvect{p}_{t}- \vect{w}^*}^2$. If we set the effective noise power $\sigma^2_{t+1}$ and $\zeta^2_{t}$ such that  $\sigma^2_{t+1} \leq N^2 \eta^2_t$ and $\zeta^2_{t} \leq N^2 \frac{\eta_t^2}{1-\eta_t\mu}$, we always have
      $  \Delta_{t+1} \leq (1-\eta_t\mu) \Delta_{t} + \eta_t^2 D$,  
    where
      $  D = \sum_{k=1}^N \frac{\delta_k^2}{N^2} + 6L \Gamma + 8(E-1)^2 H^2 +  2d$. 
    We decay the learning rate as $\eta_t=\frac{\beta}{t+\gamma}$ for some $\beta \geq \frac{1}{\mu}$ and $\gamma \geq 0$ such that $\eta_1 \leq \min\{\frac{1}{\mu}, \frac{1}{4L}\}=\frac{1}{4L}$ and $\eta_t \leq 2 \eta_{t+E}$. Now we prove that $\Delta_{t} \leq \frac{v}{\gamma +t}$ where $v = \max \{ \frac{\beta^2 D}{\beta \mu -1}, (\gamma+1) \Delta_{0} \}$ by induction. First, the definition of $v$ ensures that it holds for $t=0$. Assume the conclusion holds for some $t>0$. It then follows that
    \begin{equation*}
        \begin{split}
            \Delta_{t+1} &\leq (1 - \eta \mu) \Delta_t + \eta_t^2 D = \left ( 1-\frac{\beta\mu}{t+\gamma} \right ) \frac{v}{t+\gamma} + \frac{\beta^2D}{(t+\gamma)^2} \\
            &= \frac{t+\gamma-1}{(t+\gamma)^2}v + \squab{\frac{\beta^2D}{(t+\gamma)^2} - \frac{\mu \beta -1}{(t+\gamma)^2}v}  \leq \frac{v}{t+\gamma+1}.  
        \end{split}
    \end{equation*}
    Then by the strong convexity of $F(\cdot)$,
$        \expt \squab{F(\avgvect{w}_t)} - F^* \leq \frac{L}{2} \Delta_t \leq \frac{L}{2}\frac{v}{\gamma+t}$. 
    Specially, if we choose $\beta = \frac{2}{\mu}$, $\gamma = \max\{ 8\frac{L}{\mu}-1, E\}$ and denote $\phi = \frac{L}{\mu}$, then $\eta_t = \frac{2}{\mu}\frac{1}{\gamma+t}$. Using $\max\{a, b\} \leq a+b$, we have
     $       v  \leq \frac{4D}{\mu^2} + (\gamma+1)\Delta_0 \leq  \frac{4D}{\mu^2} + \left (8 \phi + E \right)\norm{\vect{w}_0 - \vect{w}^*}^2$.         
    Therefore,
        $\Delta_{t} \leq \frac{v}{\gamma +t} = \frac{1}{\gamma+t} \squab{ \frac{4D}{\mu^2} + \left( 8\phi + E \right) \norm{\vect{w}_0 - \vect{w}^*}^2}$. 
    Setting $t=T$ concludes the proof.

\section{Proof of Theorem~\ref{thm:MTpart}}
\label{app:proof_thm_MTpart}

The additional difficulty in proving Theorem~\ref{thm:MTpart} comes from partial clients participation. The approach we take is to study a ``virtual'' FL process where \textit{all clients} receive the noisy downlink broadcast of the latest global model, and they all participate in the subsequent local model update phase. However, only the selected clients in $\mathcal{S}_{t+1}$ upload their updated local model to the server via the noisy uplink channel. It is clear that this ``virtual'' FL is equivalent to the original process in terms of the convergence -- clients that are not selected do not contribute to the global model aggregation. This seemingly redundant process, however, circumvents the difficulty due to partial clients participation as can be seen in the analysis.

Before presenting the proof, we first elaborate on some necessary changes of notation. The notation defined in \zixiangTCCN{Appendix~\ref{sec:notation}} can be largely reused, with the notable distinction that now we have to separate the cases for $K$ and for $N$. For $t+1 \in \mathcal{I}_E$, the variables of $\vect{u}_{t+1}^k$ and $\vect{p}_{t+1}^k$ are now defined as: $\vect{u}_{t+1}^k = \frac{1}{K} \sum_{i \in S_{t}} \vect{v}_{t+1}^i$ and $\vect{p}_{t+1}^k=\vect{u}_{t+1}^k+\frac{1}{K}\sum_{i \in \mathcal{S}_{t} }\vect{n}_{t+1}^i$. 
Note that Lemma~\ref{lemma:quan} still holds with the following update: $\expt\norm{\avgvect{u}_{t+1} - \avgvect{p}_{t+1}}^2 =  \frac{d\bar{\sigma}_{t+1}^2}{K}$, $\expt\norm{\avgvect{w}_{t+1} - \avgvect{p}_{t+1}}^2 =  \frac{d\bar{\zeta}_{t+1}^2}{N}$. 
\zixiangTCCN{In addition, we need the following lemma, whose proof is available on a complete online version of this paper\cite{wei2021federated}.}
\begin{lemma}
\label{lemma:sample}
    Let Assumption \ref{as:F}-4) hold. With $\eta_t \leq 2\eta_{t+E}$ for all $t \geq 0$ and $\forall t+1 \in \mathcal{I}_E$, we have $\expt \squab{\avgvect{u}_{t+1}} = \avgvect{v}_{t+1}$ and $\expt \norm{\avgvect{v}_{t+1} - \avgvect{u}_{t+1}}^2 \leq \frac{N-K}{N-1} \frac{4}{K} \eta_t^2 E^2 H^2$.
\end{lemma}
\zixiangTCCN{
\begin{proof}
    Let $\mathcal{S}_{t+1}$ denote the set of chosen indexes. Note that the number of possible $\mathcal{S}_{t+1}$ is $C_N^K$ and we denote the $l$th possible result as $\mathcal{S}_{t+1}^l = \{i_1^l, \dots, i_K^l\}$, where $l=1,\dots,C_N^K$. Therefore,
    \begin{equation*}
       \sum_{j=1}^{C_N^K} \sum_{k=1}^K \vect{v}_{t+1}^{i_k^l} = \frac{K \cdot C_N^K}{N} \sum_{i=1}^N {\vect{v}_{t+1}^k} 
    = C_{N-1}^{K-1}\sum_{i=1}^N {\vect{v}_{t+1}^k}.
    \end{equation*}
    Since when $t+1 \in \mathcal{I}_E$, $$\vect{u}_{t+1}^k = \frac{1}{K} \sum_{k \in S{t+1}} \vect{v}_{t+1}^k$$ for all $k$, we have
    \begin{equation*}
   \avgvect{u}_{t+1} = \sum_{k=1}^N \vect{u}_{t+1}^k = \frac{1}{K} \sum_{k \in S_{t+1}} \vect{v}_{t+1}^k.
    \end{equation*}
    Then
    \begin{equation*}
    \begin{split}
        & \expt_{\mathcal{S}_{t}}\squab{\avgvect{u}_{t+1}}  = \sum_{l=1}^{C_N^K} \mathbb{P}\left (\mathcal{S}_{t+1} = \mathcal{S}_{t+1}^l \right )  \frac{1}{K} \sum_{k \in S_{t+1}^l} \vect{v}_{t+1}^k = \frac{1}{C_N^K} \frac{1}{K} \sum_{j=1}^{C_N^K} \sum_{k=1}^K \vect{v}_{t+1}^{i_k^l} = \frac{C_{N-1}^{K-1}}{C_N^K} \frac{1}{K}\sum_{k=1}^N {\vect{v}_{t+1}^k} \\
        & = \frac{1}{N} \sum_{k=1}^N {\vect{v}_{t+1}^k} = \avgvect{v}_{t+1}.
    \end{split}
    \end{equation*}
    As for the variance, we have \cite{li2019convergence}
    \begin{equation}
    \label{eqn:uv_st}
        \begin{split}
           &\expt_{\mathcal{S}_{t}}  \norm{\avgvect{u}_{t+1}-\avgvect{v}_{t+1}}^2  = \expt_{\mathcal{S}_{t}} \norm{\frac{1}{K} \sum_{i \in S_{t+1}} \vect{v}_{t+1}^{i}-\avgvect{v}_{t+1}}  =  \frac{1}{K^2} \expt_{\mathcal{S}_{t}} \norm{\sum_{i=1}^{N} \mathbb{I}\left\{i \in S_{t}\right\}\left(\vect{v}_{t+1}^{i}-\avgvect{v}_{t+1}\right)}^2 \\
           & =  \frac{1}{K^{2}}\left[\sum_{i \in[N]} \mathbb{P}\left(i \in S_{t+1}\right) \norm{\vect{v}_{t+1}^{i}-\avgvect{v}_{t+1}}^2 \right. \left.+ \sum_{i \neq j} \mathbb{P}\left(i, j \in S_{t+1}\right) \dotp{\vect{v}_{t+1}^{i}-\avgvect{v}_{t+1}}{\vect{v}_{t+1}^{j}-\avgvect{v}_{t+1}} \right ] \\
            &= \frac{1}{K N} \sum_{i=1}^{N} \norm{\vect{v}_{t+1}^{i}-\avgvect{v}_{t+1}}^2  + \sum_{i \neq j} \frac{K-1}{K N(N-1)} \dotp{\vect{v}_{t+1}^{i}-\avgvect{v}_{t+1}}{\vect{v}_{t+1}^{j}-\avgvect{v}_{t+1}} \\
            & =  \frac{1-\frac{K}{N}}{K(N-1)} \sum_{i=1}^{N}\norm{\vect{v}_{t+1}^{i}-\avgvect{v}_{t+1}}^2
        \end{split}
    \end{equation}
    where we use the following results: $$\mathbb{P}\left(i \in S_{t+1}\right) = \frac{K}{N}$$ and $$\mathbb{P}\left(i,j \in S_{t+1}\right) = \frac{K(K-1)}{N(N-1)}$$ for all $i \neq j$, and $$\sum_{i \in [N]} \norm{\vect{v}_{t+1}^i -\avgvect{v}_{t+1}}^2 + \sum_{i \neq j} \dotp{\vect{v}_{t+1}^{i}-\avgvect{v}_{t+1}}{\vect{v}_{t+1}^{j}-\avgvect{v}_{t+1}} = 0.$$    
    Since $t+1 \in \mathcal{I}_E$, we know that $t_0=t-E+1 \in \mathcal{I}_E$ is the communication time, implying that $\{\vect{u}_{t_0}^k \}_{k=1}^N$ are identical. Then
    \begin{equation*}
        \begin{split}
            & \sum_{i=1}^N \norm{\vect{v}_{t+1}^i - \avgvect{v}_{t+1}}^2  = \sum_{i=1}^N \norm{(\vect{v}_{t+1}^i - \avgvect{u}_{t_0}) - (\avgvect{v}_{t+1} - \avgvect{u}_{t_0})}^2 \\
            & = \sum_{i=1}^N \norm{\vect{v}_{t+1}^i - \avgvect{u}_{t_0}}^2 - 2 \dotp{\sum_{i=1}^N \vect{v}_{t+1}^i - \avgvect{u}_{t_0}}{\avgvect{v}_{t+1}-\avgvect{u}_{t_0}} + \sum_{i=1}^N \norm{\avgvect{v}_{t+1} - \avgvect{u}_{t_0}}^2 \\
            & = \sum_{i=1}^N \norm{\vect{v}_{t+1}^i - \avgvect{u}_{t_0}}^2 - \sum_{i=1}^N \norm{\avgvect{v}_{t+1} - \avgvect{u}_{t_0}}^2 \leq \sum_{i=1}^N \norm{\vect{v}_{t+1}^i - \avgvect{u}_{t_0}}^2
        \end{split}
    \end{equation*}
    Taking expectation over the randomness of stochastic gradient on Eqn.~\eqref{eqn:uv_st}, we have
    \begin{equation*}
        \begin{split}
            \expt & \squab{\frac{1}{K(N-1)}\left(1-\frac{K}{N}\right) \sum_{k=1}^N \norm{\vect{v}_{t+1}^i - \avgvect{v}_{t+1}}^2}  \leq \frac{N-K}{K(N-1)} \frac{1}{N} \sum_{k=1}^N \expt \norm{\vect{v}_{t+1}^i - \avgvect{u}_{t_0}}^2 \\
            & \leq \frac{N-K}{K(N-1)}\frac{1}{N} \sum_{k=1}^N E \sum_{i=t_0}^t \expt \norm{\eta_i \nabla F_k{(\vect{u}_i^k, \xi_i^k)}}^2   \leq \frac{N-K}{K(N-1)} E^2 \eta_{t_0}^2 H^2  \leq \frac{N-K}{N-1}\frac{4}{K} E^2 \eta_t^2 H^2
        \end{split}
    \end{equation*}
    where in the last line is because $\eta_t$ is non-increasing and $\eta_{t_0} \leq 2\eta_t$.
\end{proof}}
We can now similarly analyze the four cases as in Section~\ref{sec:compl_proof1}. Cases 1) and 2) remain the same as before. For Case 3) we need to consider $t \notin \mathcal{I}_E$ and $t+1 \in \mathcal{I}_E$. Note that \eqref{eqn:depart0} still holds, but we need to re-evaluate the expectation of $B_2$ because of partial clients participation. We have:
    \begin{equation} \label{eqn:depart1}
        \begin{split}
            \norm{\avgvect{u}_{t+1} - \vect{w}^*}^2 & = \norm{\avgvect{u}_{t+1} - \avgvect{v}_{t+1} + \avgvect{v}_{t+1}- \vect{w}^*}^2 \\
            & = \underbrace{\norm{\avgvect{u}_{t+1} - \avgvect{v}_{t+1}}^2}_{C_1} + \underbrace{\norm{\avgvect{v}_{t+1}- \vect{w}^*}^2}_{C_2} + \underbrace{2\dotp{\avgvect{u}_{t+1} - \avgvect{v}_{t+1}}{\avgvect{v}_{t+1}- \vect{w}^*}}_{C_3}.
        \end{split}
    \end{equation}
    When the expectation is taken over the random clients sampling, the expectation of $C_3$ is zero since we have $\expt \squab{\avgvect{u}_{t+1} - \avgvect{v}_{t+1}} = \vect{0}$. 
    The expectation of $C_1$ can be bounded using Lemma \ref{lemma:sample}. Therefore we have
       $ \expt\norm{\avgvect{p}_{t+1}-\vect{w}^*}^2 \leq \expt\norm{\avgvect{v}_{t+1}- \vect{w}^*}^2 + \frac{d \bar{\sigma}_{t+1}^2}{K} + \frac{N-K}{N-1} \frac{4}{K} \eta_t^2 E^2 H^2$. 
   Using Lemma~\ref{lemma:sgd} and the new definition of $D$ in Theorem~\ref{thm:MTpart}, we have
   $
   \expt\norm{\avgvect{p}_{t+1}-\vect{w}^*}^2  \leq \expt\norm{\avgvect{v}_{t+1}- \vect{w}^*}^2 + \frac{d \bar{\sigma}_{t+1}^2}{K} + \frac{4\eta_t^2 E^2 H^2(N-K)}{K(N-1)} \leq (1-\eta_t \mu) \expt \norm{\avgvect{p}_t - \vect{w}^*}^2 + \frac{d \bar{\sigma}_{t+1}^2}{K} + \eta_t^2 (D-2d).
   $
Case 4) can be similarly updated based on the new result in Case 3). Finally, 
we have that
\begin{equation}
\label{eqn:case4_new}
\expt\norm{\avgvect{p}_{t+1}-\vect{w}^*}^2  \leq (1-\eta_t \mu) \expt \norm{\avgvect{p}_t - \vect{w}^*}^2 + (1-\eta_t \mu)\frac{d \bar{\zeta}_t^2}{N} + \frac{d \bar{\sigma}_{t+1}^2}{K} + \eta_t^2  (D-2d)
\end{equation}
holds for all cases. If we set $\bar{\sigma}_{t+1}^2$ and $\bar{\zeta}_{t}^2$ such that $\bar{\sigma}_{t+1}^2 \leq K \eta^2_t$ and $\bar{\zeta}_{t}^2 \leq N \frac{\eta_t^2}{1-\eta_t\mu}$, the remaining proof follows the same way as in \zixiangTCCN{Appendix~\ref{sec:compl_proof1}}.

\section{Proof of Theorem~\ref{thm:MDTpart}}
\label{app:proof_thmMDTpart}

For model differential transmission (MDT), if $t + 1 \in \mathcal{I}_E$, the global aggregation is given in \eqref{eqn:glbMDT}. 
Similar to Appendix \ref{app:proof_thm1} and \ref{app:proof_thm_MTpart}, we expand the timeline to be with respect to the overall SGD iteration time steps, and define the following variables to facilitate the proof.
$\vect{v}_{t+1}^k \triangleq \vect{w}_t^k - \eta_t \nabla F_k(\vect{w}_t^k, \xi_t^k)$, $\vect{d}_{t+1}^k \triangleq \vect{v}_{t+1}^k - \vect{w}^k_{t+1-E}$. Furthermore, when $t+1 \notin \mathcal{I}_E$ we define $\vect{u}_{t+1}^k = \vect{p}_{t+1}^k = \vect{w}_{t+1}^k \triangleq \vect{v}_{t+1}^k$, and when $t+1 \in \mathcal{I}_E$ we define $\vect{u}_{t+1}^k \triangleq \frac{1}{K} \sum_{i \in S_{t}} \vect{v}_{t+1}^i$, $\vect{p}_{t+1}^k \triangleq \vect{w}_{t+1-E}+\frac{1}{K}\sum_{i \in \mathcal{S}_{t} }[\vect{d}_{t+1}^i + \vect{n}_{t+1}^i]$, and $\vect{w}_{t+1}^k \triangleq \vect{p}_{t+1}^k+\vect{e}_{t+1}^k$.  
The virtual sequences $\avgvect{v}_t$, $\avgvect{u}_t$, $\avgvect{p}_t$ and $\avgvect{w}_{t}$ remain the same as \eqref{eqn:vir_seq}. $\avgvect{g}_{t}$ and $\vect{g}_{t}$ are also similarly defined.  Note that the global model at the server is the same as $\avgvect{p}_t$, i.e., $\vect{w}_{t+1} = \avgvect{p}_{t+1}$. 


We first establish the follow in lemma, which is instrumental in the proof of Theorem~\ref{thm:MDTpart}.
\begin{lemma}
\label{lemma:unbiased2}
Let Assumption \ref{as:F}-4) hold. Assume that $\eta_t \leq 2\eta_{t+E}$ for all $t \geq 0$, and further assume that the uplink communication adopts a constant SNR control policy: $\ssf{SNR}_{t,k}^\text{S,MDT} = \nu$. Then, for $t+1 \in \mathcal{I}_E$, we have: 
       $ \expt\squab{\avgvect{p}_{t+1}} = \avgvect{u}_{t+1}$ and $\expt\norm{\avgvect{u}_{t+1} - \avgvect{p}_{t+1}}^2 \leq  \left( 1+\frac{1}{\nu} \right) \frac{d}{K} \bar{\zeta}_{t+1-E}^2 + \frac{4E^2}{K\nu} \eta_{t}^2 H^2$.
\end{lemma}

\begin{proof}
Note that if  $t+1 \in \mathcal{I}_E$, so does  $t+1-E$. Insert $\vect{d}_{t+1}^k = \vect{v}_{t+1}^k - \vect{w}^k_{t+1-E}$ into $\vect{p}_{t+1}^k$, we have
 $\expt\squab{\avgvect{p}_{t+1}} = \avgvect{u}_{t+1} + \frac{1}{K}\expt\squab{\sum_{k\in\mathcal{S}_{t}}\vect{n}_{t+1}^k} - \frac{1}{K}\expt\squab{\sum_{k\in\mathcal{S}_{t}}\vect{e}_{t+1-E}^k} =  \expt\squab{ \avgvect{u}_{t+1} }$.
As for the variance, we have
\begin{equation}
\label{eqn:app3_1}
    \expt\norm{\avgvect{u}_{t+1} - \avgvect{p}_{t+1}}^2 = \frac{1}{K^2} \expt \norm{ \sum_{k \in \mathcal{S}_{t}} \vect{n}^k_{t+1}}^2 + \frac{1}{K^2} \expt \norm{ \sum_{k \in \mathcal{S}_{t}}\vect{e}^k_{t+1-E}}^2 =  \frac{1}{K^2 \nu} \expt\norm{\sum_{k\in\mathcal{S}_{t}}\vect{d}_{t+1}^k}^2  + \frac{d \bar{\zeta}_{t+1-E}^2 }{K}
\end{equation}
where the last equality comes from the constant uplink SNR control, \eqref{eqn:MDTulSNR}, and the assumption that each client has the same downlink noise power $\bar{\zeta}_{t}^2$, $\forall k \in [N]$. 
We further have
\begin{equation}
\begin{split}
    &\expt\norm{\sum_{k\in\mathcal{S}_{t}}\vect{d}_{t+1}^k}^2 
      = \expt\norm{\sum_{k\in\mathcal{S}_{t}}(\vect{v}_{t+1}^k - \vect{ w}_{t+1-E} )}^2 + dK\bar{\zeta}_{t+1-E}^2  \\
     & \leq \expt_{\mathcal{S}_t} \left[ \sum_{k\in\mathcal{S}_{t}}\expt_{\text{SG}} \norm{\sum_{\tau=t+1-E}^t\eta_{\tau} \nabla F_k(\vect{w}_{\tau}^k, \xi_{\tau}^k)}^2 \right] + dK\bar{\zeta}_{t+1-E}^2 
    \leq {4E^2} K\eta_{t}^2 H^2 + dK\bar{\zeta}_{t+1-E}^2  \label{eqn:app3_2}
\end{split}
\end{equation}
using the Cauchy-Schwarz inequality, Assumption \ref{as:F}-4), and $\eta_{t+1-E}<\eta_{t-E}\leq2\eta_t$. Plugging \eqref{eqn:app3_2} back to \eqref{eqn:app3_1} gives
$$\expt\norm{\avgvect{u}_{t+1} - \avgvect{p}_{t+1}}^2 = \frac{1}{K^2 \nu} \expt\norm{\sum_{k\in\mathcal{S}_{t+1}}\vect{d}_{t+1}^k}^2  + \frac{d \bar{\zeta}_{t+1-E}^2 }{K} \leq \left( 1+\frac{1}{\nu} \right) \frac{d}{K} \bar{\zeta}_{t+1-E}^2 + \frac{4E^2}{K\nu} \eta_{t}^2 H^2,$$ 
which completes the proof.
\end{proof}

We are now ready to present the proof of Theorem~\ref{thm:MDTpart}, which is similar to that of Theorem \ref{thm:MTpart}. In particular, the analysis of four cases in Section~\ref{app:proof_thm_MTpart} still hold, with the only change that \eqref{eqn:case4_new} is updated to \eqref{eqn:inIE2} below using Lemma \ref{lemma:unbiased2} and the new definition of $D$ in Theorem~\ref{thm:MDTpart}.
\begin{equation}
    \label{eqn:inIE2}
    \expt\norm{\avgvect{p}_{t+1}-\vect{w}^*}^2  \leq (1-\eta_t \mu) \expt \norm{\avgvect{p}_t - \vect{w}^*}^2 + (1-\eta_t \mu)\frac{d }{N}{\bar{\zeta}}_t^2  + \left( 1+\frac{1}{\nu} \right) \frac{d}{K} \bar{\zeta}_{t}^2 + \eta_t^2 (D-d).
\end{equation}
We note that the  constant uplink SNR control is already used in Lemma \ref{lemma:unbiased2} and \eqref{eqn:inIE2}. Then, by the definition of $\Delta_t = \expt \norm{\avgvect{p}_t - \vect{w}^*}^2$ and controlling the downlink SNR such that $\bar{\zeta}^2_{t} \leq \frac{NK \eta^2_t}{(1-\eta_t \mu)K +  \left( 1+\frac{1}{\nu} \right) N}$, we have 
    $\Delta_{t+1} \leq (1 - \eta_t \mu) \Delta_t + \eta_t^2 D $.
The remaining proof follows using the same induction method. 

\bibliographystyle{IEEEtran}
\bibliography{FedLearn,wireless,Shen}
\end{document}